%% file: main.tex
\documentclass[journal,letterpaper,final,twocolumn]{IEEEtran} 
\usepackage[cmex10]{amsmath}
\usepackage{amssymb}
\usepackage[mathscr]{eucal}  
\usepackage{url}

\usepackage{stfloats}
\usepackage{pstricks}
\usepackage{graphicx}

\newcommand{\defn}{\stackrel{\Delta}{=}}
\newcommand{\bigoh}{O}

\newtheorem{Proposition}{Proposition}
\newtheorem{Corollary}{Corollary}
\newtheorem{Theorem}{Theorem}
\newtheorem{Lemma}{Lemma}

\begin{document}

\title{Channel polarization: A method for constructing capacity-achieving codes for symmetric binary-input memoryless channels}
\author{Erdal~Ar{\i}kan,~\IEEEmembership{Senior Member,~IEEE}%
\thanks{E. Ar{\i}kan is with the Department
of Electrical-Electronics Engineering, Bilkent University, Ankara,
06800, Turkey (e-mail: arikan@ee.bilkent.edu.tr)}%
\thanks{This work was supported in part by The Scientific and Technological Research
Council of Turkey (T¨ UBÿITAK) under Project 107E216 and in part
by the European Commission FP7 Network of Excellence NEWCOM++ under contract 216715.}
}
\maketitle

\begin{abstract}
A method is proposed, called channel polarization, to construct code sequences that achieve the symmetric capacity $I(W)$ of any given binary-input discrete memoryless channel (B-DMC) $W$. The symmetric capacity is the highest rate achievable subject to using the input letters of the channel with equal probability.
Channel polarization refers to the fact that it is possible to synthesize, out of $N$ independent copies of a given B-DMC $W$, a second set of $N$ binary-input channels $\{W_N^{(i)}:1\le i\le N\}$ such that, as $N$ becomes large, the fraction of indices $i$ for which $I(W_N^{(i)})$ is near 1 approaches $I(W)$ and the fraction for which $I(W_N^{(i)})$ is near 0 approaches $1-I(W)$.
The polarized channels $\{W_N^{(i)}\}$ are well-conditioned for channel coding: one need only send data at rate 1 through those with capacity near $1$ and at rate 0 through the remaining.
Codes constructed on the basis of this idea are called polar codes.
The paper proves that, given any B-DMC $W$ with $I(W)>0$ and any target rate $R < I(W)$, there exists a sequence of polar codes $\{{\mathscr C}_n;n\ge 1\}$ such that ${\mathscr C}_n$ has block-length $N=2^n$, rate $\ge R$, and probability of block error under successive cancellation decoding bounded as $P_{e}(N,R) \le \bigoh(N^{-\frac14})$ independently of the code rate. This performance is achievable by encoders and decoders with complexity $O(N\log N)$ for each.
\end{abstract}

\section{Introduction and overview}\label{sec:Introduction}

A fascinating aspect of Shannon's proof of the noisy channel coding theorem is the random-coding method that he used to show the existence of capacity-achieving code sequences without exhibiting any specific such sequence \cite{Shannon}.
Explicit construction of provably capacity-achieving code sequences with low encoding and decoding complexities has since then been an elusive goal.
This paper is an attempt to meet this goal for the class of B-DMCs.

We will give a description of the main ideas and results of the paper in this section.
First, we give some definitions and state some basic facts that are used throughout the paper.

\subsection{Preliminaries}

We write $W:{\cal X}\to {\cal Y}$ to denote a generic B-DMC with input alphabet ${\cal X}$, output alphabet ${\cal Y}$, and transition probabilities
$W(y|x)$, $x\in {\cal X}$, $y\in {\cal Y}$. The input alphabet ${\cal X}$ will always be $\{0,1\}$, the output alphabet and the transition probabilities may be arbitrary.
We write $W^N$ to denote the channel corresponding to $N$ uses of $W$; thus, $W^N: {\cal X}^N\to {\cal Y}^N$ with
$W^N(y_1^N\mid x_1^N) = \prod_{i=1}^N W(y_i\mid x_i)$.

Given a B-DMC $W$, there are two channel parameters of primary interest in this paper: the symmetric capacity
\begin{align*}
I(W) & \defn \sum_{y\in {\cal Y}} \sum_{x\in {\cal X}} \frac{1}{2} W(y|x) \log \frac{W(y|x)}{\frac{1}{2} W(y|0) +\frac{1}{2} W(y|1) }
\end{align*}
and the Bhattacharyya parameter
\begin{align*}
Z(W) & \defn \sum_{y\in {\mathcal Y}} \sqrt{W(y|0)W(y|1)}.
\end{align*}
These parameters are used as measures of {\sl rate} and {\sl reliability}, respectively. $I(W)$ is the highest rate at which reliable communication is possible across $W$ using the inputs of $W$ with equal frequency. $Z(W)$ is an upper bound on the probability of maximum-likelihood (ML) decision error when W is used
only once to transmit a 0 or 1.

It is easy to see that $Z(W)$ takes values in $[0,1]$.
Throughout, we will use base-2 logarithms; hence, $I(W)$ will also take values in $[0,1]$. The unit for code rates and channel capacities will be {\sl bits}.

Intuitively, one would expect that $I(W) \approx 1$ iff $Z(W)  \approx 0$, and $I(W) \approx 0$ iff $Z(W)  \approx 1$.
The following bounds, proved in the Appendix, make this precise.
\begin{Proposition}\label{rateReliability}
For any B-DMC $W$, we have
\begin{align}\label{Eq:rateReliability}
I(W) \ge \log \frac{2}{1+ Z(W)},
\end{align}
\begin{align}\label{Eq:rateReliability4}
I(W) \le \sqrt{1-Z(W)^2}.
\end{align}
\end{Proposition}
\vspace*{3mm}

The symmetric capacity $I(W)$ equals the Shannon capacity when $W$ is a {\sl symmetric} channel, i.e., a channel for which there exists a permutation $\pi$ of the output alphabet ${\cal Y}$ such that (i) $\pi^{-1} = \pi$ and (ii) $W(y|1) = W(\pi (y)|0)$ for all $y\in {\cal Y}$.
The binary symmetric channel (BSC) and the binary erasure channel (BEC) are examples of symmetric channels.
A BSC is a B-DMC $W$ with ${\cal Y} =\{0,1\}$, $W(0|0)=W(1|1)$, and $W(1|0)=W(0|1)$.
A B-DMC $W$ is called a BEC if for each $y\in{\cal Y}$, either $W(y|0)W(y|1)=0$ or $W(y|0)=W(y|1)$.
In the latter case, $y$ is said to be an {\sl erasure\/} symbol.
The sum of $W(y|0)$ over all erasure symbols $y$ is called the erasure probability of the BEC.

We denote random variables (RVs) by upper-case letters, such as $X$, $Y$, and their realizations (sample values) by the corresponding lower-case letters, such as $x$, $y$.
For $X$ a RV, $P_X$ denotes the probability assignment on  $X$.
For a joint ensemble of RVs $(X,Y)$, $P_{X,Y}$ denotes the joint probability
assignment.
We use the standard notation $I(X;Y)$, $I(X;Y|Z)$ to denote the mutual information
and its conditional form, respectively.

We use the notation $a_1^N$ as shorthand for denoting a row vector $(a_1,\ldots,a_N)$.
Given such a vector $a_1^N$, we write $a_i^j$, $1\le i,j\le N$, to denote the subvector
$(a_i,\ldots,a_j)$; if $j<i$, $a_i^j$ is regarded as void.
Given $a_1^N$ and ${\cal A}\subset \{1,\ldots,N\}$, we write $a_{\cal A}$ to denote the
subvector $(a_i: i\in {\cal A})$.
We write $a_{1,o}^{j}$ to denote the subvector with odd indices
$(a_k: 1\le k \le j;\; \text{$k$ odd})$.
We write $a_{1,e}^{j}$ to denote the subvector with even indices
$(a_k: 1\le k \le j; \; \text{$k$ even})$.
For example, for $a_1^5=(5,4,6,2,1)$, we have $a_2^4=(4,6,2)$,
$a_{1,e}^5=(4,2)$, $a_{1,o}^4 = (5,6)$.
The notation $0_1^N$ is used to denote the all-zero vector.

Code constructions in this paper will be carried out in vector spaces over the binary field GF(2). Unless specified otherwise, all vectors, matrices, and operations on them will be over GF(2).
In particular, for $a_1^N$, $b_1^N$ vectors over GF(2), we write $a_1^N \oplus b_1^N$ to denote their componentwise mod-2 sum.
The Kronecker product of an $m$-by-$n$ matrix $A=[A_{ij}]$ and an $r$-by-$s$ matrix $B=[B_{ij}]$ is defined as
\begin{align*}
A\otimes B & = \left[
\begin{array}{ccc}
A_{11} B & \cdots & A_{1n} B\\
 \vdots & \ddots & \vdots  \\
A_{m1} B & \cdots & A_{mn} B
\end{array} \right],
\end{align*}
which is an $mr$-by-$ns$ matrix.
The Kronecker power $A^{\otimes n}$ is defined as $A\otimes A^{\otimes (n-1)}$ for all $n\ge 1$. We will follow the convention that $A^{\otimes 0}\defn [1]$.

We write $|{\cal A}|$ to denote the number of elements in a set ${\cal A}$.
We write $1_{\cal A}$ to denote the indicator function of a set ${\cal A}$; thus, $1_{\cal A}(x)$ equals $1$ if $x\in {\cal A}$ and 0 otherwise.

We use the standard Landau notation $O(N)$, $o(N)$, $\omega(N)$ to denote the asymptotic behavior of functions.

\subsection{Channel polarization}\label{sec:Method}
Channel polarization is an operation by which one manufactures out of $N$ independent copies of a given B-DMC $W$ a second set of $N$ channels $\{W_N^{(i)}:1\le i\le N\}$ that show a polarization effect in the sense that, as $N$ becomes large, the symmetric capacity terms $\{I(W_N^{(i)})\}$ tend towards 0 or 1 for all but a vanishing fraction of indices $i$.
This operation consists of a channel combining phase and a channel splitting phase.

\subsubsection{Channel combining}\label{sec:Combining}
This phase combines copies of a given B-DMC $W$ in a recursive manner to produce a vector channel $W_N:{\cal X}^N \to {\cal Y}^N$, where $N$ can be any power of two, $N=2^n$, $n\ge 0$.
The recursion begins at the 0-th level ($n=0$) with only one copy of $W$ and we set $W_1 \defn W$.
The first level ($n=1$) of the recursion combines two independent copies of $W_1$ as shown in Fig.~1  and obtains the channel $W_2 :{\cal X}^2 \to {\cal Y}^2$ with the transition probabilities
\begin{align}\label{eqn:combine2}
W_2(y_1,y_2|u_1,u_2)= W(y_1|u_1\oplus u_2)W(y_2|u_2).
\end{align}
\input{fig1}

The next level of the recursion is shown in Fig.~2 where two independent copies of $W_2$ are combined to create the channel $W_4:{\cal X}^4 \to {\cal Y}^4$ with transition probabilities
$W_4(y_1^4|u_1^4) = W_2(y_1^2|u_1\oplus u_2,u_3\oplus u_4) W_2(y_3^4|u_2,u_4)$.
\input{fig2}

In Fig.~2, $R_4$ is the permutation operation that maps an input $(s_1,s_2,s_3,s_4)$ to $v_1^4 = (s_1,s_3,s_2,s_4)$.
The mapping $u_1^4 \mapsto x_1^4$ from the input of $W_4$ to the input of $W^4$ can be written as $x_1^4 = u_1^4 G_4$ with
$G_4 = \left[\begin{smallmatrix} 1 & 0 & 0 & 0 \\ 1 & 0 & 1 & 0 \\ 1 & 1 & 0 & 0 \\ 1 & 1 & 1 & 1 \end{smallmatrix}\right].$
Thus, we have the relation $W_4(y_1^4|u_1^4) = W^4(y_1^4|u_1^4G_4)$ between the transition probabilities of $W_4$ and those of $W^4$.

\input{fig3}

The general form of the recursion is shown in Fig.~3 where two independent copies of $W_{N/2}$ are combined to produce the channel $W_{N}$.
The input vector $u_1^{N}$ to $W_{N}$ is first transformed into $s_1^{N}$ so that $s_{2i-1}=u_{2i-1}\oplus u_{2i}$ and $s_{2i}=u_{2i}$ for $1\le i\le N/2$.
The operator $R_N$ in the figure is a permutation, known as the {\sl reverse shuffle} operation, and acts on its input $s_1^{N}$to produce $v_1^N = (s_1,s_3,\ldots,s_{N-1},s_2,s_4,\ldots,s_{N})$, which becomes the input to the two copies of $W_{N/2}$ as shown in the figure.

We observe that the mapping $u_1^N\mapsto v_1^N$ is linear over GF(2).
It follows by induction that the overall mapping $u_1^N\mapsto x_1^N$, from the input of the synthesized channel $W_N$ to the input of the underlying raw channels $W^N$, is also linear and may be represented by a matrix $G_N$ so that $x_1^N =u_1^N G_N$. We call $G_N$ the {\sl generator matrix} of size $N$.
The transition probabilities of the two channels $W_N$ and $W^N$ are related by
\begin{align}\label{defn:GN}
W_N(y_1^N|u_1^N) = W^N(y_1^N|u_1^N G_N)
\end{align}
for all $y_1^N\in {\cal Y}^N$, $u_1^N\in {\cal X}^N$.
We will show in Sect.~\ref{sec:Encoder} that $G_N$ equals $B_NF^{\otimes n}$ for any $N=2^n$, $n\ge 0$, where $B_N$ is a permutation matrix known as {\sl bit-reversal} and $F\defn\left[\begin{smallmatrix} 1 & 0 \\ 1 & 1\end{smallmatrix} \right]$. Note that the channel combining operation is fully specified by the matrix $F$. Also note that $G_N$ and $F^{\otimes n}$ have the same set of rows, but in a different (bit-reversed) order; we will discuss this topic more fully in Sect.~\ref{sec:Encoder}.

\subsubsection{Channel splitting}
Having synthesized the vector channel $W_N$ out of $W^N$, the next step of channel polarization is to split $W_N$ back into a set of $N$ binary-input coordinate channels $W_N^{(i)}:{\cal X} \to {\cal Y}^N\times {\cal X}^{i-1}$, $1\le i\le N$, defined by the transition probabilities
\begin{align}\label{eqn:splitChannels}
W_N^{(i)}(y_1^N,u_1^{i-1}|u_i) \defn \sum_{u_{i+1}^N\in {\cal X}^{N-i}} \frac{1}{2^{N-1}} W_N(y_1^N|u_{1}^N),
\end{align}
where $(y_1^N,u_1^{i-1})$ denotes the output of $W_N^{(i)}$ and $u_i$ its input.

To gain an intuitive understanding of the channels $\{W_N^{(i)}\}$, consider a genie-aided successive cancellation decoder in which the $i$th decision element estimates $u_i$ after observing $y_1^N$ and the {\sl past} channel inputs $u_1^{i-1}$ (supplied correctly by the genie regardless of any decision errors at earlier stages). If $u_1^N$ is a-priori uniform on ${\cal X}^N$, then $W_N^{(i)}$ is the effective channel seen by the $i$th decision element in this scenario.

\subsubsection{Channel polarization}
\begin{Theorem}\label{thm:Polarization}
For any B-DMC $W$, the channels $\{W_N^{(i)}\}$ {\sl polarize} in the sense that, for any fixed $\delta \in(0,1)$, as $N$ goes to infinity through powers of two, the fraction of indices $i\in \{1,\ldots,N\}$ for which $I(W_N^{(i)})\in (1-\delta,1]$ goes to $I(W)$ and the fraction for which $I(W_N^{(i)})\in [0,\delta)$ goes to $1-I(W)$.
\end{Theorem}

This theorem is proved in Sect.~\ref{sec:PeProposedScheme}.

\begin{figure}[th]
\begin{center}
\resizebox{!}{2.5in}{\includegraphics{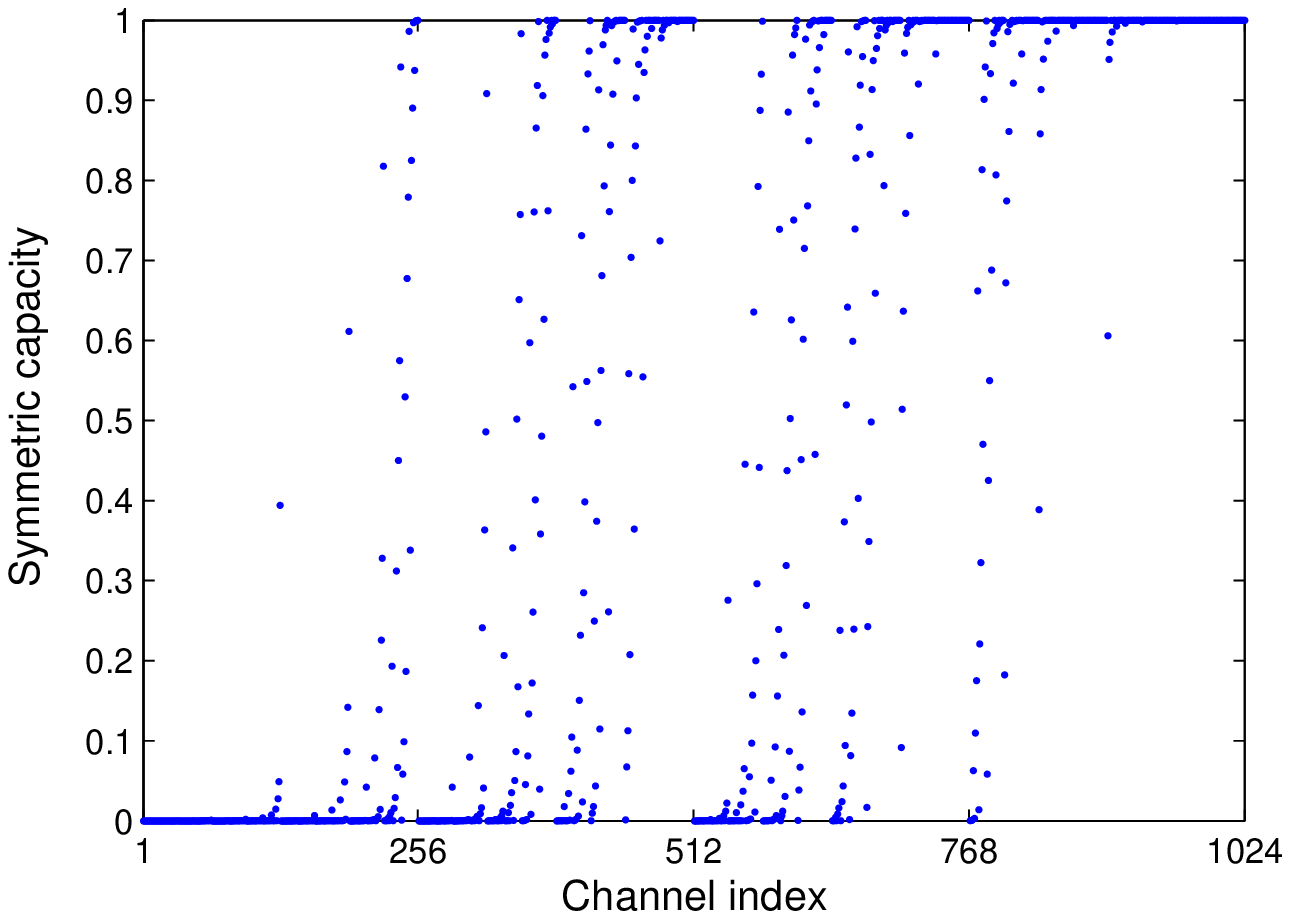}}
\caption{Plot of $I(W_N^{(i)})$ vs. $i=1,\ldots,N=2^{10}$ for a BEC with $\epsilon = 0.5$.}
\label{fig:random_walk}
\end{center}
\end{figure}

The polarization effect is illustrated in Fig.~4 for the case $W$ is a BEC with erasure probability $\epsilon = 0.5$.
The numbers $\{I(W_N^{(i)})\}$ have been computed using the recursive relations
\begin{align}\label{eqn:IBEC}
\begin{split}
I(W_{N}^{(2i-1)}) & =  I(W_{N/2}^{(i)})^2,\\
I(W_{N}^{(2i)}) & =  2I(W_{N/2}^{(i)})- I(W_{N/2}^{(i)})^2,\\
\end{split}
\end{align}
with $I(W_{1}^{(1)})=1-\epsilon$.
This recursion is valid only for BECs and it is proved in Sect.~\ref{sec:rateReliability}.
No efficient algorithm is known for calculation of $\{I(W_N^{(i)})\}$ for a general B-DMC $W$.

Figure~4 shows that $I(W^{(i)})$ tends to be near 0 for small $i$ and near 1 for large $i$.
However, $I(W_N^{(i)})$ shows an erratic behavior for an intermediate range of $i$.
For general B-DMCs, determining the subset of indices $i$ for which $I(W_N^{(i)})$ is above a given threshold is an important computational problem that will be addressed in Sect.~\ref{sec:Construction}.

\subsubsection{Rate of polarization}

For proving coding theorems, the speed with which the polarization effect takes hold as a function of $N$ is important.
Our main result in this regard is given in terms of the parameters
\begin{align}\label{defn:ZWNi}
Z&(W_N^{(i)})  = \nonumber \\ & \sum_{y_1^N \in {\cal Y}^N} \sum_{u_1^{i-1} \in {\cal X}^{i-1}}   \sqrt{W_N^{(i)}(y_1^N,u_1^{i-1}\mid 0)\;W_N^{(i)}(y_1^N,u_1^{i-1}\mid 1)}.
\end{align}

\begin{Theorem}\label{ZnFastConvergence}
For any B-DMC $W$ with $I(W) >0$, and any fixed $R< I(W)$,
there exists a sequence of sets ${\cal A}_N \subset \{1,\ldots,N\}$, $N\in\{1,2,\ldots,2^n,\ldots\}$, such that
$|{\cal A}_N|\ge N R$ and $Z(W_{N}^{(i)}) \le O(N^{-5/4})$ for all $i\in {\cal A}_N$.
\end{Theorem}

This theorem is proved in Sect.~\ref{sec:ZnFastConvergence}.

We stated the polarization result in Theorem~\ref{ZnFastConvergence} in terms $\{Z(W_N^{(i)})\}$ rather than $\{I(W_N^{(i)})\}$ because this form is better suited to the coding results that we will develop.
A rate of polarization result in terms of $\{I(W_N^{(i)})\}$ can be obtained from Theorem~\ref{ZnFastConvergence} with the help of Prop.~\ref{rateReliability}.

\subsection{Polar coding}\label{sec:proposedSCS}

We take advantage of the polarization effect to construct codes that achieve the symmetric channel capacity $I(W)$ by a method we call \emph{polar coding}.
The basic idea of polar coding is to create a coding system where one can access each coordinate channel $W_N^{(i)}$ individually and send data only through those for which $Z(W_N^{(i)})$ is near $0$.

\subsubsection{$G_N$-coset codes}\label{sec:CosetCodes}
We first describe a class of block codes that contain polar codes---the codes of main interest---as a special case.
The block-lengths $N$ for this class are restricted to powers of two, $N=2^n$ for some $n\ge 0$.
For a given $N$, each code in the class is encoded in the same manner, namely,
\begin{align}\label{defn:Encoder}
x_1^N & = u_1^N G_N
\end{align}
where $G_N$ is the generator matrix of order $N$, defined above.
For ${\cal A}$ an arbitrary subset of $\{1,\ldots,N\}$, we may write \eqref{defn:Encoder} as
\begin{align}\label{defn:Encoder2}
x_1^N & = u_{\cal A} G_N({\cal A}) \oplus u_{{\cal A}^c} G_N({{\cal A}^c})
\end{align}
where $G_N({\cal A})$ denotes the submatrix of $G_N$ formed by the rows with indices in ${\cal A}$.

If we now fix ${\cal A}$ and $u_{{\cal A}^c}$, but leave $u_{\cal A}$ as a free variable, we obtain a mapping from source blocks $u_{{\cal A}}$ to codeword blocks $x_1^N$.
This mapping is a {\sl coset code}: it is a coset of the linear block code with generator matrix $G_N({\cal A})$, with the coset determined by the fixed vector $u_{{\cal A}^c} G_N({{\cal A}^c})$.
We will refer to this class of codes collectively as {\sl $G_N$-coset codes}.
Individual $G_N$-coset codes will be identified by a parameter vector $(N,K,{\cal A},u_{{\cal A}^c})$, where $K$ is the code dimension and specifies the size of ${\cal A}$.\footnote{We include the redundant parameter $K$ in the parameter set because often we consider an ensemble of codes with $K$ fixed and ${\cal A}$ free.} The ratio $K/N$ is called the {\sl code rate}.
We will refer to ${\cal A}$ as the {\sl information set} and to $u_{{\cal A}^c} \in {\cal X}^{N-K}$ as {\sl frozen} bits or vector.

For example, the $(4,2,\{2,4\},(1,0))$ code has the encoder mapping
\begin{align}
x_1^4 & = u_1^4 G_4\nonumber \\
 & = (u_2,u_4)
\begin{bmatrix} 1 & 0 & 1 & 0 \\ 1 & 1 & 1 & 1 \end{bmatrix} + (1,0) \begin{bmatrix} 1 & 0 & 0 & 0\\1 & 1 & 0 & 0\end{bmatrix}.
\end{align}
For a source block $(u_2,u_4) = (1,1)$, the coded block is $x_1^4 = (1,1,0,1)$.

Polar codes will be specified shortly by giving a particular rule for the selection of the information set ${\cal A}$.

\subsubsection{A successive cancellation decoder}\label{sec:SCD}
Consider a $G_N$-coset code with parameter $(N,K,{\cal A},u_{{\cal A}^c})$.
Let $u_1^N$ be encoded into a codeword $x_1^N$, let $x_1^N$ be sent over the channel $W^N$, and let a channel output $y_1^N$ be received.
The decoder's task is to generate an estimate $\hat{u}_1^N$ of $u_1^N$, given knowledge of ${\cal A}$, $u_{{\cal A}^c}$, and $y_1^N$.
Since the decoder can avoid errors in the frozen part by setting $\hat{u}_{{\cal A}^c} = u_{{\cal A}^c}$, the real decoding task is to generate an estimate $\hat{u}_{\cal A}$ of $u_{\cal A}$.

The coding results in this paper will be given with respect to a specific successive cancellation (SC) decoder, unless some other decoder is mentioned. Given any $(N,K,{\cal A},u_{{\cal A}^c})$ $G_N$-coset code, we will use a SC decoder that generates its decision $\hat{u}_1^N$ by computing
\begin{align}\label{eqn:decisions}
\hat{u}_i & \defn \begin{cases}
u_i, & \text{if $i\in {\cal A}^c$}\\
h_i(y_1^N,\hat{u}_1^{i-1}), & \text{if $i\in {\cal A}$}
\end{cases}
\end{align}
in the order $i$ from 1 to $N$, where $h_i:{\cal Y}^N \times {\cal X}^{i-1} \to {\cal X}$, $i\in {\cal A}$, are {\sl decision functions} defined as
\begin{align}\label{eqn:DecRegion}
h_i(y_1^N,\hat{u}_1^{i-1}) & \defn \begin{cases} 0, & \text{if  $\frac{W_N^{(i)}(y_1^N,\hat{u}_1^{i-1}\mid 0)}{W_N^{(i)}(y_1^N,\hat{u}_1^{i-1}\mid 1)}\ge 1$}\\
1, & \text{otherwise}
\end{cases}
\end{align}
for all $y_1^N\in {\cal Y}^N$, $\hat{u}_1^{i-1}\in {\cal X}^{i-1}$.
We will say that a decoder {\sl block error} occurred if $\hat{u}_1^N \neq u_1^N$ or equivalently if $\hat{u}_{\cal A} \neq u_{\cal A}$.

The decision functions $\{h_i\}$ defined above resemble ML decision functions but are not exactly so, because they treat the {\sl future} frozen bits $(u_j:\text{$j>i$, $j \in {\cal A}^c$})$ as RVs, rather than as known bits.
In exchange for this suboptimality, $\{h_i\}$ can be computed efficiently using recursive formulas, as we will show in Sect.~\ref{sec:channelSplitting}.
Apart from algorithmic efficiency, the recursive structure of the decision functions is important because it renders the performance analysis of the decoder tractable.
Fortunately, the loss in performance due to not using true ML decision functions happens to be negligible: $I(W)$ is still achievable.

\subsubsection{Code performance}\label{sec:CodingTheorems}
The notation $P_e(N,K,{\cal A},u_{{\cal A}^c})$ will denote the probability of block error for a $(N,K,{\cal A},u_{{\cal A}^c})$ code, assuming that each data vector $u_{\cal A}\in {\cal X}^K$ is sent with probability $2^{-K}$ and decoding is done by the above SC decoder. More precisely,
\begin{multline*}
P_e(N,K,{\cal A},u_{{\cal A}^c})  \defn \\
\sum_{u_{\cal A}\in {\cal X}^K} \frac{1}{2^K} \sum_{y_1^N\in{\cal Y}^N\, :\, \hat{u}_1^N(y_1^N) \neq u_1^N}
 W_N(y_1^N|u_1^N).
\end{multline*}
The average of $P_e(N,K,{\cal A},u_{{\cal A}^c})$ over all choices for $u_{{\cal A}^c}$ will be denoted by $P_e(N,K,{\cal A})$:
\begin{align*}
P_e(N,K,{\cal A}) & \defn \sum_{u_{{\cal A}^c}\in {\cal X}^{N-K}} \frac{1}{2^{N-K}} P_e(N,K,{\cal A},u_{{\cal A}^c}).
\end{align*}

A key bound on block error probability under SC decoding is the following.
\begin{Proposition}\label{prop:probErrorBound}
For any B-DMC $W$ and any choice of the parameters $(N,K,{\cal A})$,
\begin{align}\label{errorprob3}
P_e(N,K,{\cal A}) & \le \sum_{i\in {\cal A}} Z(W_N^{(i)}).
\end{align}
Hence, for each $(N,K,{\cal A})$, there exists a frozen vector $u_{{\cal A}^c}$ such that
\begin{align}\label{errorprob4}
P_e(N,K,{\cal A},u_{{\cal A}^c}) & \le \sum_{i\in {\cal A}} Z(W_N^{(i)}).
\end{align}
\end{Proposition}
\vspace*{2mm}

This is proved in Sect.~\ref{sec:PeBound}.
This result suggests choosing ${\cal A}$ from among all $K$-subsets of $\{1,\ldots,N\}$ so as to minimize the RHS of \eqref{errorprob3}. This idea leads to the definition of polar codes.

\subsubsection{Polar codes}\label{sec:PolarCodes}
Given a B-DMC $W$, a $G_N$-coset code with parameter $(N,K,{\cal A},u_{{\cal A}^c})$  will be called a {\sl polar code} for $W$ if the information set ${\cal A}$ is chosen as a $K$-element subset of $\{1,\ldots,N\}$ such that $Z(W_N^{(i)})\le Z(W_N^{(j)})$ for all $i\in {\cal A}$, $j\in {\cal A}^c$.

Polar codes are channel-specific designs: a polar code for one channel may not be a polar code for another.
The main result of this paper will be to show that polar coding achieves the symmetric capacity $I(W)$ of any given B-DMC $W$.

An alternative rule for polar code definition would be to specify ${\cal A}$ as a $K$-element subset of $\{1,\ldots,N\}$ such that $I(W_N^{(i)})\ge I(W_N^{(j)})$ for all $i\in{\cal A}$, $j\in {\cal A}^c$. This alternative rule would also achieve $I(W)$. However, the rule based on the Bhattacharyya parameters has the advantage of being connected with an explicit bound on block error probability.

The polar code definition does not specify how the frozen vector $u_{{\cal A}^c}$ is to be chosen; it may be chosen at will.
This degree of freedom in the choice of $u_{{\cal A}^c}$ simplifies the performance analysis of polar codes by allowing averaging over an ensemble.
However, it is not for analytical convenience alone that we do not specify a precise rule for selecting $u_{{\cal A}^c}$, but also because it appears that the code performance is relatively insensitive to that choice. In fact, we prove in Sect.~\ref{sec:SymmetricPeBound} that, for symmetric channels, any choice for $u_{{\cal A}^c}$ is as good as any other.

\subsubsection{Coding theorems}\label{sec:mainResults}

Fix a B-DMC $W$ and a number $R\ge 0$.
Let $P_e(N,R)$ be defined as $P_e(N,\lfloor NR\rfloor ,{\cal A})$ with ${\cal A}$ selected in accordance with the polar coding rule for $W$.
Thus, $P_e(N,R)$ is the probability of block error under SC decoding for polar coding over $W$ with block-length $N$ and rate $R$, averaged over all choices for the frozen bits $u_{{\cal A}^c}$.
The main coding result of this paper is the following:
\begin{Theorem}\label{thm:rate}
For any given B-DMC $W$ and fixed $R< I(W)$, block error probability for polar coding under successive cancellation decoding satisfies
\begin{align}\label{eqn:PeMainResult}
P_e(N,R) & = O(N^{-\frac{1}{4}}).
\end{align}
\end{Theorem}

This theorem follows as an easy corollary to Theorem~\ref{ZnFastConvergence} and the bound \eqref{errorprob3}, as we show in Sect.~\ref{sec:PeBound}.
For symmetric channels, we have the following stronger version of Theorem~\ref{thm:rate}.
\begin{Theorem}\label{thm:rateSymmetric}
For any symmetric B-DMC $W$ and any fixed $R<I(W)$, consider any sequence of $G_N$-coset codes $(N,K,{\cal A},u_{{\cal A}^c})$ with $N$ increasing to infinity, $K=\lfloor NR\rfloor$, ${\cal A}$ chosen in accordance with the polar coding rule for $W$, and $u_{{\cal A}^c}$ fixed arbitrarily.
The block error probability under successive cancellation decoding satisfies
\begin{align}\label{eqn:PeMainResultSymmetric}
P_e(N,K,{\cal A},u_{{\cal A}^c}) & = O(N^{-\frac{1}{4}}).
\end{align}
\end{Theorem}
\vspace*{2mm}

This is proved in Sect.~\ref{sec:SymmetricPeBound}.
Note that for symmetric channels $I(W)$ equals the Shannon capacity of $W$.

\subsubsection{Complexity}\label{subsec:Complexity}

An important issue about polar coding is the complexity of encoding, decoding, and code construction. The recursive structure of the channel polarization construction leads to low-complexity encoding and decoding algorithms for the class of $G_N$-coset codes, and in particular, for polar codes.
\begin{Theorem}\label{thm:complexity}
For the class of $G_N$-coset codes, the complexity of encoding and the complexity of successive cancellation decoding are both $O(N\log N)$ as functions of code block-length $N$.
\end{Theorem}
\vspace*{1mm}

This theorem is proved in Sections~\ref{sec:Encoder} and \ref{sec:Decoding}.
Notice that the complexity bounds in Theorem~\ref{thm:complexity} are independent of the code rate and the way the frozen vector is chosen. The bounds hold even at rates above $I(W)$, but clearly this has no practical significance.

As for code construction, we have found no low-complexity algorithms for constructing polar codes.
One exception is the case of a BEC for which we have a polar code construction algorithm with complexity $\bigoh(N)$.
We discuss the code construction problem further in Sect.~\ref{sec:Construction} and suggest a low-complexity statistical algorithm for approximating the exact polar code construction.

\subsection{Relations to previous work}\label{subsec:relationPreviousWork}
This paper is an extension of work begun in \cite{Arikan2006}, where channel combining and splitting were used to show that improvements can be obtained in the sum cutoff rate for some specific DMCs. However, no recursive method was suggested there to reach the ultimate limit of such improvements.

As the present work progressed, it became clear that polar coding had much in common with Reed-Muller (RM) coding  \cite{Muller}, \cite{Reed}.
Indeed, recursive code construction and SC decoding, which are two essential ingredients of polar coding, appear to have been introduced into coding theory by RM codes.

According to one construction of RM codes, for any $N=2^n$, $n\ge 0$, and $0\le K\le N$, an RM code with block-length $N$ and dimension $K$, denoted $\text{RM}(N,K)$, is defined as a linear code whose generator matrix $G_{RM}(N,K)$ is obtained by deleting $(N-K)$ of the rows of $F^{\otimes n}$ so that none of the deleted rows has a larger Hamming weight (number of 1s in that row) than any of the remaining $K$ rows.
For instance,
$G_{RM}(4,4)  = F^{\otimes 2} = \left[\begin{smallmatrix}
     1 & 0 & 0 & 0 \\
     1 & 1 & 0 & 0 \\
     1 & 0 & 1 & 0 \\
     1 & 1 & 1 & 1
\end{smallmatrix}\right]$
and
$G_{RM}(4,2) = \left[\begin{smallmatrix}
     1 & 0 & 1 & 0 \\
     1 & 1 & 1 & 1
\end{smallmatrix}\right]$.

This construction brings out the similarities between RM codes and polar codes.
Since $G_N$ and $F^{\otimes n}$ have the same set of rows (only in a different order) for any $N=2^n$, it is clear that RM codes belong to the class of $G_N$-coset codes.
For example, $\text{RM}(4,2)$ is the $G_4$-coset code with parameter $(4,2,\{2,4\},(0,0))$.
So, RM coding and polar coding may be regarded as two alternative rules for selecting the information set ${\cal A}$ of a $G_N$-coset code of a given size $(N,K)$.
Unlike polar coding, RM coding selects the information set in a channel-independent manner; it is not as fine-tuned to the channel polarization phenomenon as polar coding is.
We will show in Sect.~\ref{sec:RM} that, at least for the class of BECs, the RM rule for information set selection leads to asymptotically unreliable codes under SC decoding. So, polar coding goes beyond RM coding in a non-trivial manner by paying closer attention to channel polarization.

Another connection to existing work can be established by noting that polar codes are multi-level $|u|u+v|$ codes, which are a class of codes originating from Plotkin's method for code combining \cite{Plotkin}.
This connection is not surprising in view of the fact that RM codes are also multi-level $|u|u+v|$ codes \cite[pp.~114-125]{Costello}.
However, unlike typical multi-level code constructions where one begins with specific small codes to build larger ones,
in polar coding the multi-level code is obtained by expurgating rows of a full-order generator matrix, $G_N$, with respect to a channel-specific criterion. The special structure of $G_N$ ensures that, no matter how expurgation is done, the resulting code is a multi-level $|u|u+v|$ code. In essence, polar coding enjoys the freedom to pick a multi-level code from an ensemble of such codes so as to suit the channel at hand, while conventional approaches to multi-level coding do not have this degree of flexibility.

Finally, we wish to mention a ``spectral'' interpretation of polar codes which is similar to Blahut's
treatment of BCH codes \cite[Ch.~9]{Blahut83}; this type of similarity has already been pointed out by Forney \cite[Ch.~11]{ForneyLectures} in connection with RM codes.
From the spectral viewpoint, the encoding operation \eqref{defn:Encoder} is regarded as a transform of a ``frequency'' domain information vector $u_1^N$ to a ``time'' domain codeword vector $x_1^N$. The transform is invertible with $G_N^{-1} = G_N$.
The decoding operation is regarded as a spectral estimation problem in which one is given a time domain observation $y_1^N$, which is a noisy version of $x_1^N$, and asked to estimate $u_1^N$. To aid the estimation task, one is allowed to freeze a certain number of spectral components of $u_1^N$.
This spectral interpretation of polar coding suggests that it may be possible to treat polar codes and BCH codes in a unified framework. The spectral interpretation also opens the door to the use of various signal processing techniques in polar coding; indeed, in Sect.~\ref{sec:Encoder}, we exploit some fast transform techniques in designing encoders for polar codes.

\subsection{Paper outline}
The rest of the paper is organized as follows.
Sect.~\ref{sec:channelSplitting} explores the recursive properties of the channel splitting operation.
In Sect.~\ref{sec:rateReliability}, we focus on how $I(W)$ and $Z(W)$ get transformed through a single step of channel combining and splitting.  We extend this to an asymptotic analysis in Sect.~\ref{sec:PeProposedScheme} and complete the proofs of Theorem~\ref{thm:Polarization} and Theorem~\ref{ZnFastConvergence}.
This completes the part of the paper on channel polarization; the rest of the paper is mainly about polar coding.
Section~\ref{sec:errorAnalysis} develops an upper bound on the block error probability of polar coding under SC decoding and proves Theorem~\ref{thm:rate}.
Sect.~\ref{sec:symmetricChannels} considers polar coding for symmetric B-DMCs and proves Theorem~\ref{thm:rateSymmetric}.
Sect.~\ref{sec:Encoder} gives an analysis of the encoder mapping $G_N$, which results in efficient encoder implementations.
In Sect.~\ref{sec:Decoding}, we give an implementation of SC decoding with complexity $O(N\log N)$.
In Sect.~\ref{sec:Construction}, we discuss the code construction complexity and propose an $O(N\log N)$ statistical algorithm for approximate code construction.
In Sect.~\ref{sec:RM}, we explain why RM codes have a poor asymptotic performance under SC decoding.
In Sect.~\ref{sec:Generalizations}, we point out some generalizations of the present work, give some complementary remarks, and state some open problems.

\section{Recursive channel transformations}\label{sec:channelSplitting}

We have defined a blockwise channel combining and splitting operation by \eqref{defn:GN} and \eqref{eqn:splitChannels} which transformed $N$ independent copies of $W$ into $W_N^{(1)}$, \dots, $W_N^{(N)}$.
The goal in this section is to show that this blockwise channel transformation can be broken recursively into single-step channel transformations.

We say that a pair of binary-input channels $W^\prime:{\cal X}\to \tilde{\cal Y}$ and $W^{\prime\prime}:{\cal X}\to \tilde{\cal Y}\times {\cal X}$ are obtained by a single-step transformation of two independent copies of a binary-input channel $W:{\cal X}\to {\cal Y}$ and write $$(W,W) \mapsto (W^\prime,W^{\prime\prime})$$
iff there exists a one-to-one mapping $f:{\cal Y}^2 \to \tilde{\cal Y}$ such that
\begin{gather}
W^\prime(f(y_1,y_2)|u_1) = \sum_{u_2^\prime} \frac{1}{2} W(y_1|u_1\oplus u_2^\prime)W(y_2|u_2^\prime),\label{eqn:splitPPGen1}\\
W^{\prime\prime}(f(y_1,y_2),u_1|u_2) = \frac{1}{2} W(y_1|u_1\oplus u_2)W(y_2|u_2)\label{eqn:splitPPGen2}
\end{gather}
for all $u_1,u_2\in {\cal X}$, $y_1,y_2 \in {\cal Y}$.

According to this, we can write $(W,W)\mapsto (W_2^{(1)},W_2^{(2)})$ for any given B-DMC $W$ because
\begin{align}
W_{2}^{(1)}(y_1^{2}|u_1) & \defn \sum_{u_2} \frac{1}{2} W_2(y_1^2|u_1^2) \notag \\
& =  \sum_{u_{2}} \frac{1}{2} W(y_1|u_1\oplus u_2) W(y_2|u_2) \label{split1},\\
W_{2}^{(2)}(y_1^{2},u_1|u_2) & \defn \frac{1}{2} W_2(y_1^2|u_1^2) \notag\\
& = \frac{1}{2} W(y_1|u_1\oplus u_2)W(y_2|u_2), \label{split2}
\end{align}
which are in the form of \eqref{eqn:splitPPGen1} and \eqref{eqn:splitPPGen2} by taking $f$ as the identity mapping.

It turns out we can write, more generally,
\begin{align}\label{localstep}
(W_N^{(i)},W_N^{(i)}) \mapsto (W_{2N}^{(2i-1)},W_{2N}^{(2i)}).
\end{align}
This follows as a corollary to the following:
\begin{Proposition}\label{prop:channelRecursion}
For any $n\ge 0$, $N=2^n$, $1\le i\le N$,
\begin{multline}
W_{2N}^{(2i-1)}(y_1^{2N},u_1^{2i-2}|u_{2i-1}) = \\
 \sum_{u_{2i}}  \frac{1}{2}\; W_{N}^{(i)}(y_1^{N},u_{1,o}^{2i-2} \oplus u_{1,e}^{2i-2}|u_{2i-1} \oplus u_{2i})\\
 \cdot W_{N}^{(i)}(y_{N+1}^{2N},u_{1,e}^{2i-2}|u_{2i}) \label{eqn:Wrecursion1}
\end{multline}
and
\begin{multline}
W_{2N}^{(2i)}(y_1^{2N},u_1^{2i-1}|u_{2i}) =  \\
\frac{1}{2}\; W_{N}^{(i)}(y_1^{N},u_{1,o}^{2i-2} \oplus u_{1,e}^{2i-2}|u_{2i-1} \oplus u_{2i})\\
\cdot W_{N}^{(i)}(y_{N+1}^{2N},u_{1,e}^{2i-2}|u_{2i}). \label{eqn:Wrecursion2}
\end{multline}
\end{Proposition}
\vspace*{3mm}

This proposition is proved in the Appendix.
The transform relationship \eqref{localstep} can now be justified by
noting that \eqref{eqn:Wrecursion1} and \eqref{eqn:Wrecursion2} are identical in form to \eqref{eqn:splitPPGen1} and \eqref{eqn:splitPPGen2}, respectively, after the following substitutions:
\begin{align*}
W &\leftarrow W_N^{(i)},  &
W^\prime &\leftarrow W_{2N}^{(2i-1)},\\
W^{\prime\prime} &\leftarrow W_{2N}^{(2i)}, &
u_1 &\leftarrow u_{2i-1},\\
u_2  & \leftarrow u_{2i}, &
y_1 & \leftarrow (y_1^N,u_{1,o}^{2i-2} \oplus u_{1,e}^{2i-2}),\\
y_2 & \leftarrow (y_{N+1}^{2N},u_{1,e}^{2i-2}), &
f(y_1,y_2) &\leftarrow (y_1^{2N},u_1^{2i-2}).
\end{align*}

\input{fig5}

Thus, we have shown that the blockwise channel transformation from $W^N$ to $(W_N^{(1)},\ldots,W_N^{(N)})$ breaks at a local level into  single-step channel transformations of the form \eqref{localstep}.
The full set of such transformations form a fabric as shown in Fig.~5 for $N=8$.
Reading from right to left, the figure starts with four copies of the transformation $(W,W)\mapsto (W_2^{(1)},W_2^{(2)})$ and continues in \emph{butterfly} patterns, each representing a channel transformation of the form $(W_{2^i}^{(j)},W_{2^i}^{(j)})\mapsto (W_{2^{i+1}}^{(2j-1)},W_{2^{i+1}}^{(2j)})$.
The two channels at the right end-points of the butterflies are always identical and independent. At the rightmost level there are 8 independent copies of $W$; at the next level to the left, there are 4 independent copies of $W_2^{(1)}$ and $W_2^{(2)}$ each; and so on. Each step to the left doubles the number of channel types, but halves the number of independent copies.

\section{Transformation of rate and reliability}\label{sec:rateReliability}
We now investigate how the rate and reliability parameters, $I(W_N^{(i)})$ and $Z(W_N^{(i)})$, change through a local (single-step)  transformation  \eqref{localstep}. By understanding the local behavior, we will be able to reach conclusions about the overall transformation from $W^N$ to $(W_N^{(1)},\ldots,W_N^{(N)})$.
Proofs of the results in this section are given in the Appendix.

\subsection{Local transformation of rate and reliability}\label{subsec:Basic}

\begin{Proposition}\label{prop:rateConservationS}
Suppose $(W,W) \mapsto (W^\prime,W^{\prime\prime})$ for some set of binary-input channels.
Then,
\begin{align}
I(W^\prime)+I(W^{\prime\prime}) & = 2I(W),\label{eqn:conservation1S}\\
I(W^{\prime}) & \le I(W^{\prime\prime}) \label{eqn:FixedPointS}
\end{align}
with equality iff $I(W)$ equals 0 or 1.
\end{Proposition}

The equality \eqref{eqn:conservation1S} indicates that the single-step channel transform preserves the symmetric capacity.
The inequality \eqref{eqn:FixedPointS} together with \eqref{eqn:conservation1S} implies that the symmetric capacity remains unchanged under a single-step transform, $I(W^\prime)=I(W^{\prime\prime})=I(W)$, iff $W$ is either a perfect channel or a completely noisy one.
If $W$ is neither perfect nor completely noisy, the single-step transform moves the symmetric capacity away from the center in the sense that $I(W^\prime) < I(W) <I(W^{\prime\prime})$, thus helping polarization.

\begin{Proposition}\label{Proposition:BasicInequalityS}
Suppose $(W,W) \mapsto (W^\prime,W^{\prime\prime})$ for some set of binary-input channels.
Then,
\begin{gather}
Z(W^{\prime\prime})  = Z(W)^2, \label{eqn:BasicInequality1S}\\
Z(W^\prime)  \le 2Z(W) - Z(W)^2, \label{eqn:BasicInequality2S}\\
Z(W^\prime)  \ge Z(W) \ge Z(W^{\prime\prime}).\label{eqn:BasicInequality3S}
\end{gather}
Equality holds in \eqref{eqn:BasicInequality2S} iff $W$ is a BEC.
We have $Z(W^{\prime})= Z(W^{\prime\prime})$ iff $Z(W)$ equals 0 or 1, or equivalently, iff $I(W)$ equals 1 or 0.
\end{Proposition}

This result shows that reliability can only improve under a single-step channel transform in the sense that
\begin{align}
Z(W^\prime) + Z(W^{\prime\prime}) & \le 2 \, Z(W) \label{eqn:BasicInequality4S}
\end{align}
with equality iff $W$ is a BEC.

Since the BEC plays a special role w.r.t. extremal behavior of reliability, it deserves special attention.
\begin{Proposition}\label{Prop:BECRecursion}
Consider the channel transformation $(W,W) \mapsto (W^\prime,W^{\prime\prime})$.
If $W$ is a BEC with some erasure probability $\epsilon$, then the channels $W^\prime$ and $W^{\prime\prime}$ are BECs with erasure probabilities $2\epsilon - \epsilon^2$ and $\epsilon^2$, respectively.
Conversely, if $W^{\prime}$ or $W^{\prime\prime}$ is a BEC, then $W$ is BEC.
\end{Proposition}

\subsection{Rate and reliability for $W_N^{(i)}$}\label{sec:splitWNi}

We now return to the context at the end of Sect.~\ref{sec:channelSplitting}.
\begin{Proposition}\label{prop:rateConservation}
For any B-DMC $W$, $N=2^n$, $n\ge 0$, $1\le i \le N$, the transformation $(W_N^{(i)},W_N^{(i)}) \mapsto (W_{2N}^{(2i-1)},W_{2N}^{(2i)})$ is rate-preserving and reliability-improving in the sense that
\begin{gather}
I(W_{2N}^{(2i-1)})+I(W_{2N}^{(2i)})  = 2\,I(W_{N}^{(i)}), \label{eqn:conservation1}\\
Z(W_{2N}^{(2i-1)}) + Z(W_{2N}^{(2i)})  \le 2\,Z(W_{N}^{(i)}),\label{eqn:BasicInequality1}
\end{gather}
with equality in \eqref{eqn:BasicInequality1} iff $W$ is a BEC.
Channel splitting moves the rate and reliability away from the center in the sense that
\begin{gather}
I(W_{2N}^{(2i-1)}) \le I(W_N^{(i)}) \le I(W_{2N}^{(2i)}), \label{eqn:FixedPoint}\\
Z(W_{2N}^{(2i-1)}) \ge Z(W_N^{(i)}) \ge Z(W_{2N}^{(2i)}), \label{eqn:BasicInequality2}
\end{gather}
with equality in \eqref{eqn:FixedPoint} and \eqref{eqn:BasicInequality2} iff $I(W)$ equals 0 or 1.
The reliability terms further satisfy
\begin{align}
Z(W_{2N}^{(2i-1)}) & \le  2Z(W_N^{(i)}) - Z(W_N^{(i)})^2, \label{eqn:Zodd}\\
Z(W_{2N}^{(2i)}) & = Z(W_N^{(i)})^2,\label{eqn:Zeven}
\end{align}
with equality in \eqref{eqn:Zodd} iff $W$ is a BEC.
The cumulative rate and reliability satisfy
\begin{align}
\sum_{i=1}^{N} I(W_{N}^{(i)}) & = N I(W), \label{eqn:conservation2}\\
\sum_{i=1}^{N} Z(W_{N}^{(i)}) & \le N Z(W),\label{eqn:BasicInequality4}
\end{align}
with equality in \eqref{eqn:BasicInequality4} iff $W$ is a BEC. \QEDopen
\end{Proposition}

This result follows from Prop.~\ref{prop:rateConservationS} and Prop.~\ref{Proposition:BasicInequalityS} as a special case and no separate proof is needed.
The cumulative relations \eqref{eqn:conservation2} and \eqref{eqn:BasicInequality4} follow by repeated application of \eqref{eqn:conservation1} and \eqref{eqn:BasicInequality1}, respectively.
The conditions for equality in Prop.~\ref{prop:rateConservationS} are stated in terms of $W$ rather than $W_N^{(i)}$; this is possible because:
(i)  by Prop.~\ref{prop:rateConservationS}, $I(W)\in \{0,1\}$ iff $I(W_N^{(i)})\in \{0,1\}$; and (ii) $W$ is a BEC iff $W_N^{(i)}$ is a BEC, which follows from Prop.~\ref{Prop:BECRecursion} by induction.

For the special case that $W$ is a BEC with an erasure probability $\epsilon$, it follows from Prop.~\ref{prop:rateConservationS} and Prop.~\ref{Prop:BECRecursion} that the parameters $\{Z(W_N^{(i)})\}$ can be computed through the recursion
\begin{align}\label{BECrecursion}
\begin{split}
Z(W_{N}^{(2j-1)}) & = 2\,Z(W_{N/2}^{(j)})  -  Z(W_{N/2}^{(j)})^2,\\
Z(W_{N}^{(2j)}) & = Z(W_{N/2}^{(j)})^2,
\end{split}
\end{align}
with $Z(W_1^{(1)})=\epsilon$.
The parameter $Z(W_N^{(i)})$ equals the erasure probability of the channel $W_N^{(i)}$.
The recursive relations \eqref{eqn:IBEC} follow from \eqref{BECrecursion} by the fact that $I(W_N^{(i)})=1-Z(W_N^{(i)})$ for $W$ a BEC.

\section{Channel polarization}\label{sec:PeProposedScheme}

We prove the main results on channel polarization in this section. The analysis is based on the recursive relationships depicted in Fig.~5; however, it will be more convenient to re-sketch Fig.~5 as a binary tree as shown in Fig.~6.
The root node of the tree is associated with the channel $W$.
The root $W$ gives birth to an upper channel $W_2^{(1)}$ and a lower channel $W_2^{(2)}$,
which are associated with the two nodes at level 1.
The channel $W_2^{(1)}$ in turn gives birth to the channels $W_4^{(1)}$ and $W_4^{(2)}$,
and so on.
The channel $W_{2^n}^{(i)}$ is located at level $n$ of the tree at node number $i$ counting from the top.

There is a natural indexing of nodes of the tree in Fig.~6 by bit sequences.
The root node is indexed with the null sequence.
The upper node at level 1 is indexed with 0 and the lower node with 1.
Given a node at level $n$ with index $b_1b_2\cdots b_n$, the upper node emanating from
it has the label $b_1b_2 \cdots b_n 0$ and the lower node $b_1b_2\cdots b_n1$.
According to this labeling, the channel $W_{2^n}^{(i)}$ is situated at the node $b_1b_2\cdots b_n$ with $i= 1+ \sum_{j=1}^n b_j 2^{n-j}$. We denote the channel $W_{2^n}^{(i)}$ located at node $b_1b_2\cdots b_n$ alternatively as $W_{b_1\ldots b_n}$.

\input{fig6}

We define a random tree process, denoted $\{K_n;n\ge 0\}$, in connection with Fig.~6.
The process begins at the root of the tree with $K_0=W$.
For any $n\ge 0$, given that $K_n=W_{b_1\cdots b_n}$,
$K_{n+1}$ equals $W_{b_1\cdots b_n0}$ or $W_{b_1\cdots b_n1}$ with probability 1/2 each.
Thus, the path taken by $\{K_n\}$ through the channel tree may be thought of as being driven by
a sequence of i.i.d. Bernoulli RVs $\{B_n;n=1,2,\ldots\}$ where $B_n$ equals 0 or 1 with equal probability.
Given that $B_1,\ldots, B_n$ has taken on a sample value $b_1,\ldots,b_n$,
the random channel process takes the value $K_n = W_{b_1 \cdots b_n}$.
In order to keep track of the rate and reliability
parameters of the random sequence of channels $K_n$, we define the random processes
$I_n = I(K_n)$ and $Z_n = Z(K_n)$.

For a more precise formulation of the problem, we consider the probability space $(\Omega,{\mathscr F},P)$ where $\Omega$ is the space of all binary sequences $(b_1,b_2,\ldots ) \in \{0,1\}^\infty$,
${\mathscr F}$ is the Borel field (BF) generated by the {\sl cylinder sets\/} $S(b_1,\ldots,b_n) \defn \{\omega \in \Omega: \omega_1 = b_1,\ldots, \omega_n = b_n\}$, $n\ge 1$, $b_1,\ldots, b_n \in \{0,1\}$, and $P$ is the probability measure defined on ${\mathscr F}$ such that $P(S(b_1,\ldots,b_n)) = 1/2^n$.
For each $n\ge 1$, we define ${\mathscr F}_n$ as the BF generated by the cylinder sets $S(b_1,\ldots,b_i)$, $1\le i\le n$, $b_1,\ldots, b_i \in \{0,1\}$. We define ${\mathscr F}_0$ as the trivial BF consisting of the null set and $\Omega$ only.
Clearly, ${\mathscr F}_0\subset {\mathscr F}_1 \subset \cdots \subset {\mathscr F}$.

The random processes described above can now be formally defined as follows.
For ${\mathbf \omega} =(\omega_1,\omega_2,\ldots)\in \Omega$ and $n\ge 1$, define
$B_n({\mathbf \omega}) = \omega_n$, $K_n({\mathbf \omega}) = W_{\omega_1\cdots \omega_n}$,
$I_n({\mathbf \omega}) = I(K_n({\mathbf \omega}))$, and
$Z_n({\mathbf \omega}) = Z(K_n({\mathbf \omega}))$.
For $n=0$, define $K_0=W$, $I_0=I(W)$, $Z_0=Z(W)$.
It is clear that, for any fixed $n\ge 0$, the RVs $B_n$, $K_n$, $I_n$, and $Z_n$ are measurable with respect to the BF ${\mathscr F}_n$.

\subsection{Proof of Theorem~\ref{thm:Polarization}}\label{subsec:convergenceZandI}

We will prove Theorem~\ref{thm:Polarization} by considering the stochastic convergence properties of the random sequences $\{I_n\}$ and $\{Z_n\}$.
\begin{Proposition}\label{prop:In} The sequence of random variables and Borel fields $\{I_n,{\mathscr F}_n; n\ge 0\}$ is a martingale, i.e.,
\begin{align}
& {\mathscr F}_n  \subset {\mathscr F}_{n+1} \text{ and $I_n$ is ${\mathscr F}_n$-measurable}, \label{cond1}\\
& E[|I_n|]  < \infty,  \label{cond2} \\
& I_n  = E[I_{n+1}|{\mathscr F}_n]. \label{cond3}
\end{align}
Furthermore, the sequence $\{I_n;n\ge 0\}$ converges a.e. to a random variable $I_\infty$ such that
$E[I_\infty] = I_0$.
\end{Proposition}

\begin{proof} Condition \eqref{cond1} is true by construction and \eqref{cond2} by the fact that $0 \le I_n \le 1$.
To prove \eqref{cond3}, consider a cylinder set $S(b_1,\ldots,b_n)\in {\mathscr F}_n$ and use Prop.~\ref{prop:rateConservation} to write
\begin{align*}
E[I_{n+1}|S(b_1,\cdots,b_n)] & = \frac{1}{2}I(W_{b_1\cdots b_n0})+\frac{1}{2}I(W_{b_1\cdots b_n1})\\
& = I(W_{b_1\cdots b_n}).
\end{align*}
Since $I(W_{b_1\cdots b_n})$ is the value of $I_n$ on $S(b_1,\ldots,b_n)$, \eqref{cond3} follows.
This completes the proof that $\{I_n,{\mathscr F}_n\}$ is a martingale.
Since $\{I_n,{\mathscr F}_n\}$ is a uniformly integrable martingale,
by general convergence results about such martingales (see, e.g., \cite[Theorem 9.4.6]{Chung}), the claim about $I_\infty$ follows.
\end{proof}

It should not be surprising that the limit RV $I_\infty$ takes values a.e. in $\{0,1\}$, which is the set of fixed points of $I(W)$ under the transformation $(W,W)\mapsto (W_2^{(1)},W_2^{(2)})$, as determined by the condition for equality in \eqref{eqn:FixedPointS}.
For a rigorous proof of this statement, we take an indirect approach and bring the process $\{Z_n;n\ge 0\}$ also into the picture.

\begin{Proposition}\label{prop:Zn} The sequence of random variables and Borel fields $\{Z_n,{\mathscr F}_n;n\ge 0\}$ is a supermartingale, i.e.,
\begin{align}
& {\mathscr F}_n  \subset {\mathscr F}_{n+1} \text{ and $Z_n$ is ${\mathscr F}_n$-measurable}, \label{scond1}\\
& E[|Z_n|]  < \infty  \label{scond2}, \\
& Z_n  \ge E[Z_{n+1}|{\mathscr F}_n]. \label{scond3}
\end{align}
Furthermore, the sequence $\{Z_n;n\ge 0\}$ converges a.e. to a random variable $Z_\infty$ which takes values a.e. in $\{0,1\}$.
\end{Proposition}

\begin{proof}
Conditions \eqref{scond1} and \eqref{scond2} are clearly satisfied.
To verify \eqref{scond3}, consider a cylinder set $S(b_1,\ldots,b_n)\in {\mathscr F}_n$ and
use Prop.~\ref{prop:rateConservation} to write
\begin{align*}
E[Z_{n+1}|S(b_1,\ldots, b_n)] & = \frac{1}{2}Z(W_{b_1\cdots b_n0})+\frac{1}{2}Z(W_{b_1\cdots b_n1})\\
& \le Z(W_{b_1\cdots b_n}).
\end{align*}
Since $Z(W_{b_1\cdots b_n})$ is the value of $Z_n$ on $S(b_1,\ldots,b_n)$, \eqref{scond3} follows.
This completes the proof that $\{Z_n,{\mathscr F}_n\}$ is a supermartingale.
For the second claim, observe that the supermartingale $\{Z_n,{\mathscr F}_n\}$ is uniformly integrable; hence, it converges a.e. and in ${\mathcal L}^1$ to a RV $Z_\infty$ such that $E[|Z_n- Z_\infty|] \to 0$ (see, e.g., \cite[Theorem 9.4.5]{Chung}).
It follows that $E[|Z_{n+1}-Z_{n}|] \to 0$.
But, by Prop.~\ref{prop:rateConservation},
$Z_{n+1}=Z_n^2$ with probability 1/2;
hence, $E[|Z_{n+1}-Z_n|] \ge (1/2)E[Z_n(1-Z_n)]\ge 0$.
Thus, $E[Z_n(1-Z_n)] \to 0$, which implies $E[Z_\infty(1-Z_\infty)] = 0$.
This, in turn, means that $Z_\infty$ equals 0 or 1 a.e.
\end{proof}

\begin{Proposition}\label{prop:Iinfty} The limit RV $I_\infty$ takes values a.e. in the set $\{0,1\}$: $P(I_\infty = 1) = I_0$ and $P(I_\infty = 0) = 1-I_0$.
\end{Proposition}

\begin{proof}
The fact that $Z_\infty$ equals 0 or 1 a.e., combined with Prop.~\ref{rateReliability}, implies that $I_\infty=1-Z_\infty$ a.e. Since $E[I_\infty] = I_0$, the rest of the claim follows.
\end{proof}

As a corollary to Prop.~\ref{prop:Iinfty}, we can conclude that, as $N$ tends to infinity, the symmetric capacity terms $\{I(W_N^{(i)}:1\le i\le N\}$ cluster around 0 and 1, except for a vanishing fraction. This completes the proof of Theorem~\ref{thm:Polarization}.

It is interesting that the above discussion gives a new interpretation to $I_0=I(W)$ as the
probability that the random process $\{Z_n;n\ge 0\}$ converges to zero.
We may use this to strengthen the lower bound in \eqref{Eq:rateReliability}. (This stronger form is given as a side result and will not be used in the sequel.)

\begin{Proposition}\label{prop:ZnBoundBEC} For any B-DMC $W$, we have $I(W)+Z(W) \ge 1$ with equality
iff $W$ is a BEC.
\end{Proposition}

This result can be interpreted as saying that, among all B-DMCs $W$, the BEC presents the most favorable rate-reliability trade-off: it minimizes $Z(W)$ (maximizes reliability) among all channels with a given symmetric capacity $I(W)$; equivalently, it minimizes $I(W)$ required to achieve a given level of reliability $Z(W)$.

\begin{proof}
Consider two channels $W$ and $W^\prime$ with $Z(W)=Z(W^\prime)\defn z_0$.
Suppose that $W^\prime$ is a BEC. Then, $W^\prime$ has erasure probability $z_0$ and $I(W^\prime)=1-z_0$.
Consider the random processes $\{Z_n\}$ and $\{Z^\prime_n\}$ corresponding to $W$ and $W^\prime$, respectively.
By the condition for equality in \eqref{eqn:Zodd}, the process $\{Z_n\}$ is stochastically
dominated by $\{Z^\prime_n\}$ in the sense that $P(Z_n \le z) \ge P(Z^\prime_n \le z)$ for all $n\ge 1$,
$0\le z\le 1$.
Thus, the probability of $\{Z_n\}$ converging to zero is lower-bounded by the probability that
$\{Z^\prime_n\}$ converges to zero, i.e., $I(W) \ge I(W^\prime).$
This implies $I(W) + Z(W) \ge 1$.
\end{proof}

\subsection{Proof of Theorem~\ref{ZnFastConvergence}}\label{sec:ZnFastConvergence}

We will now prove Theorem~\ref{ZnFastConvergence}, which strengthens the above polarization results by specifying a rate of polarization.
Consider the probability space $(\Omega,{\mathscr F},P)$.
For $\omega \in \Omega$, $i\ge 0$, by Prop.~\ref{prop:rateConservation}, we have $Z_{i+1}(\omega)=Z_{i}^2(\omega)$ if $B_{i+1}(\omega)=1$ and $Z_{i+1}(\omega) \le 2Z_{i}(\omega) - Z_{i}(\omega)^2 \le 2Z_{i}(\omega)$ if $B_{i+1}(\omega)=0$.
For $\zeta \ge 0$ and $m\ge 0$, define
 $${\cal T}_m(\zeta) \defn \{\omega \in \Omega: Z_i(\omega) \le \zeta \text{ for all $i\ge m$}\}.$$
For $\omega \in {\cal T}_m(\zeta)$ and $i\ge m$, we have
\begin{align*}
\frac{Z_{i+1}(\omega)}{Z_{i}(\omega)} & \le \begin{cases} 2, & \text{if $B_{i+1}(\omega) =0$}\\
\zeta, & \text{if $B_{i+1}(\omega)=1$}
\end{cases}
\end{align*}
which implies
\begin{align*}
Z_{n}(\omega) & {}\le{}  \zeta \cdot 2^{n-m} \cdot \prod_{i=m+1}^n (\zeta/2)^{B_i(\omega)}, \quad \omega \in T_m(\zeta), \;n> m.
\end{align*}
For $n> m \ge 0$ and $0 <\eta <1/2$, define
$${\cal U}_{m,n}(\eta) \defn \{\omega\in \Omega:\; \sum_{i=m+1}^{n} B_i(\omega) > (1/2 - \eta)(n-m)\}.$$
Then, we have
\begin{align*}
Z_{n}(\omega) & \le  \zeta \cdot \left[2^{\frac{1}{2}+ \eta} \; \zeta^{\frac{1}{2} - \eta}\right]^{n-m}, \quad \omega \in {\cal T}_m(\zeta) \cap {\cal U}_{m,n}(\eta);
\end{align*}
from which, by putting $\zeta_0 \defn 2^{-4}$ and $\eta_0 \defn 1/20$, we obtain
\begin{align}\label{eqn:Zsmall}
Z_{n}(\omega) & \le 2^{-4-5(n-m)/4},\quad \omega \in {\cal T}_m(\zeta_0) \cap {\cal U}_{m,n}(\eta_0).
\end{align}
Now, we show that \eqref{eqn:Zsmall} occurs with sufficiently high probability.
First, we use the following result, which is proved in the Appendix.
\begin{Lemma}\label{ZnConvergence}
For any fixed $\zeta > 0$, $\delta > 0$,
there exists a finite integer $m_0(\zeta,\delta)$ such that
\begin{align*}
P\left[{\cal T}_{m_0}(\zeta)\right] \ge I_0 - \delta/2.
\end{align*}
\end{Lemma}
\vspace*{2mm}
Second, we use Chernoff's bound \cite[p.~531]{Gallager} to write
\begin{align}\label{Gmn}
P\left[{\cal U}_{m,n}( \eta)\right] & \ge 1- 2^{-(n-m) \left[1- H(1/2-\eta)\right]}
\end{align}
where $H$ is the binary entropy function.
Define $n_0(m,\eta,\delta)$ as the smallest $n$ such that the RHS of \eqref{Gmn} is greater than or equal to $1-\delta/2$; it is clear that $n_0(m,\eta,\delta)$ is finite for any $m\ge 0$, $0<\eta <1/2$, and $\delta >0$.
Now, with $m_1=m_1(\delta) \defn m_0(\zeta_0,\delta)$ and $n_1=n_1(\delta) \defn n_0(m_1,\eta_0,\delta)$,
we obtain the desired bound:
\begin{align*}
P[{\cal T}_{m_1}(\zeta_0) \cap {\cal U}_{m_1,n}(\eta_0)] \ge I_0-\delta, \quad n\ge n_1.
\end{align*}
Finally, we tie the above analysis to the claim of Theorem~\ref{ZnFastConvergence}.
Define $c \defn 2^{-4+5m_1/4}$ and
$${\cal V}_n \defn \{\omega \in \Omega:\, Z_n(\omega) \le c\, 2^{-5n/4}\}, \quad n\ge 0;$$
and, note that
$${\cal T}_{m_1}(\zeta_0) \cap {\cal U}_{m_1,n}(\eta_0) \subset {\cal V}_n,\quad n\ge n_1.$$
So, $P({\cal V}_n) \ge I_0-\delta$ for $n\ge n_1$.
On the other hand,
\begin{align*}
P({\cal V}_n) & = \sum_{\omega_1^n \in {\cal X}^n} \frac{1}{2^n} 1\{Z(W_{\omega_1^n})\le c\, 2^{-5n/4} \}\\
& =\frac{1}{N} |{\cal A}_N|
\end{align*}
where ${\cal A}_N \defn \{i\in \{1,\ldots,N\}: Z(W_{N}^{(i)}) \le c\, N^{-5/4}\}$ with $N=2^n$.
We conclude that
$|{\cal A}_N| \ge N (I_0-\delta)$ for $n\ge n_1(\delta)$.
This completes the proof of Theorem~\ref{ZnFastConvergence}.

Given Theorem~\ref{ZnFastConvergence}, it is an easy exercise to show that polar coding can achieve rates approaching $I(W)$, as we will show in the next section.
It is clear from the above proof that Theorem~\ref{ZnFastConvergence} gives only an ad-hoc result on the asymptotic rate of channel polarization; this result is sufficient for proving a capacity theorem for polar coding; however, finding the exact asymptotic rate of polarization remains an important goal for future research.\footnote{A recent result in this direction is discussed in Sect.~\ref{sec:RecentResult}.}

\section{Performance of polar coding}\label{sec:errorAnalysis}

We show in this section that polar coding can achieve the symmetric capacity $I(W)$ of any B-DMC $W$.
The main technical task will be to prove Prop.~\ref{prop:probErrorBound}. We will carry out the analysis over the class of $G_N$-coset codes before specializing the discussion to polar codes.
Recall that individual $G_N$-coset codes are identified by a parameter vector $(N,K,{\cal A},u_{{\cal A}^c})$.
In the analysis, we will fix the parameters $(N,K,{\cal A})$ while keeping $u_{{\cal A}^c}$ free to take any value over ${\cal X}^{N-K}$.
In other words, the analysis will be over the ensemble of $2^{N-K}$ $G_N$-coset codes with a fixed $(N,K,{\cal A})$. The decoder in the system will be the SC decoder described in Sect.~\ref{sec:SCD}.

\subsection{A probabilistic setting for the analysis}\label{sec:ProbabilisticSetting}

Let $({\cal X}^N \times {\cal Y}^N,P)$ be a probability space with the probability assignment
\begin{align}
P(\{(u_1^N,y_1^N)\}) & \defn  2^{-N} W_{N}(y_1^N|u_1^N)
\end{align}
for all $(u_1^N,y_1^N) \in {\cal X}^N \times {\cal Y}^N$.
On this probability space, we define an ensemble of
random vectors $(U_1^N,X_1^N,Y_1^N,\hat{U}_1^N)$ that represent, respectively,
the input to the synthetic channel $W_{N}$,
the input to the product-form channel $W^N$,
the output of $W^N$ (and also of $W_N$), and the decisions by the decoder.
For each sample point $(u_1^N,y_1^N)\in {\cal X}^N \times {\cal Y}^N$, the first three vectors take on the values
$U_1^N(u_1^N,y_1^N) = u_1^N$, $X_1^N(u_1^N,y_1^N) = u_1^N G_N$, and $Y_1^N(u_1^N,y_1^N) = y_1^N$, while the decoder output takes on the value $\hat{U}_1^N(u_1^N,y_1^N)$ whose coordinates are defined recursively as
\begin{align}\label{eqn:DecisionsRV}
\hat{U}_i(u_1^N,y_1^N) & = \begin{cases}
u_i, & i\in {\cal A}^c\\
h_i(y_1^N,\hat{U}_1^{i-1}(u_1^N,y_1^N)), & i\in {\cal A}
\end{cases}
\end{align}
for $i=1,\ldots,N$.

A realization $u_1^N \in {\cal X}^N$ for the input random vector $U_1^N$ corresponds to sending the data vector $u_{\cal A}$ together with the frozen vector $u_{{\cal A}^c}$.
As random vectors, the data part $U_{\cal A}$ and the frozen part $U_{{\cal A}^c}$ are uniformly distributed over their respective ranges and statistically independent.
By treating $U_{{\cal A}^c}$ as a random vector over ${\cal X}^{N-K}$, we obtain a convenient method for analyzing code performance averaged over all codes in the ensemble $(N,K,{\cal A})$.

The main event of interest in the following analysis is the block error event under SC decoding, defined as
\begin{align}\label{defn:ErrorEvent}
{\cal E} & \defn \{(u_1^N,y_1^N)\in {\cal X}^N \times {\cal Y}^N: \hat{U}_{\cal A}(u_1^N,y_1^N) \neq u_{\cal A} \}.
\end{align}
Since the decoder never makes an error on the frozen part of $U_1^N$, i.e., $\hat{U}_{{\cal A}^c}$ equals $U_{{\cal A}^c}$ with probability one, that part has been excluded from the definition of the block error event.

The probability of error terms $P_e(N,K,{\cal A})$ and $P_e(N,K,{\cal A},u_{{\cal A}^c})$ that were defined in Sect.~\ref{sec:CodingTheorems} can be expressed in this probability space as
\begin{align}\label{eqn:PeDefs}
\begin{split}
P_e(N,K,{\cal A}) & = P({\cal E}),\\
P_e(N,K,{\cal A},u_{{\cal A}^c}) & = P({\cal E}\mid  \{U_{{\cal A}^c} = u_{{\cal A}^c}\}),
\end{split}
\end{align}
where $\{U_{{\cal A}^c} = u_{{\cal A}^c}\}$ denotes the event $\{(\tilde{u}_1^N,y_1^N) \in{\cal X}^N\times {\cal Y}^N: \tilde{u}_{{\cal A}^c} = u_{{\cal A}^c}\}$.

\subsection{Proof of Proposition~\ref{prop:probErrorBound}}\label{sec:PeBound}

We may express the block error event as ${\cal E} = \cup_{i\in {\cal A}}{\cal B}_i$ where
\begin{multline}
{\cal B}_i \defn \{(u_1^N,y_1^N)\in{\cal X}^N\times {\cal Y}^N:  \\
u_1^{\,i-1}= \hat{U}_1^{\,i-1}(u_1^N,y_1^N),\; u_i\neq \hat{U}_i(u_1^N,y_1^N) \}
\end{multline}
is the event that the first decision error in SC decoding occurs at stage $i$.
We notice that
\begin{multline*}
{\cal B}_i
 = \{(u_1^N,y_1^N)\in{\cal X}^N\times {\cal Y}^N:
 u_1^{i-1} = \hat{U}_1^{i-1}(u_1^N,y_1^N), \\ u_i \neq h_i(y_1^N,\hat{U}_1^{i-1}(u_1^N,y_1^N))\}\\
= \{(u_1^N,y_1^N)\in{\cal X}^N\times {\cal Y}^N:
 u_1^{i-1} = \hat{U}_1^{i-1}(u_1^N,y_1^N), \\ u_i \neq
h_i(y_1^N,u_1^{i-1})\} \\
 \subset  \{(u_1^N,y_1^N)\in{\cal X}^N\times {\cal Y}^N:
u_i \neq h_i(y_1^N,u_1^{i-1})\}  \subset {\cal E}_i
\end{multline*}
where
\begin{multline}
{\cal E}_i  \defn \{ (u_1^N,y_1^N)\in{\cal X}^N\times {\cal Y}^N: W_N^{(i-1)}(y_1^N,u_1^{i-1}\mid u_i )\\
\le W_N^{(i-1)}(y_1^N,u_1^{i-1}\mid u_i\oplus 1) \}\label{defn:Ei}.
\end{multline}
Thus, we have
$${\cal E} \subset \bigcup_{i\in {\cal A}} {\cal E}_i, \qquad P({\cal E}) \le \sum_{i\in {\cal A}} P({\cal E}_i).$$
For an upper bound on $P({\cal E}_i)$, note that
\begin{align}\label{eqn:PeBoundEi}
\begin{split}
P&({\cal E}_i)  = \sum_{u_1^N,y_1^N} \frac{1}{2^N} W_N(y_1^N\mid u_1^N)  1_{{\cal E}_i}(u_1^N,y_1^N)\\
 & \le \sum_{u_1^N,y_1^N} \frac{1}{2^N} W_N(y_1^N\mid u_1^N)  \sqrt{\frac{W_N^{(i)}(y_1^N,u_1^{i-1}|u_i\oplus 1)}{W_N^{(i)}(y_1^N,u_1^{i-1}|u_i)}}\\
 & = Z(W_N^{(i)}).
 \end{split}
\end{align}
We conclude that
\begin{align*}
P({\cal E}) & \le \sum_{i\in {\cal A}} Z(W_N^{(i)}),
\end{align*}
which is equivalent to \eqref{errorprob3}.
This completes the proof of Prop.~\ref{prop:probErrorBound}.
The main coding theorem of the paper now follows readily.

\subsection{Proof of Theorem~\ref{thm:rate}}
By Theorem~\ref{ZnFastConvergence}, for any given rate $R <I(W)$, there exists a sequence of information sets ${\cal A}_N$ with size $|{\cal A}_N| \ge NR$ such that
\begin{align}\label{eqn:PeBnd}
\sum_{i\in {\cal A}_N} Z(W_N^{(i)}) & \le N \,\max_{i\in {\cal A}_N}\{Z(W_N^{(i)})\} = O(N^{-\frac14}).
\end{align}
In particular, the bound \eqref{eqn:PeBnd} holds if ${\cal A}_N$ is chosen in accordance with the polar coding rule because by definition this rule minimizes the sum in \eqref{eqn:PeBnd}.
Combining this fact about the polar coding rule with Prop.~\ref{prop:probErrorBound}, Theorem~\ref{thm:rate} follows.

\subsection{A numerical example}\label{sec:BEC}

Although we have established that polar codes achieve the symmetric capacity, the proofs have been of an asymptotic nature and the exact asymptotic rate of polarization has not been found.
It is of interest to understand how quickly the polarization effect takes hold and what performance can be expected of polar codes under SC decoding in the non-asymptotic regime.
To investigate these, we give here a numerical study.

Let $W$ be a BEC with erasure probability 1/2.
Figure~7 shows the rate vs. reliability trade-off for $W$ using polar codes with block-lengths $N\in \{2^{10},
2^{15},2^{20}\}$.
This figure is obtained by using codes whose information sets are of the form
${\cal A}(\eta) \defn \{i\in\{1,\ldots,N\}: Z(W_{N}^{(i)}) < \eta\}$, where $0 \le \eta\le 1$ is a variable threshold parameter.
There are two sets of three curves in the plot. The solid lines are plots of $R(\eta) \defn |{\cal A}(\eta)|/N$ vs. $B(\eta) \defn \sum_{i\in {\cal A}(\eta)} Z(W_N^{(i)})$.
The dashed lines are plots of $R(\eta)$ vs. $L(\eta) \defn  \max_{i\in {\cal A}(\eta)} \{Z(W_N^{(i)})\}$.
The parameter $\eta$ is varied over a subset of $[0,1]$ to obtain the curves.

\begin{figure}[thb]
\begin{center}
\resizebox{!}{2.5in}{
\includegraphics{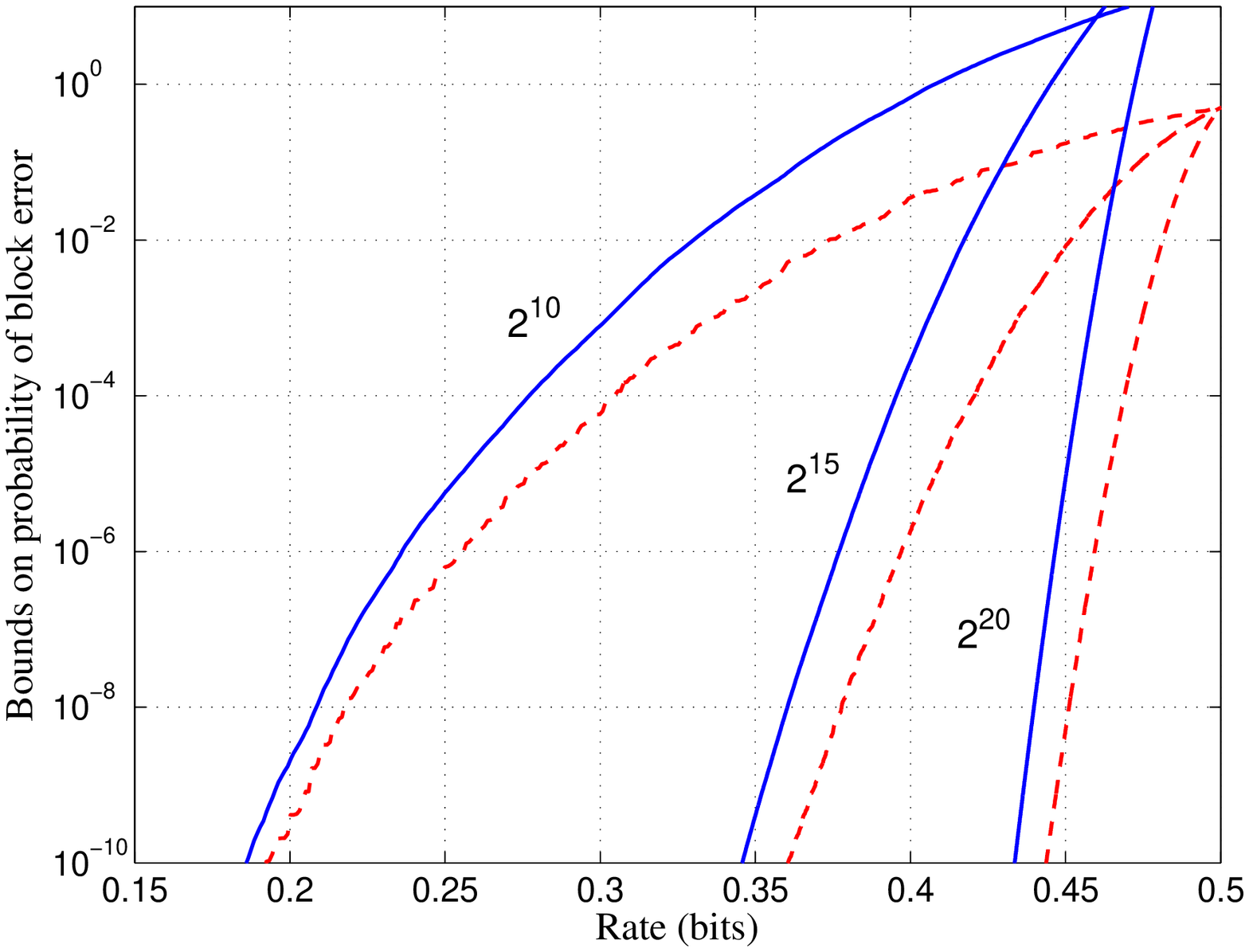}
}
\caption{Rate vs. reliability for polar coding and SC decoding at block-lengths $2^{10}$, $2^{15}$, and $2^{20}$ on a BEC with erasure probability $1/2$.}
\label{fig:rate_vs_reliability}
\end{center}
\end{figure}

The parameter $R(\eta)$ corresponds to the code rate.
The significance of $B(\eta)$ is also clear: it is an upper-bound on $P_e(\eta)$, the probability of  block-error for polar coding at rate $R(\eta)$ under SC decoding. The parameter $L(\eta)$ is intended to serve as a lower bound to $P_e(\eta)$.

This example provides empirical evidence that polar coding achieves channel capacity as the block-length is increased---a fact already established theoretically.
More significantly, the example also shows that the rate of polarization is too slow to make near-capacity polar coding under SC decoding feasible in practice.

\section{Symmetric channels}\label{sec:symmetricChannels}

The main goal of this section is to prove Theorem~\ref{thm:rateSymmetric}, which is a strengthened version of Theorem~\ref{thm:rate} for symmetric channels.

\subsection{Symmetry under channel combining and splitting}
Let $W:{\mathcal X}\to {\mathcal Y}$ be a symmetric B-DMC with ${\cal X}=\{0,1\}$ and ${\mathcal Y}$ arbitrary.
By definition, there exists a a permutation $\pi_1$ on ${\mathcal Y}$ such that (i) $\pi_1^{-1} = \pi_1$ and (ii) $W(y|1)= W(\pi_1(y)|0)$ for all $y\in {\mathcal Y}$.
Let $\pi_0$ be the identity permutation on ${\cal Y}$.
Clearly, the permutations $(\pi_0,\pi_1)$ form an abelian group under function composition.
For a compact notation, we will write $x\cdot y$ to denote $\pi_x(y)$, for $x\in {\mathcal X}$, $y\in {\mathcal Y}$.

Observe that $W(y|x\oplus a) = W(a\cdot y |x)$ for all $a,x\in {\cal X}$, $y\in {\cal Y}$.
This can be verified by exhaustive study of possible cases or by noting that $W(y|x\oplus a) = W((x\oplus a) \cdot y| 0) = W(x\cdot (a\cdot y)|0) = W(a\cdot y|x)$.
Also observe that $W(y|x\oplus a)= W(x\cdot y|a)$ as $\oplus$ is a commutative operation on ${\cal X}$.

For $x_1^N \in {\cal X}^N$, $y_1^N\in {\mathcal Y}^N$, let
\begin{align}
x_1^N\cdot y_1^N & \defn (x_1 \cdot y_1,\ldots ,x_N\cdot y_N).
\end{align}
This associates to each element of ${\mathcal X}^N$ a permutation on ${\mathcal Y}^N$.

\begin{Proposition}\label{prop:symmetryWN}
If a B-DMC $W$ is symmetric, then $W^N$ is also symmetric in the sense that
\begin{align}\label{eqn:symmetryWN}
W^N(y_1^N|x_1^N\oplus a_1^N) & = W^N(x_1^N\cdot y_1^N|a_1^N)
\end{align}
for all $x_1^N,a_1^N\in {\mathcal X}^N$, $y_1^N\in \mathcal{Y}^N$.
\end{Proposition}

The proof is immediate and omitted.

\begin{Proposition}\label{prop:symmetryW}
If a B-DMC $W$ is symmetric, then the channels $W_N$ and $W_N^{(i)}$ are also symmetric in the sense that
\begin{align}
W_{N}(y_1^N \mid u_1^N) & = W_{N}(a_1^NG_N \cdot y_1^N \mid u_1^N \oplus a_1^N),\label{eqn:symmetricChannelProperty1}
\end{align}
\begin{multline}
W_N^{(i)}(y_1^N,u_1^{i-1} \mid u_i)  = \\ W_N^{(i)}(a_1^NG_N\cdot y_1^N,u_1^{i-1}\oplus a_1^{i-1} \mid u_i \oplus a_i)\label{eqn:symmetricChannelProperty2}
\end{multline}
for all $u_1^N,a_1^N \in {\cal X}^N$, $y_1^N\in {\cal Y}^N$, $N=2^n$, $n\ge 0$, $1\le i\le N$.
\end{Proposition}
\vspace{5mm}
\begin{proof}
Let $x_1^N = u_1^N G_N$ and observe that
$W_{N}(y_1^N \mid u_1^N) = \prod_{i=1}^N W(y_i\mid x_i)
= \prod_{i=1}^N W(x_i \cdot y_i\mid 0)
 = W_{N}(x_1^N \cdot y_1^N\mid 0_1^N).$
Now, let $b_1^N =a_1^NG_N$, and use the same reasoning to see that
$W_{N}(b_1^N \cdot y_1^N \mid u_1^N \oplus a_1^N)  =
W_{N}((x_1^N \oplus b_1^N) \cdot (b_1^N \cdot y_1^N)\mid 0_1^N)
=W_{N}(x_1^N \cdot y_1^N\mid 0_1^N).$ This proves the first claim.
To prove the second claim, we use the first result.
\begin{align*}
W_N^{(i)}&(y_1^N,u_1^{i-1} \mid u_i)  = \sum_{u_{i+1}^N} \frac{1}{2^{N-1}} W_{N}(y_1^N\mid u_1^N)\\
&  = \sum_{u_{i+1}^N} \frac{1}{2^{N-1}} W_{N}(a_1^NG_N\cdot y_1^N\mid u_{1}^N \oplus a_1^N)\\
&  = W_{N}(a_1^NG_N \cdot y_1^N, u_1^{i-1}\oplus a_{1}^{i-1}\mid u_i \oplus a_i)
\end{align*}
where we used the fact that the sum over $u_{i+1}^N\in {\cal X}^{N-i}$ can be replaced with a sum over $u_{i+1}^N\oplus a_{i+1}^N$ for any fixed $a_1^N$ since $\{u_{i+1}^N\oplus a_{i+1}^N: u_{i+1}^N\in {\cal X}^{N-i}\} = X^{N-i}$.
\end{proof}

\subsection{Proof of Theorem~\ref{thm:rateSymmetric}}\label{sec:SymmetricPeBound}

We return to the analysis in Sect.~\ref{sec:errorAnalysis} and
consider a code ensemble $(N,K,{\cal A})$ under SC decoding, only this time assuming that $W$ is a symmetric channel.
We first show that the error events $\{{\cal E}_i\}$ defined by \eqref{defn:Ei} have a symmetry property.
\begin{Proposition}
For a symmetric B-DMC $W$, the event ${\cal E}_i$ has the property that
\begin{align}\label{eqn:DecRegionSymmetry}
(u_1^N,y_1^N)\in {\cal E}_i \quad \text{iff}\quad (a_1^N \oplus u_1^N,a_1^NG_N \cdot y_1^N) \in {\cal E}_i
\end{align}
for each $1\le i\le N$, $(u_1^N,y_1^N)\in {\cal X}^N\times {\cal Y}^N$, $a_1^{N}\in {\cal X}^{N}$.
\end{Proposition}
\begin{proof}
This follows directly from the definition of ${\cal E}_i$ by using the symmetry property \eqref{eqn:symmetricChannelProperty2} of the channel $W_N^{(i)}$.
\end{proof}

Now, consider the transmission of a particular source vector $u_{\cal A}$ and a frozen vector $u_{{\cal A}^c}$, jointly forming an input vector $u_1^N$ for the channel $W_N$.
This event is denoted below as $\{U_1^N = u_1^N\}$ instead of the more formal $\{u_1^N\}\times {\cal Y}^N$.
\begin{Corollary}
For a symmetric B-DMC $W$, for each $1\le i\le N$ and $u_1^N\in {\cal X}^N$, the events ${\cal E}_i$ and $\{U_1^N = u_1^N\}$ are independent;
hence, $P({\cal E}_i) = P({\cal E}_i \mid \{U_1^N = u_1^N\})$.
\end{Corollary}

\begin{proof} For $(u_1^N,y_1^N)\in {\cal X}^N\times {\cal Y}^N$ and $x_1^{N}=u_1^N G_N$, we have
\begin{align}
P&({\cal E}_i \mid \{U_1^N = u_1^N\})
= \sum_{y_1^N} W_N(y_1^N \mid u_1^N) \;1_{{\cal E}_i}(u_1^N,y_1^N) \nonumber\\
& = \sum_{y_1^N} W_N(x_1^N\cdot y_1^N \mid 0_1^N) \; 1_{{\cal E}_i}(0_1^N,x_1^N\cdot y_1^N) \label{eqn:PEsymmetry2}\\
& = P({\cal E}_i\mid \{ U_1^N = 0_1^N\}). \label{eqn:PEsymmetry3}
\end{align}
Equality follows in \eqref{eqn:PEsymmetry2} from \eqref{eqn:symmetricChannelProperty1} and \eqref{eqn:DecRegionSymmetry} by taking $a_1^N = u_1^N$, and in \eqref{eqn:PEsymmetry3} from the fact that $\{x_1^N\cdot y_1^N: y_1^N \in{\cal Y}^N \} = {\cal Y}^N$ for any fixed $x_1^N \in {\cal X}^N$.
The rest of the proof is immediate.
\end{proof}

Now, by \eqref{eqn:PeBoundEi}, we have, for all $u_1^N \in {\cal X}^N$,
\begin{align}
P({\cal E}_i \mid \{U_1^N = u_1^N\}) & \le Z(W_N^{(i)})
\end{align}
and, since ${\cal E}\subset \cup_{i\in {\cal A}} \;{\cal E}_i$, we obtain
\begin{align}
P({\cal E} \mid \{U_1^N = u_1^N\}) & \le \sum_{i\in {\cal A}} Z(W_N^{(i)}).
\end{align}
This implies that, for every symmetric B-DMC $W$ and every $(N,K,{\cal A},u_{{\cal A}^c})$ code,
\begin{align}
P_e(N,K,{\cal A},u_{{\cal A}^c}) & = \sum_{u_{\cal A}\in {\cal X}^K} \frac{1}{2^K} P({\cal E}\mid \{U_1^N = u_1^N\}) \nonumber \\
& \le \sum_{i\in {\cal A}} Z(W_N^{(i)}).
\end{align}
This bound on $P_e(N,K,{\cal A},u_{{\cal A}^c})$ is independent of the frozen vector $u_{{\cal A}^c}$.
Theorem~\ref{thm:rateSymmetric} is now obtained by combining Theorem~\ref{ZnFastConvergence} with Prop.~\ref{prop:probErrorBound}, as in the proof of
Theorem~\ref{thm:rate}.

Note that although we have given a bound on $P({\cal E}|\{U_1^N =u_1^N\})$ that is independent of $u_1^N$, we stopped short of claiming that the error event ${\cal E}$ is independent of $U_1^N$ because our decision functions $\{h_i\}$ break ties always in favor of $\hat{u}_i=0$. If this bias were removed by randomization, then ${\cal E}$ would become independent of $U_1^N$.

\subsection{Further symmetries of the channel $W_N^{(i)}$}
We may use the degrees of freedom in the choice of $a_1^N$ in \eqref{eqn:symmetricChannelProperty2} to explore the symmetries inherent in the channel $W_N^{(i)}$.
For a given $(y_1^N,u_1^i)$, we may select $a_1^N$ with $a_1^i= u_1^i$ to obtain
\begin{align}
W_N^{(i)}(y_1^N,u_1^{i-1} \mid u_i) & = W_N^{(i)}(a_1^NG_N\cdot y_1^N, 0_1^{i-1} \mid 0).\label{eqn:symmetricChannelProperty3}
\end{align}
So, if we were to prepare a look-up table for the transition probabilities $\{W_N^{(i)}(y_1^N,u_1^{i-1}\mid u_i): y_1^N\in {\cal Y}^N, u_1^{i}\in {\cal X}^i\}$, it would suffice to store only the subset of probabilities $\{W_N^{(i)}(y_1^N,0_1^{i-1}\mid 0): y_1^N\in {\cal Y}^N\}$.

The size of the look-up table can be reduced further by using the remaining degrees of freedom in the choice of $a_{i+1}^N$.
Let ${\cal X}_{i+1}^N \defn \{a_1^N\in {\cal X}^N: a_1^{i}=0_1^{i}\}$, $1\le i\le N$.
Then, for any $1\le i \le N$, $a_1^N\in {\cal X}_{i+1}^N$, and $y_1^N\in{\cal Y}^N$, we have
\begin{align}\label{eqn:symmetricChannelProperty4}
W_N^{(i)}(y_1^N,0^{i-1}|0) & = W_N^{(i)}(a_1^NG_N\cdot y_1^N,0_1^{i-1}|0)
\end{align}
which follows from \eqref{eqn:symmetricChannelProperty3} by taking $u_1^i=0_1^i$ on the left hand side.

To explore this symmetry further, let ${\cal X}_{i+1}^{N}\cdot y_1^N \defn \{a_1^NG_N \cdot y_1^N: a_1^N\in {\cal X}_{i+1}^{N}\}$.
The set ${\cal X}_{i+1}^{N}\cdot y_1^N$ is the \emph{orbit} of $y_1^N$ under the \emph{action group} ${\cal X}_{i+1}^{N}$.
The orbits ${\cal X}_{i+1}^{N}\cdot y_1^N$ over variation of $y_1^N$ partition the space
${\cal Y}^N$ into equivalence classes.
Let ${\cal Y}_{i+1}^{N}$ be a set formed by taking one representative from each equivalence class.
The output alphabet of the channel $W_N^{(i)}$ can be represented effectively by the set ${\cal Y}_{i+1}^{N}$.

For example, suppose $W$ is a BSC with ${\cal Y}=\{0,1\}$.
Each orbit ${\cal X}_{i+1}^{N}\cdot y_1^N$ has $2^{N-i}$ elements and there are $2^i$ orbits.
In particular, the channel $W_N^{(1)}$ has effectively two outputs, and being symmetric, it has to be a BSC.
This is a great simplification since $W_N^{(1)}$ has an apparent output alphabet size of $2^N$.
Likewise, while $W_N^{(i)}$ has an apparent output alphabet size of $2^{N+i-1}$, due to symmetry, the size shrinks to $2^{i}$.

Further output alphabet size reductions may be possible by exploiting other properties specific to certain B-DMCs. For example, if $W$ is a BEC, the channels $\{W_N^{(i)}\}$ are known to be BECs, each with an effective output alphabet size of three.

The symmetry properties of $\{W_N^{(i)}\}$ help simplify the computation of the channel parameters.
\begin{Proposition}\label{prop:zSymmetric}
For any symmetric B-DMC $W$, the parameters $\{Z(W_N^{(i)})\}$ given by \eqref{defn:ZWNi} can be calculated by the simplified formula
\begin{multline*}
Z(W_N^{(i)}) = 2^{i-1} \sum_{y_1^N\in {\cal Y}_{i+1}^{N}} |{\cal X}_{i+1}^{N}\cdot y_1^N|  \\
\cdot \sqrt{W_N^{(i)}(y_1^N,0_1^{i-1}|0)W_N^{(i)}(y_1^N,0_1^{i-1}|1)}.
\end{multline*}
\end{Proposition}
\vspace*{3mm}

We omit the proof of this result.

For the important example of a BSC, this formula becomes
\begin{multline*}
Z(W_N^{(i)}) = 2^{N-1} \\
 \cdot \sum_{y_1^N\in {\cal Y}_{i+1}^N} \sqrt{W_N^{(i)}(y_1^N,0_1^{i-1}|0)\;W_N^{(i)}(y_1^N,0_1^{i-1}|1)}.
\end{multline*}
This sum  for $Z(W_N^{(i)})$ has $2^{i}$ terms, as compared to $2^{N+i-1}$ terms in \eqref{defn:ZWNi}.

\section{Encoding}\label{sec:Encoder}

In this section, we will consider the encoding of polar codes and prove the part of Theorem~\ref{thm:complexity} about encoding complexity.
We begin by giving explicit algebraic expressions for $G_N$, the generator matrix for polar coding, which so far has been defined only in a schematic form by Fig.~3.
The algebraic forms of $G_N$ naturally point at efficient implementations of the encoding operation $x_1^N = u_1^N G_N$.
In analyzing the encoding operation $G_N$, we exploit its relation to fast transform methods in signal processing;
in particular, we use the bit-indexing idea of \cite{CooleyTukey} to interpret the various permutation operations that are part of $G_N$.

\subsection{Formulas for $G_N$}

In the following, assume $N=2^n$ for some $n\ge 0$.
Let $I_k$ denote the $k$-dimensional identity matrix for any $k\ge 1$.
We begin by translating the recursive definition of $G_N$ as given by Fig.~3 into an algebraic form:
\begin{align*}
G_{N} & = (I_{N/2}\otimes F) \,R_{N} \,(I_2 \otimes G_{N/2}), \quad \text{for $N\ge 2$},
\end{align*}
with $G_1 = I_1$.

\input{fig8}

Either by verifying algebraically that $(I_{N/2} \otimes F) R_N = R_N(F \otimes I_{N/2})$ or by observing that channel combining operation in Fig.~3 can be redrawn equivalently as in Fig.~8, we obtain a second recursive formula
\begin{align}
G_{N} & = R_N(F \otimes I_{N/2}) (I_2\otimes G_{N/2})\nonumber \\
& = R_N(F\otimes G_{N/2}), \label{eqn:recursionG2}
\end{align}
valid for $N\ge 2$.
This form appears more suitable to derive a recursive relationship.
We substitute $G_{N/2}=R_{N/2} (F\otimes G_{N/4})$ back into \eqref{eqn:recursionG2} to obtain
\begin{align}
G_{N} & =
R_{N}
\left(
F\otimes
\left(
R_{N/2}
\left(
F
\otimes
G_{N/4}
\right)
\right)
\right)
\nonumber \\
& =
R_{N}
\left(I_2 \otimes R_{N/2}\right)
\left(
F^{\otimes 2}
\otimes G_{N/4}
\right)\label{eqn:recursionG3}
\end{align}
where \eqref{eqn:recursionG3} is obtained by using the identity $(AC)\otimes (BD) =(A\otimes B)(C\otimes D)$ with $A=I_2$, $B=R_{N/2}$, $C=F$, $D=F\otimes G_{N/4}$.
Repeating this, we obtain
\begin{align}\label{eqn:recursionF3}
G_{N} & = B_N
F^{\otimes n}
\end{align}
where
$B_N \defn R_N (I_2 \otimes R_{N/2}) (I_4 \otimes R_{N/4}) \cdots (I_{N/2} \otimes R_{2})$.
It can seen by simple manipulations that
\begin{align}\label{eqn:RecursionBN}
B_N & = R_N(I_2\otimes B_{N/2}).
\end{align}

We can see that $B_N$ is a permutation matrix by the following induction argument.
Assume that $B_{N/2}$ is a permutation matrix for some $N\ge 4$; this is true for $N=4$ since $B_2=I_2$. Then, $B_N$ is a permutation matrix because it is the product of two permutation matrices, $R_N$ and $I_2 \otimes B_{N/2}$.

In the following, we will say more about the nature of $B_N$ as a permutation.

\subsection{Analysis by bit-indexing}\label{subsec:BitIndexing}
To analyze the encoding operation further, it will be convenient to index vectors and matrices with bit sequences.
Given a vector $a_1^N$ with length $N=2^n$ for some $n\ge 0$, we denote its $i$th element, $a_i$, $1\le i \le N$, alternatively as $a_{b_1 \cdots b_n}$ where $b_1\cdots b_n$ is the binary expansion of the integer $i-1$ in the sense that $i= 1+\sum_{j=1}^{n} b_j 2^{n-j}$.
Likewise, the element $A_{ij}$ of an $N$-by-$N$ matrix $A$ is denoted alternatively as $A_{b_1\cdots b_n,b^\prime_1\cdots b^\prime_n}$
where $b_1 \cdots b_n$ and $b^\prime_1\cdots b^\prime_n$ are the binary representations of $i-1$ and $j-1$, respectively.
Using this convention, it can be readily verified that the product $C= A\otimes B$ of a $2^n$-by-$2^n$ matrix $A$ and a $2^m$-by-$2^m$ matrix $B$ has elements $C_{b_{1}\cdots b_{n+m},b^\prime_{1}\cdots b^\prime_{n+m}} = A_{b_{1}\cdots b_{n},b^\prime_{1}\cdots b^\prime_{n}} B_{b_{n+1}\cdots b_{n+m},b^\prime_{n+1}\cdots b^\prime_{n+m}}$.
\vspace*{1mm}

We now consider the encoding operation under bit-indexing. First, we observe that the elements of $F$ in bit-indexed form are given by $F_{b,b^\prime} = 1\oplus b^\prime \oplus b b^\prime$ for all $b,b^\prime\in \{0,1\}$.
Thus, $F^{\otimes n}$ has elements
\begin{align}\label{eqn:Fn}
F^{\otimes n}_{b_1\cdots b_n,b^\prime_1\cdots b^\prime_n} & = \prod_{i=1}^n F_{b_i,b_i^\prime} = \prod_{i=1}^n (1\oplus b_i^\prime \oplus b_i b_i^\prime).
\end{align}

Second, the reverse shuffle operator $R_N$ acts on a row vector $u_1^N$ to replace the element in bit-indexed position $b_1\cdots b_n$ with the element in position $b_{2}\cdots b_n b_1$; that is, if $v_1^N = u_1^N R_N$, then $v_{b_1\cdots b_n}= u_{b_{2}\cdots b_n b_1}$ for all $b_1,\ldots,b_n\in \{0,1\}$.
In other words, $R_N$ cyclically rotates the bit-indexes of the elements of a left operand $u_1^N$ to the right by one place.

Third, the matrix $B_{N}$ in \eqref{eqn:recursionF3} can be interpreted as the {\sl bit-reversal\/} operator: if $v_1^{N} = u_1^{N} B_{N}$, then $v_{b_1\cdots b_n} = u_{b_n\cdots b_1}$ for all $b_1,\ldots,b_n\in \{0,1\}$. This statement can be proved by induction using the recursive formula \eqref{eqn:RecursionBN}.
We give the idea of such a proof by an example. Let us assume that $B_4$ is a bit-reversal operator and show that the same is true for $B_8$. Let $u_1^8$ be any vector over $GF(2)$. Using bit-indexing, it can be written as $(u_{000},u_{001},u_{010},u_{011},u_{100},u_{101},u_{110},u_{111})$.
 Since $u_1^8 B_8 = u_1^8 R_8 (I_2 \otimes B_4)$, let us first consider the action of $R_8$ on $u_1^8$.
The reverse shuffle $R_8$ rearranges the elements of $u_1^8$ with respect to odd-even parity of their indices, so $u_1^8R_8$ equals $(u_{000},u_{010},u_{100},u_{110},u_{001},u_{011},u_{101},u_{111})$.
This has two halves, $c_1^4\defn (u_{000},u_{010},u_{100},u_{110})$ and $d_1^4\defn(u_{001},u_{011},u_{101},u_{111})$, corresponding to odd-even index classes.
Notice that $c_{b_1b_2}=u_{b_1b_20}$ and $d_{b_1b_2}=u_{b_1b_21}$ for all $b_1,b_2\in \{0,1\}$. This is to be expected since the reverse shuffle rearranges the indices in increasing order within each odd-even index class.
Next, consider the action of $I_2\otimes B_4$ on $(c_1^4,d_1^4)$. The result is $(c_1^4B_4,d_1^4B_4)$.
By assumption, $B_4$ is a bit-reversal operation, so $c_1^4B_4 = (c_{00},c_{10},c_{01},c_{11})$,
which in turn equals $(u_{000},u_{100},u_{010},u_{110})$.
Likewise, the result of $d_1^4 B_4$ equals $(u_{001},u_{101},u_{011},u_{111})$.
Hence, the overall operation $B_8$ is a bit-reversal operation.

Given the bit-reversal interpretation of $B_N$, it is clear that $B_N$ is a symmetric matrix, so $B_N^T = B_N$.
Since $B_N$ is a permutation, it follows from symmetry that $B_N^{-1} = B_N$.

It is now easy to see that, for any $N$-by-$N$ matrix $A$, the product $C=B_N^TAB_N$ has elements
$C_{b_1\cdots b_n,b_1^\prime\cdots b_n^\prime}=A_{b_n\cdots b_1,b_n^\prime\cdots b_1^\prime}$.
It follows that if $A$ is invariant under bit-reversal, i.e., if $A_{b_1\cdots b_n,b^\prime_1\cdots b^\prime_n} = A_{b_n\cdots b_1,b^\prime_n\cdots b^\prime_1}$ for every $b_1,\ldots,b_n,b_1^\prime,\ldots,b_n^\prime \in \{0,1\}$, then $A=B_N^TAB_N$. Since $B_N^T=B_N^{-1}$, this is equivalent to $B_N A = AB_T$.
Thus, bit-reversal-invariant matrices commute with the bit-reversal operator.

\begin{Proposition}\label{prop:GN}
For any $N=2^n$, $n\ge 1$, the generator matrix $G_N$ is given by $G_N=B_NF^{\otimes n}$ and $G_N=F^{\otimes n}B_N$ where $B_N$ is the bit-reversal permutation.
$G_N$ is a bit-reversal invariant matrix with
\begin{align}\label{eqn:GN}
(G_N)_{b_1\cdots b_n,b^\prime_1\cdots b^\prime_n} & = \prod_{i=1}^n (1\oplus b_i^\prime \oplus b_{n-i} b_i^\prime).
\end{align}
\end{Proposition}

\begin{proof}
$F^{\otimes n}$ commutes with $B_N$ because it is invariant under bit-reversal, which is immediate from \eqref{eqn:Fn}. The statement $G_N=B_NF^{\otimes n}$ was established before; by proving that $F^{\otimes n}$ commutes with $B_N$, we have established the other statement: $G_N=F^{\otimes n}B_N$.
The bit-indexed form \eqref{eqn:GN} follows by applying bit-reversal to \eqref{eqn:Fn}.
\end{proof}

Finally, we give a fact that will be useful in Sect.~\ref{sec:RM}.
\begin{Proposition}\label{prop:HammingWeight}
For any $N=2^n$, $n\ge 0$, $b_1,\ldots,b_n\in \{0,1\}$, the rows of $G_N$ and $F^{\otimes n}$ with index $b_1\cdots b_n$ have the same Hamming weight given by  $2^{w_H(b_1,\ldots,b_n)}$
where
\begin{align}
w_H(b_1,\ldots,b_n) & \defn \sum_{i=1}^n b_i
\end{align}
is the Hamming weight of $(b_1,\ldots, b_n)$.
\end{Proposition}

\begin{proof}
For fixed $b_1,\ldots,b_n$, the sum of the terms $(G_N)_{b_1\cdots b_n,b^\prime_1\cdots b^\prime_n}$ (as integers) over all $b_1^\prime,\ldots,b_n^\prime\in \{0,1\}$ gives the Hamming weight of the row of $G_N$ with index $b_1\cdots b_n$. From the preceding formula for $(G_N)_{b_1\cdots b_n,b^\prime_1\cdots b^\prime_n}$, this sum is easily seen to be $2^{w_H(b_1,\ldots,b_n)}$. 
The proof for $F^{\otimes n}$ is similar.
\end{proof}

\subsection{Encoding complexity}

For complexity estimation, our computational model will be a single processor machine with a random access memory. The complexities expressed will be time complexities. The discussion will be given for an arbitrary $G_N$-coset code with parameters $(N,K,{\cal A},u_{{\cal A}^c})$.

Let $\chi_E(N)$ denote the worst-case encoding complexity over all $(N,K,{\cal A},u_{{\cal A}^c})$ codes with a given block-length $N$.
If we take the complexity of a scalar mod-2 addition as 1 unit and the complexity of the reverse shuffle operation $R_N$ as $N$ units, we see from Fig.~3 that $\chi_E(N) \le N/2 + N + 2 \chi_E(N/2)$.
Starting with an initial value $\chi_E(2) = 3$ (a generous figure), we obtain by induction that $\chi_E(N) \le \frac32 N \log N$ for all $N=2^n$, $n\ge 1$.
Thus, the encoding complexity is $O(N\log N)$.

\input{fig9}

A specific implementation of the encoder using the form $G_N = B_N F^{\otimes n}$ is shown in Fig.~9 for $N=8$.
The input to the circuit is the bit-reversed version of $u_1^8$, i.e., $\tilde{u}_1^8 = u_1^8 B_8$.
The output is given by $x_1^8 = \tilde{u}_1^8 F^{\otimes 3} = u_1^8 G_8$.
In general, the complexity of this implementation is $O(N\log N)$ with $O(N)$ for $B_N$ and $O(N\log N)$ for $F^{\otimes n}$.

An alternative implementation of the encoder would be to apply $u_1^8$ in natural index order at the input of the circuit in Fig.~9. Then, we would obtain $\tilde{x}_1^8 = u_1^8 F^{\otimes 3}$ at the output.
Encoding could be completed by a post bit-reversal operation: $x_1^8 = \tilde{x}_1^8 B_8 = u_1^8 G_8$.

The encoding circuit of Fig.~9 suggests many parallel implementation alternatives for $F^{\otimes n}$: for example, with $N$ processors, one may do a ``column by column'' implementation, and reduce the total latency to $\log N$. Various other trade-offs are possible between latency and hardware complexity.

In an actual implementation of polar codes, it may be preferable to use $F^{\otimes n}$ in place of $B_NF^{\otimes n}$ as the encoder mapping in order to simplify the implementation. In that case, the SC decoder should compensate for this by decoding the elements of the source vector $u_1^N$ in bit-reversed index order.
We have included $B_N$ as part of the encoder in this paper in order to have a SC decoder that decodes $u_1^N$ in the natural index order, which simplified the notation.

\section{Decoding}\label{sec:Decoding}

In this section, we consider the computational complexity of the SC decoding algorithm.
As in the previous section, our computational model will be a single processor machine with a random access memory and the complexities expressed will be time complexities.
Let $\chi_D(N)$ denote the worst-case complexity of SC decoding over all $G_N$-coset codes with a given block-length $N$.
We will show that $\chi_D(N) = O(N\log N)$.

\subsection{A first decoding algorithm}\label{sec:Crude}
Consider SC decoding for an arbitrary $G_N$-coset code with parameter $(N,K,{\cal A},u_{{\cal A}^c})$.
Recall that the source vector $u_1^N$ consists of a random part $u_{\cal A}$ and a frozen part $u_{{\mathcal A}^c}$.
This vector is transmitted across $W_N$ and a channel output $y_1^N$ is obtained with probability $W_N(y_1^N|u_1^N)$.
The SC decoder observes $(y_1^N,u_{{\cal A}^c})$ and generates an estimate $\hat{u}_1^N$ of $u_1^N$.
We may visualize the decoder as consisting of $N$ decision elements (DEs), one for each source element $u_i$; the DEs are activated
in the order 1 to $N$.
If $i\in {\cal A}^c$, the element $u_i$ is known; so, the $i$th DE, when its turn comes, simply sets $\hat{u}_i = u_i$ and sends this result to all succeeding DEs.
If $i\in {\cal A}$, the $i$th DE waits until it has received the previous decisions $\hat{u}_1^{i-1}$, and upon receiving them, computes the likelihood ratio (LR)
\begin{align*}
L_N^{(i)}(y_1^N,\hat{u}_1^{i-1}) \defn \frac{W_N^{(i)}(y_1^N,\hat{u}_1^{i-1}|0)}{W_N^{(i)}(y_1^N,\hat{u}_1^{i-1}|1)}
\end{align*}
and generates its decision as
\begin{align*}
\hat{u}_i & = \begin{cases} 0, & \text{if $L_N^{(i)}(y_1^N,\hat{u}_1^{i-1}) \ge 1$}\\
1, & \text{otherwise}
\end{cases}
\end{align*}
which is then sent to all succeeding DEs.
This is a single-pass algorithm, with no revision of estimates.
The complexity of this algorithm is determined essentially by the complexity of computing the LRs.

A straightforward calculation using the recursive formulas \eqref{eqn:Wrecursion1} and \eqref{eqn:Wrecursion2} gives
\begin{multline}\label{eqn:Lrecursion1}
L_{N}^{(2i-1)}(y_1^{N},\hat{u}_1^{2i-2}) = \\
 \frac{L_{N/2}^{(i)}(y_1^{N/2},\hat{u}_{1,o}^{2i-2} \oplus \hat{u}_{1,e}^{2i-2})\;
L_{N/2}^{(i)}(y_{N/2+1}^{N},\hat{u}_{1,e}^{2i-2}) + 1}
{L_{N/2}^{(i)}(y_1^{N/2},\hat{u}_{1,o}^{2i-2} \oplus \hat{u}_{1,e}^{2i-2}) +
L_{N/2}^{(i)}(y_{N/2+1}^{N},\hat{u}_{1,e}^{2i-2})}
\end{multline}
and
\begin{multline}\label{eqn:Lrecursion2}
L_{N}^{(2i)}(y_1^{N},\hat{u}_1^{2i-1}) = \left[L_{N/2}^{(i)}(y_1^{N/2},\hat{u}_{1,o}^{2i-2} \oplus \hat{u}_{1,e}^{2i-2})\right]^{1-2\hat{u}_{2i-1}} \\
\cdot L_{N/2}^{(i)}(y_{N/2+1}^{N},\hat{u}_{1,e}^{2i-2}).
\end{multline}
Thus, the calculation of an LR at length $N$ is reduced to the calculation of two LRs at length $N/2$.
This recursion can be continued down to block-length 1,
at which point the LRs have the form $L_1^{(1)}(y_i)=W(y_i|0)/W(y_i|1)$ and can be computed directly.

To estimate the complexity of LR calculations, let $\chi_L(k)$, $k\in \{N,N/2,N/4,\ldots,1\}$, denote the worst-case complexity of computing $L_k^{(i)}(y_1^k,v_1^{i-1})$ over $i\in [1,k]$ and $(y_1^k,v_1^{i-1})\in {\mathcal Y}^k\times {\mathcal X}^{i-1}$.
From the recursive LR formulas, we have the complexity bound
\begin{align}
\chi_L(k) & \le 2 \chi_L(k/2) + \alpha
\end{align}
where $\alpha$ is the worst-case complexity of assembling two LRs at length $k/2$ into an LR at length $k$.
Taking $\chi_L^{(1)}(1)$ as 1 unit, we obtain the bound
\begin{align}
\chi_L(N) & \le (1+\alpha) N = O(N).
\end{align}

The overall decoder complexity can now be bounded as $\chi_D(N) \le K \chi_L(N)\le N \chi_L(N) = O(N^2)$.
This complexity corresponds to a decoder whose DEs do their LR calculations privately, without sharing any partial results with each other.
It turns out, if the DEs pool their scratch-pad results, a more efficient decoder implementation is possible with overall complexity $O(N\log N)$, as we will show next.

\subsection{Refinement of the decoding algorithm}\label{sec:Decoder}

We now consider a decoder that computes the full set of LRs, $\{L_N^{(i)}(y_1^N,\hat{u}_1^{i-1}): 1\le i\le N\}$.
The previous decoder could skip the calculation of $L_N^{(i)}(y_1^N,\hat{u}_1^{i-1})$ for $i\in {\cal A}^c$; but now we do not allow this.
The decisions $\{\hat{u}_i: 1\le i\le N\}$ are made in exactly the same manner as before; in particular, if $i\in {\cal A}^c$, the decision $\hat{u}_i$ is set to the known frozen value $u_i$, regardless of $L_N^{(i)}(y_1^N,\hat{u}_1^{i-1})$.

To see where the computational savings will come from, we inspect \eqref{eqn:Lrecursion1} and \eqref{eqn:Lrecursion2} and note that each LR value in the pair
$$(L_N^{(2i-1)}(y_1^N,\hat{u}_1^{2i-2}),L_N^{(2i)}(y_1^N,\hat{u}_1^{2i-1}))$$
is assembled from the same pair of LRs:
$$(L_{N/2}^{(i)}(y_1^{N/2},\hat{u}_{1,o}^{2i-2} \oplus \hat{u}_{1,e}^{2i-2}),L_{N/2}^{(i)}(y_{N/2+1}^{N},\hat{u}_{1,e}^{2i-2})).$$
Thus, the calculation of all $N$ LRs at length $N$ requires exactly $N$ LR calculations at length $N/2$.\footnote{Actually, some LR calculations at length $N/2$ may be avoided if, by chance, some duplications occur, but we will disregard this.}
Let us split the $N$ LRs at length $N/2$ into two classes, namely,
\begin{align}\label{eqn:twoclasses}
\begin{split}
&\{L_{N/2}^{(i)}(y_1^{N/2},\hat{u}_{1,o}^{2i-2} \oplus \hat{u}_{1,e}^{2i-2}):1\le i\le N/2\},\\
&\{L_{N/2}^{(i)}(y_{N/2+1}^{N},\hat{u}_{1,e}^{2i-2}):1\le i\le N/2\}.
\end{split}
\end{align}
Let us suppose that we carry out the calculations in each class independently, without trying to exploit any further savings that may come from the sharing of LR values between the two classes.
Then, we have two problems of the same type as the original but at half the size.
Each class in \eqref{eqn:twoclasses} generates a set of $N/2$ LR calculation requests at length $N/4$,
for a total of $N$ requests.
For example, if we let $\hat{v}_1^{N/2} \defn \hat{u}_{1,o}^{N/2} \oplus \hat{u}_{1,e}^{N/2}$, the requests arising from the first class are
\begin{align*}
\begin{split}
&\{L_{N/4}^{(i)}(y_1^{N/4},\hat{v}_{1,o}^{2i-2} \oplus \hat{v}_{1,e}^{2i-2}):1\le i\le N/4\},\\
&\{L_{N/4}^{(i)}(y_{N/4+1}^{N/2},\hat{v}_{1,e}^{2i-2}):1\le i\le N/4\}.
\end{split}
\end{align*}
Using this reasoning inductively across the set of all lengths $\{N,N/2,\ldots,1\}$, we conclude that the total number of LRs that need to be calculated is $N (1+\log N)$.

So far, we have not paid attention to the exact order in which the LR calculations at various block-lengths are carried out.
Although this gave us an accurate count of the total number of LR calculations, for a full description of the algorithm, we need to specify an order.
There are many possibilities for such an order, but to be specific we will use a depth-first algorithm, which is easily described by a small example.

We consider a decoder for a code with parameter $(N,K,{\cal A},u_{{\cal A}^c})$ chosen as $(8,5,\{3,5,6,7,8\},(0,0,0))$.
The computation for the decoder is laid out in a graph as shown in Fig.~10.
There are $N(1+\log N)=32$ nodes in the graph, each responsible for computing an LR request that arises during the course of the algorithm. Starting from the left-side, the first column of nodes correspond to LR requests at length $8$ (decision level), the second column of nodes to requests at length $4$, the third at length 2, and the fourth at length 1 (channel level).

Each node in the graph carries two labels. For example, the third node from the bottom in the third column has the labels $(y_5^6,\hat{u}_2 \oplus \hat{u}_4)$ and $26$; the first label indicates that the LR value to be calculated at this node is $L_8^{(2)}(y_5^6,\hat{u}_2 \oplus \hat{u}_4)$ while the second label indicates that this node will be the 26th node to be activated. The numeric labels, 1 through 32, will be used as quick identifiers in referring to nodes in the graph.

The decoder is visualized as consisting of $N$ DEs situated at the left-most side of the decoder graph.
The node with label $(y_1^8,\hat{u}_1^{i-1})$ is associated with the $i$th DE, $1\le i\le 8$.
The positioning of the DEs in the left-most column follows the bit-reversed index order, as in Fig.~9.

\input{fig10}

Decoding begins with DE 1 activating node 1 for the calculation of $L_8^{(1)}(y_1^8)$.
Node 1 in turn activates node 2 for $L_4^{(1)}(y_1^4)$.
At this point, program control passes to node 2, and node 1 will wait until node 2 delivers the requested LR.
The process continues. Node 2 activates node 3, which activates node 4. Node 4 is a node at the channel level; so it computes $L_1^{(1)}(y_1)$ and passes it to nodes 3 and 23, its left-side neighbors. In general a node will send its computational result to all its left-side neighbors (although this will not be stated explicitly below). Program control will be passed back to the left neighbor from which it was received.

Node 3 still needs data from the right side and activates node 5, which delivers $L_1^{(1)}(y_2)$. Node 3 assembles $L_2^{(1)}(y_1^2)$ from the messages it has received from nodes 4 and 5 and sends it to node 2. Next, node 2 activates node 6, which activates nodes 7 and 8, and returns its result to node 2. Node 2 compiles its response $L_4^{(1)}(y_1^4)$ and sends it to node 1.
Node 1 activates node 9 which calculates $L_4^{(1)}(y_5^8)$ in the same manner as node 2 calculated $L_4^{(1)}(y_1^4)$, and returns the result to node 1.
Node 1 now assembles $L_8^{(1)}(y_1^8)$ and sends it to DE 1.
Since $u_1$ is a frozen node, DE 1 ignores the received LR, declares $\hat{u}_1 = 0$, and passes control to DE 2, located next to node 16.

DE 2 activates node 16 for $L_8^{(2)}(y_1^8,\hat{u}_1)$. Node 16 assembles $L_8^{(2)}(y_1^8,\hat{u}_1)$ from the already-received LRs $L_4^{(1)}(y_1^4)$ and $L_4^{(1)}(y_5^8)$, and returns its response without activating any node.
DE 2 ignores the returned LR since $u_2$ is frozen, announces $\hat{u}_2=0$, and passes control to DE 3.

DE 3 activates node 17 for $L_8^{(3)}(y_1^8,\hat{u}_1^2)$.
This triggers LR requests at nodes 18 and 19, but no further.
The bit $u_3$ is not frozen; so, the decision $\hat{u}_3$ is made in accordance with $L_8^{(3)}(y_1^8,\hat{u}_1^2)$, and control is passed to DE 4.
DE 4 activates node 20 for $L_8^{(4)}(y_1^8,\hat{u}_1^3)$, which is readily assembled and returned.
The algorithm continues in this manner until finally DE 8 receives $L_8^{(7)}(y_1^8,\hat{u}_1^7)$ and decides $\hat{u}_8$.

There are a number of observations that can be made by looking at this example that should provide further insight into the general decoding algorithm.
First, notice that the computation of $L_8^{(1)}(y_1^8)$ is carried out in a subtree rooted at node 1, consisting of paths going from left to right, and spanning  all nodes at the channel level. This subtree splits into two disjoint subtrees, namely, the subtree rooted at node 2 for the calculation of $L_4^{(1)}(y_1^4)$ and the subtree rooted at node 9 for the calculation of $L_4^{(1)}(y_5^8)$.
Since the two subtrees are disjoint, the corresponding calculations can be carried out independently (even in parallel if there are multiple processors).
This splitting of computational subtrees into disjoint subtrees holds for all nodes in the graph (except those at the channel level), making it possible to implement the decoder with a high degree of parallelism.

Second, we notice that the decoder graph consists of \emph{butterflies} (2-by-2 complete bipartite graphs) that tie together adjacent levels of the graph.
For example, nodes 9, 19, 10, and 13 form a butterfly. The computational subtrees rooted at nodes 9 and 19 split into a single pair of computational subtrees, one rooted at node 10, the other at node 13.
Also note that among the four nodes of a butterfly, the upper-left node is always the first node to be activated by the above depth-first algorithm and the lower-left node always the last one. The upper-right and lower-right nodes are activated by the upper-left node and they may be activated in any order or even in parallel. The algorithm we specified always activated the upper-right node first, but this choice was arbitrary. When the lower-left node is activated, it finds the LRs from its right neighbors ready for assembly.  The upper-left node assembles the LRs it receives from the right side as in formula \eqref{eqn:Lrecursion1}, the lower-left node as in \eqref{eqn:Lrecursion2}.
These formulas show that the butterfly patterns impose a constraint on the completion time of LR calculations: in any given butterfly, the lower-left node needs to wait for the result of the upper-left node which in turn needs to wait for the results of the right-side nodes.

Variants of the decoder are possible in which the nodal computations are scheduled differently. In the ``left-to-right'' implementation given above, nodes waited to be activated. However, it is possible to have a ``right-to-left'' implementation in which each node starts its computation autonomously as soon as its right-side neighbors finish their calculations; this allows exploiting parallelism in computations to the maximum possible extent.

For example, in such a fully-parallel implementation for the case in Fig.~10, all eight nodes at the channel-level start calculating their respective LRs in the first time slot following the availability of the channel output vector $y_1^8$.
In the second time slot, nodes 3, 6, 10, and 13 do their LR calculations in parallel.
Note that this is the maximum degree of parallelism possible in the second time slot. Node 23, for example, cannot calculate $L_N^{(2)}(y_1^2,\hat{u}_1\oplus \hat{u}_2\oplus \hat{u}_3\oplus \hat{u}_4)$ in this slot, because $\hat{u}_1\oplus \hat{u}_2\oplus \hat{u}_3\oplus \hat{u}_4$ is not yet available; it has to wait until decisions $\hat{u}_1$, $\hat{u}_2$, $\hat{u}_3$, $\hat{u}_4$ are announced by the corresponding DEs.
In the third time slot, nodes 2 and 9 do their calculations. In time slot 4, the first decision $\hat{u}_1$ is made at node 1 and broadcast to all nodes across the graph (or at least to those that need it).
In slot 5, node 16 calculates $\hat{u}_2$ and broadcasts it. In slot 6, nodes 18 and 19 do their calculations. This process continues until time slot 15 when node 32 decides $\hat{u}_8$.
It can be shown that, in general, this fully-parallel decoder implementation has a latency of $2N-1$ time slots for a code of block-length $N$.

\section{Code construction}\label{sec:Construction}

The input to a polar code construction algorithm is a triple $(W,N,K)$ where $W$ is the B-DMC on which the code will be used, $N$ is the code block-length, and $K$ is the dimensionality of the code. The output of the algorithm is an information set ${\cal A}\subset \{1,\ldots,N\}$ of size $K$ such that $\sum_{i\in {\cal A}} Z(W_N^{(i)})$ is as small as possible.
We exclude the search for a good frozen vector $u_{{\cal A}^c}$ from the code construction problem because the problem is already difficult enough.
Recall that, for symmetric channels, the code performance is not affected by the choice of $u_{{\cal A}^c}$.

In principle, the code construction problem can be solved by computing all the parameters $\{Z(W_N^{(i)}):1\le i\le N\}$ and sorting them; unfortunately, we do not have an efficient algorithm for doing this.
For symmetric channels, some computational shortcuts are available, as we showed by Prop.~\ref{prop:zSymmetric}, but these shortcuts have not yielded an efficient algorithm, either. One exception to all this is the BEC for which the parameters $\{Z(W_N^{(i)})\}$ can all be calculated in time $\bigoh(N)$ thanks to the recursive formulas \eqref{BECrecursion}.

Since exact code construction appears too complex, it makes sense to look for approximate constructions based on estimates of the parameters $\{Z(W_N^{(i)})\}$.
To that end, it is preferable to pose the exact code construction problem as a decision problem:
Given a threshold $\gamma \in [0,1]$ and an index $i\in\{1,\ldots,N\}$, decide whether $i\in {\cal A}_\gamma$
where
\begin{align*}
{\cal A}_{\gamma} \defn \{i\in \{1,\ldots,N\}: Z(W_N^{(i)}) < \gamma\}.
\end{align*}
Any algorithm for solving this decision problem can be used to solve the code construction problem.
We can simply run the algorithm with various settings for $\gamma$ until we obtain an information set ${\cal A}_\gamma$ of the desired size $K$.

Approximate code construction algorithms can be proposed based on statistically reliable and efficient methods for estimating whether $i\in {\cal A}_\gamma$ for any given pair $(i,\gamma)$.
The estimation problem can be approached by noting that, as we have implicitly shown in \eqref{eqn:PeBoundEi}, the parameter $Z(W_N^{(i)})$ is the expectation of the RV
\begin{align}\label{eqn:sqrtLN}
\sqrt{\frac{W_N^{(i)}(Y_1^N,U_1^{i-1}|U_i\oplus 1)}{W_N^{(i)}(Y_1^N,U_1^{i-1}|U_i)}}
\end{align}
where $(U_1^N,Y_1^N)$ is sampled from the joint probability assignment $P_{U_1^N,Y_1^N}(u_1^N,y_1^N) \defn 2^{-N} W_N(y_1^N|u_1^N)$.
A Monte-Carlo approach can be taken where samples of $(U_1^N,Y_1^N)$ are generated from the given distribution and the empirical means $\{\hat{Z}(W_N^{(i)})\}$ are calculated.
Given a sample $(u_1^N,y_1^N)$ of $(U_1^N,Y_1^N)$, the sample values of the RVs \eqref{eqn:sqrtLN} can all be computed in complexity $O(N\log N)$. A SC decoder may be used for this computation since the sample values of \eqref{eqn:sqrtLN} are just the square-roots of the decision statistics that the DEs in a SC decoder ordinarily compute. (In applying a SC decoder for this task, the information set ${\cal A}$ should be taken as the null set.)

Statistical algorithms are helped by the polarization phenomenon: for any fixed $\gamma$ and as $N$ grows, it becomes easier to resolve whether $Z(W_N^{(i)}) < \gamma$ because an ever growing fraction of the parameters $\{Z(W_N^{(i)})\}$ tend to cluster around 0 or 1.

It is conceivable that, in an operational system, the estimation of the parameters $\{Z(W_N^{(i)})\}$ is made part of a SC decoding procedure, with continual update of  the information set as more reliable estimates become available.

\section{A note on the RM rule}\label{sec:RM}
In this part, we return to the claim made in Sect.~\ref{subsec:relationPreviousWork} that the RM rule for information set selection leads to asymptotically unreliable codes under SC decoding.

Recall that, for a given $(N,K)$, the RM rule constructs a $G_N$-coset code with parameter $(N,K,{\cal A},u_{{\cal A}^c})$ by prioritizing each index $i\in\{1,\ldots,N\}$ for inclusion in the information set ${\cal A}$ w.r.t. the Hamming weight of the $i$th row of $G_N$. The RM rule sets the frozen bits $u_{{\cal A}^c}$ to zero.
In light of Prop.~\ref{prop:HammingWeight}, the RM rule can be restated in bit-indexed terminology as follows.

{\sl RM rule:\/} For a given $(N,K)$, with $N=2^n$, $n\ge 0$, $0\le K\le N$, choose ${\cal A}$ as follows:
(i) Determine the integer $r$ such that
\begin{align}\label{eqn:RMdimension}
\sum_{k=r}^{n} {n\choose k} \le K < \sum_{k=r-1}^{n} {n\choose k}.
\end{align}
(ii) Put each index $b_1\cdots b_n$ with $w_H(b_1,\ldots,b_n) \ge r$ into ${\cal A}$.
(iii) Put sufficiently many additional indices $b_1\cdots b_n$ with $w_H(b_1,\ldots, b_n) = r-1$ into ${\cal A}$ to complete its size to $K$.

We observe that this rule will select the index
\begin{align*}
0^{n-r}1^r & \defn  \overbrace{0\cdots 0}^{n-r} \overbrace{1 \cdots 1}^{r}
\end{align*}
for inclusion in ${\cal A}$.
This index turns out to be a particularly poor choice, at least for the class of BECs, as we show in the remaining part of this section.

Let us assume that the code constructed by the RM rule is used on a BEC $W$ with some erasure probability $\epsilon > 0$.
We will show that the symmetric capacity $I(W_{0^{n-r}1^r})$ converges to zero for any fixed positive coding rate as the block-length is increased.
For this, we recall the relations \eqref{eqn:IBEC}, which, in bit-indexed channel notation of Sect.~\ref{sec:PeProposedScheme}, can be written as follows.
For any $\ell \ge 1$, $b_1,\ldots,b_\ell \in \{0,1\}$,
\begin{align*}
\begin{split}
I(W_{b_1\cdots b_{\ell}0}) & = I(W_{b_{1} \cdots b_\ell})^2\\
I(W_{b_1\cdots b_{\ell}1}) & = 2I(W_{b_{1} \cdots b_\ell}) - I(W_{b_{1} \cdots b_\ell})^2\\
& \le 2I(W_{b_{1} \cdots b_\ell})
\end{split}
\end{align*}
with initial values $I(W_{0}) = I^2(W)$ and $I(W_1) = 2I(W)-I^2(W)$.
These give the bound
\begin{align}\label{eqn:IBound}
I(W_{0^{n-r}1^r}) & \le 2^{r} (1-\epsilon)^{2^{n-r}}.
\end{align}

Now, consider a sequence of RM codes with a fixed rate $0<R<1$, $N$ increasing to infinity, and $K=\lfloor NR \rfloor$. Let $r(N)$ denote the parameter $r$ in \eqref{eqn:RMdimension} for the code with block-length $N$ in this sequence. Let $n=\log_2(N)$.
A simple asymptotic analysis shows that the ratio $r(N)/n$ must go to $1/2$ as $N$ is increased.
This in turn implies by \eqref{eqn:IBound} that $I(W_{0^{n-r}1^r})$ must go to zero.

Suppose that this sequence of RM codes is decoded using a SC decoder as in Sect.~\ref{sec:SCD} where the decision metric ignores knowledge of frozen bits and instead uses randomization over all possible choices. Then, as $N$ goes to infinity, the SC decoder decision element with index $0^{n-r}1^r$ sees a channel whose capacity goes to zero, while the corresponding element of the input vector $u_1^N$ is assigned 1 bit of information by the RM rule. This means that the RM code sequence is asymptotically unreliable under this type of SC decoding.

We should emphasize that the above result does not say that RM codes are asymptotically bad under {\sl any} SC decoder, nor does it make a claim about the performance of RM codes under other decoding algorithms. (It is interesting that the possibility of RM codes being capacity-achieving codes under ML decoding seems to have received no attention in the literature.)

\section{Concluding remarks}\label{sec:Generalizations}

In this section, we go through the paper to discuss some results further, point out some generalizations, and state some open problems.

\subsection{Rate of polarization}\label{sec:RecentResult}

A major open problem suggested by this paper is to determine how fast a channel polarizes as a function of the block-length parameter $N$.
In recent work \cite{ArikanTelatarRatePolarizationArxiv}, the following result has been obtained in this direction.
\begin{Proposition}
Let $W$ be a B-DMC. For any fixed rate $R <I(W)$ and constant $\beta <\frac12$, there exists a sequence of sets $\{{\cal A}_N\}$ such that
 ${\cal A}_N\subset \{1,\ldots,N\}$, $|{\cal A}_N| \ge NR$, and
\begin{align}
\sum_{i\in {\cal A}_N} Z(W_N^{(i)}) & = o(2^{-N^\beta}).
\end{align}
Conversely, if $R> 0$ and $\beta >\frac12$, then for any sequence of sets $\{{\cal A}_N\}$ with
 ${\cal A}_N\subset \{1,\ldots,N\}$, $|{\cal A}_N| \ge NR$, we have
\begin{align}
\max\{Z(W_N^{(i)}) : i\in {\cal A}_N \}   & = \omega(2^{-N^\beta}).
\end{align}
\end{Proposition}
\vspace*{2mm}
As a corollary, Theorem~\ref{thm:rate} is strengthened as follows.
\begin{Proposition}\label{thm:rate2}
For polar coding on a B-DMC $W$ at any fixed rate $R< I(W)$, and any fixed $\beta < \frac12$,
\begin{align}\label{eqn:PeMainResult2}
P_e(N,R) & =  o(2^{-N^\beta}).
\end{align}
\end{Proposition}
\vspace*{3mm}

This is a vast improvement over the $O(N^{-\frac14})$ bound proved in this paper.
Note that the bound still does not depend on the rate $R$ as long as $R< I(W)$. A problem of theoretical interest is to obtain sharper bounds on $P_e(N,R)$ that show a more explicit dependence on $R$.

Another problem of interest related to polarization is {\sl robustness} against channel parameter variations. A finding in this regard is the following result \cite{Sahai2008}: If a polar code is designed for a B-DMC $W$ but used on some other B-DMC $W^\prime$, then the code will perform at least as well as it would perform on $W$ provided $W$ is a degraded version of $W^\prime$ in the sense of Shannon \cite{Shannon2}. This result gives reason to expect a graceful degradation of polar-coding performance due to errors in channel modeling.

\subsection{Generalizations}

\input{fig11}

The polarization scheme considered in this paper can be generalized as shown in Fig.~11.
In this general form, the channel input alphabet is assumed $q$-ary, ${\cal X}=\{0,1,\ldots,q-1\}$, for some $q\ge 2$.
The construction begins by combining $m$ independent copies of a DMC $W:{\cal X}\to {\cal Y}$ to obtain $W_m$, where $m\ge 2$ is a fixed parameter of the construction.
The general step combines $m$ independent copies of the channel $W_{N/m}$ from the previous step to obtain $W_N$.
In general, the size of the construction is $N=m^n$ after $n$ steps.
The construction is characterized by a {\sl kernel} $F_m: {\cal X}^m\times {\cal R} \to {\cal X}^m$
where ${\cal R}$ is some finite set included in the mapping for \emph{randomization}. The reason for introducing randomization will be discussed shortly.

The vectors $u_1^{N}\in {\cal X}^{N}$ and $y_1^{N}\in {\cal Y}^{N}$ in Fig.~11 denote the input and output vectors of $W_{N}$.
The input vector is first transformed into a vector $s_1^{N}\in {\cal X}^{N}$ by breaking it into $N$ consecutive sub-blocks of length $m$,
namely, $u_1^m,\dots,u_{N-m+1}^{N}$, and passing each sub-block through the transform $F_m$.
Then, a permutation $R_{N}$ sorts the components of $s_1^{N}$ w.r.t. mod-$m$ residue classes of their indices.
The sorter ensures that, for any $1\le k\le m$, the $k$th copy of $W_{N/m}$, counting from the top of the figure, gets as input those components of $s_1^{N}$ whose indices are congruent to $k$ mod-$m$.
For example, $v_1=s_1$, $v_2=s_{m+1}$, $v_{N/m}=s_{(N/m-1)m+1}$,
$v_{N/m+1}=s_{2}$, $v_{N/m+2}=s_{m+2}$, and so on.
The general formula is $v_{kN/m+j} = s_{k+(j-1)m+1}$ for all $0\le k\le (m-1)$, $1\le j\le N/m$.

We regard the randomization parameters $r_1$, \dots, $r_m$ as being chosen at random at the time of code construction, but fixed throughout the operation of the system; the decoder operates with full knowledge of them.
For the binary case considered in this paper, we did not employ any randomization.
Here, randomization has been introduced as part of the general construction because preliminary studies show that it greatly simplifies the analysis of generalized polarization schemes. This subject will be explored further in future work.

Certain additional constraints need to be placed on the kernel $F_m$ to ensure that a polar code can be defined that is suitable for SC decoding in the natural order $u_1$ to $u_N$.
To that end, it is sufficient to restrict $F_m$ to {\sl unidirectional} functions, namely, invertible functions of the form $F_m:(u_1^m ,r) \in {\cal X}^m\times {\cal R} \mapsto
x_1^m\in {\cal X}^m$ such that $x_i=f_i(u_i^m,r)$, for a given set of coordinate functions $f_i:{\cal X}^{m-i+1}\times {\cal R}\to {\cal X}$, $i=1,\ldots,m$.
For a unidirectional $F_m$, the combined channel $W_N$ can be split to channels $\{W_N^{(i)}\}$ in much the same way as in this paper.
The encoding and SC decoding complexities of such a code are both $\bigoh(N\log N)$.

Polar coding can be generalized further in order to overcome the restriction of the block-length $N$ to powers of a given number $m$ by  using a sequence of kernels $F_{m_i}$, $i=1,\ldots,n$, in the code construction. Kernel $F_{m_1}$ combines $m_1$ copies of a given DMC $W$ to create a channel $W_{m_1}$. Kernel $F_{m_2}$ combines $m_2$ copies of $W_{m_1}$ to create a channel $W_{m_1m_2}$, etc., for an overall block-length of $N=\prod_{i=1}^n m_i$.  If all kernels are unidirectional, the combined channel $W_N$ can still be split into channels $W_N^{(i)}$ whose transition probabilities can be expressed by recursive formulas and $\bigoh(N\log N)$ encoding and decoding complexities are maintained.

So far we have considered only combining copies of one DMC $W$. Another direction for generalization of the method is to combine copies of two or more distinct DMCs. For example, the kernel $F$ considered in this paper can be used to combine copies of any two B-DMCs $W$, $W'$.
The investigation of coding advantages that may result from such variations on the basic code construction method is an area for further research.

It is easy to propose variants and generalizations of the basic channel polarization scheme, as we did above; however, it is not clear if we obtain channel polarization under each such variant.
We conjecture that channel polarization is a common phenomenon, which is almost impossible to avoid as long as channels are combined with a sufficient density and mix of connections, whether chosen recursively or at random, provided the coordinatewise splitting of the synthesized vector channel is done according to a suitable SC decoding order.
The study of channel polarization in such generality is an interesting theoretical problem.

\subsection{Iterative decoding of polar codes}
We have seen that polar coding under SC decoding can achieve symmetric channel capacity; however, one needs to use codes with impractically large block lengths.
A question of interest is whether polar coding performance can improve significantly under more powerful decoding algorithms.
The sparseness of the graph representation of $F^{\otimes n}$ makes Gallager's belief propagation (BP) decoding algorithm \cite{GallagerLDPC} applicable to polar codes.
A highly relevant work in this connection is \cite{Forney2001} which proposes BP decoding for RM codes using a factor-graph of $F^{\otimes n}$, as shown in Fig.~12 for $N=8$.
We carried out experimental studies to assess the performance of polar codes under BP decoding, using RM codes under BP decoding as a benchmark \cite{ArikanPolarRM}. The results showed significantly better performance for polar codes. Also, the performance of polar codes under BP decoding was significantly better than their performance under SC decoding.
However, more work needs to be done to assess the potential of polar coding for practical applications.

\input{fig12}

\appendix 

\subsection{Proof of Proposition~\ref{rateReliability}}

The right hand side of \eqref{Eq:rateReliability} equals the channel parameter $E_0(1,Q)$ as defined in Gallager \cite[Section~5.6]{Gallager} with $Q$ taken as the uniform input distribution. (This is the {\sl symmetric cutoff rate\/} of the channel.)
It is well known (and shown in the same section of \cite{Gallager}) that $I(W)\ge E_0(1,Q)$. This proves \eqref{Eq:rateReliability}.

To prove \eqref{Eq:rateReliability4}, for any B-DMC $W:{\cal X}\to {\cal Y}$, define
\begin{align*}
d(W) & \defn \frac{1}{2} \sum_{y\in {\cal Y}} |W(y|0)-W(y|1)|.
\end{align*}
This is the variational distance between the two distributions $W(y|0)$ and $W(y|1)$ over $y\in {\cal Y}$.

\begin{Lemma}\label{lemma:CapBounds}
For any B-DMC $W$,
$I(W) \le d(W).$
\end{Lemma}

\begin{proof}
Let $W$ be an arbitrary B-DMC with output alphabet ${\cal Y}=\{1,\ldots,n\}$
and put $P_i = W(i|0)$, $Q_i = W(i|1)$, $i=1,\ldots,n$.
By definition,
\begin{align*}
I(W) & = \sum_{i=1}^n \frac{1}{2} \left[P_i \log \frac{P_i}{\frac{1}{2} P_i + \frac{1}{2} Q_i }
+ Q_i \log \frac{Q_i}{\frac{1}{2} P_i + \frac{1}{2} Q_i }\right].
\end{align*}
The $i$th bracketed term under the summation is given by
\begin{align*}
f(x) &\defn x \log \frac{x}{x + \delta} + (x+2\delta) \log \frac{x+ 2\delta}{x + \delta}
\end{align*}
where $x =\min\{P_i,Q_i\}$ and $\delta = \frac{1}{2} |P_i-Q_i|$.
We now consider maximizing $f(x)$ over $0\le x\le 1-2\delta$.
We compute
\begin{align*}
\frac{d\,f}{d\,x}  & =  \frac{1}{2}\,\log \frac{\sqrt{x(x+2\delta)}}{(x + \delta)}
\end{align*}
and recognize that $\sqrt{x(x+2\delta)}$ and $(x+\delta)$ are, respectively, the
geometric and arithmetic means of the numbers $x$ and $(x+2\delta)$.
So, $df/dx \le 0$ and $f(x)$ is maximized at $x=0$, giving the inequality
$f(x) \le 2\delta$.
Using this in the expression for $I(W)$, we obtain the claim of the lemma,
\begin{align*}
I(W) & \le \sum_{i=1} \frac{1}{2}\,|P_i - Q_i| = d(W).
\end{align*}
\end{proof}

\begin{Lemma}\label{lemma:pithagorean} For any B-DMC $W$,
$d(W) \le \sqrt{1-Z(W)^2}$.
\end{Lemma}

\begin{proof}
Let $W$ be an arbitrary B-DMC with output alphabet ${\cal Y}=\{1,\ldots,n\}$ and put $P_i = W(i|0)$, $Q_i = W(i|1)$, $i=1,\ldots,n$.
Let $\delta_i \defn \frac{1}{2} |P_i-Q_i|$,
$\delta \defn d(W) = \sum_{i=1}^n \delta_i$, and $R_i \defn (P_i+Q_i)/2$.
Then, we have $Z(W) = \sum_{i=1}^n \sqrt{(R_i -\delta_i)(R_i +\delta_i)}$.
Clearly, $Z(W)$ is upper-bounded by the maximum of $\sum_{i=1}^n \sqrt{R_i^2-\delta_i^2}$ over
$\{\delta_i\}$ subject to the constraints that $0 \le \delta_i \le R_i$, $i=1,\ldots,n$, and
$\sum_{i=1}^n \delta_i =\delta$.
To carry out this maximization, we compute the partial derivatives
of $Z(W)$ with respect to $\delta_i$,
\begin{align*}
\frac{\partial Z}{\partial \delta_i} & = -\,\frac{\delta_i}{\sqrt{R_i^2 -\delta_i^2}}, \qquad
\frac{\partial^2 Z}{\partial \delta_i^2}  = -\,\frac{R_i^2}{\sqrt[3/2]{R_i^2 -\delta_i^2}},
\end{align*}
and observe that $Z(W)$ is a decreasing, concave
 function of $\delta_i$ for each $i$, within the range $0\le \delta_i \le R_i$.
The maximum occurs at the solution of the set of equations
${\partial Z}/{\partial \delta_i} = k$, all $i$,
where $k$ is a constant, i.e., at
$\delta_i = R_i \sqrt{k^2/(1+k^2)}$.
Using the constraint $\sum_i \delta_i =\delta$ and the fact that $\sum_{i=1}^n R_i=1$,
we find $\sqrt{k^2/(1+k^2)} = \delta$. So, the maximum occurs at $\delta_i = \delta R_i$
and has the value $\sum_{i=1}^n \sqrt{R_i^2 - \delta^2 R_i^2} = \sqrt{1-\delta^2}$.
We have thus shown that $Z(W) \le \sqrt{1-d(W)^2}$, which is equivalent to $d(W) \le \sqrt{1-Z(W)^2}$.
\end{proof}

From the above two lemmas, the proof of \eqref{Eq:rateReliability4} is immediate.

\subsection{Proof of Proposition~\ref{prop:channelRecursion}}
To prove \eqref{eqn:Wrecursion1}, we write
\begin{align}
W&_{2N}^{(2i-1)}(y_1^{2N},u_1^{2i-2}|u_{2i-1}) \nonumber \\
 & = \sum_{u_{2i}^{2N}}  \frac{1}{2^{2N-1}}\; W_{2N}(y_1^{2N}|u_1^{2N}) \nonumber\\
& = \sum_{u_{2i,o}^{2N},u_{2i,e}^{2N}}
\frac{1}{2^{2N-1}}\; W_{N}(y_1^{N}|u_{1,o}^{2N}\oplus u_{1,e}^{2N})W_{N}(y_{N+1}^{2N}|u_{1,e}^{2N}) \nonumber\\
\begin{split}
= \sum_{u_{2i}} \frac{1}{2}\, \sum_{u_{2i+1,e}^{2N}} \frac{1}{2^{N-1}} \, W_{N}(y_{N+1}^{2N}|u_{1,e}^{2N})\\
\cdot \sum_{u_{2i+1,o}^{2N}}
\frac{1}{2^{N-1}}\; W_{N}(y_1^{N}|u_{1,o}^{2N}\oplus u_{1,e}^{2N}).
\end{split}\label{eqn:Wrec1}
\end{align}
By definition \eqref{eqn:splitChannels}, the sum over $u_{2i+1,o}^{2N}$ for any fixed $ u_{1,e}^{2N}$ equals
$$W_{N}^{(i)}(y_1^{N},u_{1,o}^{2i-2} \oplus u_{1,e}^{2i-2}|u_{2i-1}\oplus u_{2i}),$$
because, as $u_{2i+1,o}^{2N}$ ranges over ${\cal X}^{N-i}$, $u_{2i+1,o}^{2N}\oplus u_{2i+1,e}^{2N}$ ranges also over ${\cal X}^{N-i}$.
We now factor this term out of the middle sum in \eqref{eqn:Wrec1} and use \eqref{eqn:splitChannels} again to obtain \eqref{eqn:Wrecursion1}.
For the proof of \eqref{eqn:Wrecursion2}, we write
\begin{multline*}
W_{2N}^{(2i)}(y_1^{2N},u_1^{2i-1}|u_{2i}) = \sum_{u_{2i+1}^{2N}}
\frac{1}{2^{2N-1}}\; W_{2N}(y_1^{2N}|u_{1}^{2N})\\
= \frac{1}{2}\; \sum_{u_{2i+1,e}^{2N}} \frac{1}{2^{N-1}}\; W_{N}(y_{N+1}^{2N}|u_{1,e}^{2N})\\
\cdot \sum_{u_{2i+1,o}^{2N}} \frac{1}{2^{N-1}}\; W_{N}(y_1^{N}|u_{1,o}^{2N} \oplus u_{1,e}^{2N}).
\end{multline*}
By carrying out the inner and outer sums in the same manner as in the proof of \eqref{eqn:Wrecursion1}, we obtain \eqref{eqn:Wrecursion2}.

\subsection{Proof of Proposition~\ref{prop:rateConservationS}}

Let us specify the channels as follows: $W:{\cal X}\to {\cal Y}$, $W^\prime :{\cal X}\to \tilde{Y}$, and $W^{\prime\prime}:{\cal X}\to \tilde{Y}\times{\cal X}$.
By hypothesis there is a one-to-one function $f:{\cal Y}\to \tilde{\cal Y}$ such that \eqref{eqn:splitPPGen1} and \eqref{eqn:splitPPGen2} are satisfied.
For the proof it is helpful to define an ensemble of RVs $(U_1,U_2,X_1,X_2,Y_1,Y_2,\tilde{Y})$
so that the pair $(U_1,U_2)$ is uniformly distributed over ${\cal X}^2$, $(X_1,X_2) = (U_1\oplus U_2,U_2)$,  $P_{Y_1,Y_2|X_1,X_2}(y_1,y_2|x_1,x_2) = W(y_1|x_1)W(y_2|x_2)$, and $\tilde{Y} =f(Y_1,Y_2)$.
We now have
\begin{align*}
W^\prime(\tilde{y}|u_1) & = P_{\tilde{Y}|U_1}(\tilde{y}|u_1),\\
W^{\prime\prime}(\tilde{y},u_1|u_2) & =P_{\tilde{Y}U_1|U_2}(\tilde{y},u_1|u_2).
\end{align*}
From these and the fact that  $(Y_1,Y_2)\mapsto \tilde{Y}$ is invertible, we get
\begin{align*}
I(W^\prime) & = I(U_1;\tilde{Y}) = I(U_1;Y_1Y_2),\\
I(W^{\prime\prime}) & = I(U_2;\tilde{Y}U_1) = I(U_2\,;\,Y_1Y_2U_1).
\end{align*}
Since $U_1$ and $U_2$ are independent, $I(U_2;Y_1Y_2U_1)$ equals $I(U_2;Y_1Y_2|U_1)$. So, by the chain rule, we have
\begin{align*}
I(W^{\prime})+I(W^{\prime\prime}) & =  I(U_1U_2;Y_1Y_2) = I(X_1X_2;Y_1Y_2)
\end{align*}
where the second equality is due to the one-to-one relationship between $(X_1,X_2)$ and $(U_1,U_2)$.
The proof of \eqref{eqn:conservation1S} is completed by noting that $I(X_1X_2;Y_1Y_2)$ equals $I(X_1;Y_1)+I(X_2;Y_2)$ which in turn equals $2I(W)$.

To prove \eqref{eqn:FixedPointS}, we begin by noting that
\begin{align*}
I(W^{\prime\prime}) & = I(U_2;Y_1Y_2U_1)\\
& =  I(U_2;Y_2) + I(U_2;Y_1U_1|Y_2)\\
& =  I(W) + I(U_2;Y_1U_1|Y_2).
\end{align*}
This shows that $I(W^{\prime\prime}) \ge I(W)$. This and \eqref{eqn:conservation1S} give \eqref{eqn:FixedPointS}.
The above proof shows that equality holds in \eqref{eqn:FixedPointS} iff $I(U_2;Y_1U_1|Y_2) = 0$,
which is equivalent to having
\begin{multline*}
P_{U_1,U_2,Y_1|Y_2}(u_1,u_2,y_1|y_2) = P_{U_1,Y_1|Y_2}(u_1,y_1|y_2) \\
\cdot P_{U_2|Y_2}(u_2|y_2)
\end{multline*}
for all $(u_1,u_2,y_1,y_2)$ such that $P_{Y_2}(y_2) > 0$, or equivalently,
\begin{multline}\label{eqn:cond2}
P_{Y_1,Y_2|U_1,U_2}(y_1,y_2|u_1,u_2) P_{Y_2}(y_2) \\
= P_{Y_1,Y_2|U_1}(y_1,y_2|u_1) P_{Y_2|U_2}(y_2|u_2)
\end{multline}
for all $(u_1,u_2,y_1,y_2)$.
Since $P_{Y_1,Y_2|U_1,U_2}(y_1,y_2|u_1,u_2) = W(y_1|u_1\oplus u_2)W(y_2|u_2)$, \eqref{eqn:cond2} can be written as
\begin{align}\label{eqn:cond4}
W(y_2|u_2) \left[W(y_1|u_1\oplus u_2) P_{Y_2}(y_2) - P_{Y_1,Y_2}(y_1,y_2|u_1)\right] & = 0.
\end{align}
Substituting $P_{Y_2}(y_2) = \frac{1}{2} W(y_2|u_2) + \frac{1}{2} W(y_2|u_2\oplus 1)$ and
\begin{align*}
P_{Y_1,Y_2|U_1}(y_1,y_2|u_1) & = \frac{1}{2} W(y_1|u_1\oplus u_2)W(y_2|u_2) \\
&  + \frac{1}{2} W(y_1|u_1\oplus u_2 \oplus 1)W(y_2|u_2\oplus 1)
\end{align*}
into \eqref{eqn:cond4} and simplifying, we obtain
\begin{multline*}
W(y_2|u_2) W(y_2|u_2\oplus 1)\\
  \cdot \left[ W(y_1|u_1\oplus u_2)
  - W(y_1|u_1\oplus u_2 \oplus 1) \right]  = 0,
\end{multline*}
which for all four possible values of $(u_1,u_2)$ is equivalent to
\begin{align*}
W(y_2|0) W(y_2|1)\, \left[ W(y_1|0) - W(y_1|1) \right]  & = 0.
\end{align*}
Thus, either there exists no $y_2$ such that $W(y_2|0) W(y_2|1) >0$, in which case $I(W) = 1$, or
for all $y_1$ we have $W(y_1|0) = W(y_1|1)$, which implies $I(W) = 0$.

\subsection{Proof of Proposition~\ref{Proposition:BasicInequalityS}}
Proof of \eqref{eqn:BasicInequality1S} is straightforward.
\begin{align*}
Z(W^{\prime\prime}) &  = \sum_{y_1^2,u_1}
\sqrt{W^{\prime\prime}(f(y_1,y_2),u_1|0)} \\
& \qquad \cdot \sqrt{W^{\prime\prime}(f(y_1,y_2),u_1|1)} \\
&  = \sum_{y_1^{2},u_1}  \frac{1}{2} \,
\sqrt{W(y_1\mid u_{1})W(y_2\mid 0)} \\
& \qquad \cdot \sqrt{W(y_1\mid u_1\oplus 1)W(y_2\mid 1)} \\
& = \sum_{y_2} \sqrt{W(y_2\mid 0)W(y_2\mid 1)} \\
& \qquad \cdot \sum_{u_1} \frac{1}{2} \; \sum_{y_1}\sqrt{W(y_1\mid u_1)W(y_1\mid u_1\oplus 1)}\\
& = Z(W)^2.
\end{align*}

To prove \eqref{eqn:BasicInequality2S}, we put for shorthand
$\alpha(y_1) = W(y_1|0)$, $\delta(y_1) = W(y_1|1)$, $\beta(y_2) = W(y_2|0)$, and $\gamma(y_2) = W(y_2|1)$,
and write
\begin{align*}
Z(W^{\prime}) & = \sum_{y_1^{2}} \sqrt{W^{\prime}(f(y_1,y_2)|0)\,W^{\prime}(f(y_1,y_2)|1)} \\
& = \sum_{y_1^2} \frac{1}{2} \sqrt{\alpha(y_1) \beta(y_2) + \delta(y_1)\gamma(y_2)  } \\
& \qquad \cdot \sqrt{\alpha(y_1) \gamma(y_2) +   \delta(y_1)\beta(y_2) } \\
& \le
 \sum_{y_1^2} \frac{1}{2} \left[\sqrt{\alpha(y_1) \beta(y_2)} + \sqrt{ \delta(y_1) \gamma(y_2)}\,\right]  \\
& \qquad \cdot \left[\sqrt{\alpha(y_1) \gamma(y_2)} + \sqrt{ \delta(y_1) \beta(y_2)}\,\right] \\
& \qquad - \sum_{y_1^2} \sqrt{\alpha(y_1) \beta(y_2)  \delta(y_1) \gamma(y_2)} 
\end{align*}
where the inequality follows from the identity
\begin{multline*}
\left[\sqrt{ (\alpha \beta + \delta \gamma ) (\alpha \gamma + \delta \beta )}\right]^2 \\
 +2\,\sqrt{\alpha \beta \delta \gamma }\; (\sqrt{\alpha}-\sqrt{\delta})^2 (\sqrt{\beta}-\sqrt{\gamma})^2  \\
= \left[ (\sqrt{\alpha \beta} + \sqrt{\delta \gamma} )(\sqrt{\alpha \gamma} + \sqrt{\delta \beta })
 - 2\sqrt{\alpha \beta \delta \gamma }\right]^2.
\end{multline*}
Next, we note that
\begin{align*}
\sum_{y_1^2} \alpha(y_1) \sqrt{\beta(y_2) \gamma(y_2)} & =Z(W).
\end{align*}
Likewise, each term obtained by expanding $(\sqrt{\alpha(y_1) \beta(y_2)} + \sqrt{\delta(y_1)\gamma(y_2)})(\sqrt{\alpha(y_1) \gamma(y_2)} + \sqrt{ \delta(y_1)\beta(y_2)})$ gives $Z(W)$ when summed over $y_1^2$.
Also, $\sqrt{\alpha(y_1) \beta(y_2) \delta(y_1)\gamma(y_2) }$ summed over $y_1^{2}$ equals $Z(W)^2$.
Combining these, we obtain the claim \eqref{eqn:BasicInequality2S}.
Equality holds in \eqref{eqn:BasicInequality2S} iff, for any choice of $y_1^{2}$, one of the following is true: $\alpha(y_1) \beta(y_2) \gamma(y_2) \delta(y_1) =0 $ or $\alpha(y_1) = \delta(y_1)$ or $\beta(y_2) = \gamma(y_2)$.
This is satisfied if $W$ is a BEC. Conversely, if we take $y_1=y_2$, we see that for equality in \eqref{eqn:BasicInequality2S}, we must have, for any choice of $y_1$, either $\alpha(y_1) \delta(y_1) =0$ or $\alpha(y_1) = \delta(y_1)$; this is equivalent to saying that $W$ is a BEC.

To prove \eqref{eqn:BasicInequality3S}, we need the following result which states that the parameter $Z(W)$ is a convex function of the channel transition probabilities.
\begin{Lemma}\label{Zconvexity}
Given any collection of B-DMCs $W_j:{\cal X} \to {\cal Y}$, $j\in {\cal J}$, and
a probability distribution $Q$ on ${\cal J}$, define
$W: {\cal X}\to {\cal Y}$ as the channel
$W(y|x) = \sum_{j\in {\cal J}} Q(j) W_j(y|x)$.
Then,
\begin{align}\label{eqn:Zconvexity}
\sum_{j\in {\cal J}} Q(j) Z(W_{j}) & \le Z(W).
\end{align}
\end{Lemma}

\begin{proof}This follows by first rewriting $Z(W)$ in a different form and then applying Minkowsky's inequality \cite[p.~524, ineq.~(h)]{Gallager}.
\begin{align*}
Z(W) & = \sum_{y} \sqrt{W(y|0)W(y|1)}\\
& = -1 +  \frac{1}{2}\, \sum_y \left[ \sum_{x} \sqrt{W(y|x)} \right]^2\\
& \ge  -1 +  \frac{1}{2}\, \sum_y \sum_{j\in {\cal J}} Q(j) \,\left[ \sum_{x} \sqrt{W_j(y|x)} \right]^2\\
& = \sum_{j\in {\cal J}} Q(j)\,Z(W_j).
\end{align*}
\end{proof}

We now write $W^{\prime}$ as the mixture
\begin{align*}
W^{\prime}(f(y_1,y_2)|u_{1}) = \frac{1}{2} \left[ W_0(y_1^{2}\mid u_{1}) + W_1(y_1^2|u_1)\right]
\end{align*}
where
\begin{align*}
W_0(y_1^{2}|u_{1}) & = W(y_1|u_{1})W(y_2|0),\\
W_1(y_1^{2}|u_{1}) &= W(y_1|u_1\oplus 1)W(y_2|1),
\end{align*}
and apply Lemma~\ref{Zconvexity} to obtain the claimed inequality
\begin{align*}
Z(W^{\prime}) & \ge \frac{1}{2}\left[Z(W_0)+Z(W_1)\right] = Z(W).
\end{align*}
Since $0\le Z(W) \le 1$ and $Z(W^{\prime\prime})=Z(W)^2$, we have $Z(W)\ge Z(W^{\prime\prime})$, with equality iff $Z(W)$ equals 0 or 1. Since $Z(W^{\prime}) \ge Z(W)$, this also shows that $Z(W^{\prime})=Z(W^{\prime\prime})$ iff $Z(W)$ equals 0 or 1.
So, by Prop.~\ref{rateReliability}, $Z(W^{\prime})=Z(W^{\prime\prime})$ iff $I(W)$ equals 1 or 0.

\subsection{Proof of Proposition~\ref{Prop:BECRecursion}}
From \eqref{eqn:splitPPGen1}, we have the identities
\begin{align}
W^{\prime}&(f(y_1,y_2)|0)W^{\prime}(f(y_1,y_2)|1)  = \nonumber \\
& \frac{1}{4}\,
\left[W(y_1|0)^2 + W(y_1|1)^2 \right] W(y_2|0)W(y_2|1) {}+ \nonumber\\
 & \frac{1}{4}\,\left[W(y_2|0)^2 + W(y_2|1)^2 \right] W(y_1|0)W(y_1|1),
\label{BECidentity1}
\end{align}
\begin{align}
W^{\prime}&(f(y_1,y_2)|0)-W^{\prime}(f(y_1,y_2)|1) = \nonumber \\
& \frac{1}{2} \left[W(y_1|0)-W(y_1|1)\right]
\left[W(y_2|0)-W(y_2|1)\right]. \label{BECidentity2}
\end{align}
Suppose $W$ is a BEC, but $W^{\prime}$ is not. Then, there exists $(y_1,y_2)$ such that the left sides of \eqref{BECidentity1} and \eqref{BECidentity2} are both different from zero.
From \eqref{BECidentity2}, we infer that neither $y_1$ nor $y_2$ is an erasure symbol for $W$.
But then the RHS of \eqref{BECidentity1} must be zero, which is a contradiction.
Thus, $W^{\prime}$ must be a BEC. From \eqref{BECidentity2}, we conclude that $f(y_1,y_2)$ is an erasure symbol for $W^{\prime}$ iff either $y_1$ or $y_2$ is an erasure symbol for $W$.
This shows that the erasure probability for $W^{\prime}$ is $2\epsilon -\epsilon^2$, where $\epsilon$ is the  erasure probability of $W$.

Conversely, suppose $W^{\prime}$ is a BEC but $W$ is not. Then, there exists $y_1$ such that $W(y_1|0)W(y_1|1)>0$ and $W(y_1|0)-W(y_1|1) \neq 0$. By taking $y_2=y_1$, we see that the RHSs of \eqref{BECidentity1} and \eqref{BECidentity2} can both be made non-zero, which contradicts the assumption that $W^{\prime}$ is a BEC.

The other claims follow from the identities
\begin{multline*}
W^{\prime\prime}(f(y_1,y_2),u_1|0) \,W^{\prime\prime}(f(y_1,y_2),u_1|1)  \\
= \frac{1}{4} W(y_1|u_1)W(y_1|u_1\oplus 1)W(y_2|0)W(y_2|1),
\end{multline*}
\begin{multline*}
W^{\prime\prime}(f(y_1,y_2),u_1|0)-W^{\prime\prime}(f(y_1,y_2),u_1|1)  \\
= \frac{1}{2} \left[W(y_1|u_1)W(y_2|0)-W(y_1|u_1\oplus 1)W(y_2|1)\right].
\end{multline*}
The arguments are similar to the ones already given and we omit the details, other than noting that $(f(y_1,y_2),u_1)$ is an erasure symbol for $W^{\prime\prime}$ iff both $y_1$ and $y_2$ are erasure symbols for $W$.

\subsection{Proof of Lemma~\ref{ZnConvergence}}

The proof follows that of a similar result from Chung \cite[Theorem 4.1.1]{Chung}.
Fix $\zeta >0$.
Let $\Omega_0\defn \{\omega \in \Omega: \lim_{n\to \infty} Z_n(\omega) = 0\}$.
By Prop.~\ref{prop:Iinfty},
$P(\Omega_0)=I_0$.
Fix $\omega \in \Omega_0$. $Z_n(\omega) \to 0$ implies that
there exists $n_0(\omega,\zeta)$ such that $n \ge n_0(\omega,\zeta) \Rightarrow
Z_n(\omega) \le \zeta$.
Thus, $\omega \in {\cal T}_m(\zeta)$ for some $m$.
So, $\Omega_0 \subset \bigcup_{m=1}^\infty {\cal T}_m(\zeta)$.
Therefore,
$P\left(\bigcup_{m=1}^\infty {\cal T}_m(\zeta)\right)  \ge P(\Omega_0)$.
Since ${\cal T}_m(\zeta) \uparrow \bigcup_{m=1}^\infty {\cal T}_m(\zeta)$,
by the monotone convergence property of a measure,
$\lim_{m\to \infty} P\left[{\cal T}_m(\zeta)\right] = P\left[\bigcup_{m=1}^\infty {\cal T}_m(\zeta)\right]$.
So, $\lim_{m\to \infty} P\left[{\cal T}_m(\zeta)\right] \ge I_0$.
It follows that, for any $\zeta >0$, $\delta >0$, there exists a finite $m_0 = m_0(\zeta,\delta)$ such that, for all $m\ge m_0$,
$P\left[{\cal T}_m(\zeta)\right] \ge I_0-\delta/2$.
This completes the proof.


\end{document}

%% file: fig1.tex
\begin{figure}[thb]
\begin{center}
\psset{arrowscale=1}
\psset{unit=0.5cm}
\psset{xunit=1,yunit=1}
\begin{pspicture}(0,0)(9,6)
\psline[linecolor=black,linewidth=0.7pt]{->}(0,2)(4.5,2)
\psline[linecolor=black,linewidth=0.7pt]{->}(0,4)(2.2,4)
\pscircle(2.5,4){0.3}
\rput[c](2.5,4){$+$}
\psline[linecolor=black,linewidth=0.7pt]{->}(2.8,4)(4.5,4)
\psframe(4.5,3.3)(6.5,4.7)
\rput(5.5,4){$W$}
\psline[linecolor=black,linewidth=0.7pt]{->}(2.5,2)(2.5,3.7)
\pscircle*(2.5,2){0.05}
\psframe(4.5,1.3)(6.5,2.7)
\rput(5.5,2){$W$}
\psline[linecolor=black,linewidth=0.7pt]{->}(6.5,2)(8.5,2)
\psline[linecolor=black,linewidth=0.7pt]{->}(6.5,4)(8.5,4)
\psframe(1.3,0.7)(7.3,5.3)
\rput[c](0.6,2.3){$u_2$}
\rput[c](0.6,4.3){$u_1$}
\rput[c](3.5,2.3){$x_2$}
\rput[c](3.5,4.3){$x_1$}
\rput[c](8,2.3){$y_2$}
\rput[c](8,4.3){$y_1$}
\rput[c](4.5,0.3){$W_2$}
\end{pspicture}
\caption{The channel $W_2$.}
\label{combine2}
\end{center}
\end{figure}

%% file: fig2.tex
\begin{figure}[thb]
\begin{center}
\psset{arrowscale=1}
\psset{unit=0.5cm}
\psset{xunit=1,yunit=1}
\begin{pspicture}(-7.5,-2)(10,12.5)
\psline[linecolor=black,linewidth=0.7pt]{->}(-0.5,2)(4.5,2)
\psline[linecolor=black,linewidth=0.7pt]{->}(-0.5,4)(2.2,4)
\psline[linecolor=black,linewidth=0.7pt]{->}(-4.2,4)(-3,4)
\psline[linecolor=black,linewidth=0.7pt]{->}(-7,2)(-3,2)
\psline[linecolor=black,linewidth=0.7pt]{->}(-7,4)(-4.8,4)
\pscircle(2.5,4){0.3}
\rput[c](2.5,4){$+$}
\psline[linecolor=black,linewidth=0.7pt]{->}(2.8,4)(4.5,4)
\psframe(4.5,3.3)(6.5,4.7)
\rput(5.5,4){$W$}
\psline[linecolor=black,linewidth=0.7pt]{->}(2.5,2)(2.5,3.7)
\pscircle*(2.5,2){0.05}
\psframe(4.5,1.3)(6.5,2.7)
\rput(5.5,2){$W$}
\psline[linecolor=black,linewidth=0.7pt]{->}(6.5,2)(9.5,2)
\psline[linecolor=black,linewidth=0.7pt]{->}(6.5,4)(9.5,4)
\psframe(1.3,0.7)(7.3,5.3)
\rput[c](3.5,2.3){$x_4$}
\rput[c](3.5,4.3){$x_3$}
\rput[c](8.75,2.3){$y_4$}
\rput[c](8.75,4.3){$y_3$}
\rput[c](4.5,0.3){$W_2$}

\psline[linecolor=black,linewidth=0.7pt]{->}(-7,8)(-3,8)
\psline[linecolor=black,linewidth=0.7pt]{->}(-7,10)(-4.8,10)

\pscircle(2.5,10){0.3}
\rput[c](2.5,10){$+$}

\psline[linecolor=black,linewidth=0.7pt]{->}(2.8,10)(4.5,10)
\psframe(4.5,9.3)(6.5,10.7)
\rput(5.5,10){$W$}
\psline[linecolor=black,linewidth=0.7pt]{->}(2.5,8)(2.5,9.7)
\pscircle*(2.5,8){0.05}
\psframe(4.5,7.3)(6.5,8.7)
\rput(5.5,8){$W$}
\psline[linecolor=black,linewidth=0.7pt]{->}(6.5,8)(9.5,8)
\psline[linecolor=black,linewidth=0.7pt]{->}(6.5,10)(9.5,10)
\psframe(1.3,6.7)(7.3,11.3)
\rput[c](3.5,8.3){$x_2$}
\rput[c](3.5,10.3){$x_1$}
\rput[c](8.75,8.3){$y_2$}
\rput[c](8.75,10.3){$y_1$}
\rput[c](4.5,6.3){$W_2$}

\pscircle(-4.5,4){0.3}
\pscircle(-4.5,2){0.05}
\rput[c](-4.5,4){$+$}
\psline[linecolor=black,linewidth=0.7pt]{->}(-4.5,2)(-4.5,3.7)
\pscircle(-4.5,10){0.3}
\pscircle(-4.5,8){0.05}
\rput[c](-4.5,10){$+$}
\psline[linecolor=black,linewidth=0.7pt]{->}(-4.5,8)(-4.5,9.7)
\psframe(-6,-0.5)(8,12)
\rput[c](1,-1.2){$W_4$}

\psline[linecolor=black,linewidth=0.7pt]{->}(-4.2,10)(-3,10)
\psline[linecolor=black,linewidth=0.7pt]{->}(-0.5,8)(4.5,8)
\psline[linecolor=black,linewidth=0.7pt]{->}(-0.5,10)(2.2,10)
\psline[linecolor=black,linewidth=0.4pt,linestyle=dashed]{-}(-3,10)(-0.5,10)
\psline[linecolor=black,linewidth=0.4pt,linestyle=dashed]{-}(-3,8)(-0.5,4)
\psline[linecolor=black,linewidth=0.4pt,linestyle=dashed]{-}(-3,4)(-0.5,8)
\psline[linecolor=black,linewidth=0.4pt,linestyle=dashed]{-}(-3,2)(-0.5,2)

\rput[c](0.5,8.3){$v_2$}
\rput[c](0.5,10.3){$v_1$}
\rput[c](0.5,2.3){$v_4$}
\rput[c](0.5,4.3){$v_3$}

\rput[c](-7,10.3){$u_1$}
\rput[c](-7,8.3){$u_2$}
\rput[c](-7,4.3){$u_3$}
\rput[c](-7,2.3){$u_4$}

\psframe(-3,0.7)(-0.5,11.3)
\rput[c](-1.75,0.3){$R_{4}$}
\end{pspicture}
\caption{The channel $W_4$ and its relation to $W_2$ and $W$.}
\label{combine4}
\end{center}
\end{figure}

%% file: fig3.tex
\begin{figure}[thb]
\begin{center}
\psset{arrowscale=1}
\psset{unit=0.45cm}
\psset{xunit=1,yunit=1}
\begin{pspicture}(2.3,-1)(20.3,22)
\psframe(4.5,0.5)(17.5,21.5)
\rput[c](11.5,0){$W_{N}$}
\psframe(8.5,1)(11.5,21)
\psframe(14,1)(17,10)
\psframe(14,12)(17,21)
\rput[c](10,11){$R_{N}$}
\rput[c](15.5,5.5){$W_{N/2}$}
\rput[c](15.5,16.5){$W_{N/2}$}

\psline[linecolor=black,linewidth=0.7pt]{->}(3,20)(6.2,20)
\rput[c](3.3,20.3){$u_1$}
\psline[linecolor=black,linewidth=0.7pt]{->}(6.8,20)(8.5,20)
\rput[c](7.5,20.3){$s_1$}
\pscircle(6.5,20){0.3}
\rput[c](6.5,20){$+$}

\psline[linecolor=black,linewidth=0.7pt]{->}(11.5,20)(14,20)
\rput[c](12.75,20.3){$v_1$}
\psline[linecolor=black,linewidth=0.7pt]{->}(17,20)(19,20)
\rput[l](17.7,20.4){$y_1$}

\psline[linecolor=black,linewidth=0.7pt]{->}(3,18)(8.5,18)
\rput[c](3.3,18.3){$u_2$}
\psline[linecolor=black,linewidth=0.7pt]{->}(6.5,18)(6.5,19.7)
\rput[c](7.5,18.3){$s_2$}
\pscircle(6.5,18){0.1}

\psline[linecolor=black,linewidth=0.7pt]{->}(11.5,18)(14,18)
\rput[c](12.75,18.3){$v_2$}
\psline[linecolor=black,linewidth=0.7pt]{->}(17,18)(19,18)
\rput[l](17.7,18.4){$y_2$}

\psline[linecolor=black,linewidth=0.7pt]{->}(3,15)(6.2,15)
\rput[c](3.3,15.3){$u_{N/2-1}$}
\psline[linecolor=black,linewidth=0.7pt]{->}(6.8,15)(8.5,15)
\rput[c](7.3,15.6){$s_{N/2-1}$}
\pscircle(6.5,15){0.3}
\rput[c](6.5,15){$+$}

\psline[linecolor=black,linewidth=0.7pt]{->}(11.5,15)(14,15)
\rput[c](12.75,15.4){$v_{N/2-1}$}
\psline[linecolor=black,linewidth=0.7pt]{->}(17,15)(19,15)
\rput[l](17.7,15.4){$y_{N/2-1}$}

\psline[linecolor=black,linewidth=0.7pt]{->}(3,13)(8.5,13)
\rput[c](3.3,13.3){$u_{N/2}$}
\psline[linecolor=black,linewidth=0.7pt]{->}(6.5,13)(6.5,14.7)
\rput[c](7.5,13.3){$s_{N/2}$}
\pscircle(6.5,13){0.1}

\psline[linecolor=black,linewidth=0.7pt]{->}(11.5,13)(14,13)
\rput[c](12.75,13.3){$v_{N/2}$}
\psline[linecolor=black,linewidth=0.7pt]{->}(17,13)(19,13)
\rput[l](17.7,13.4){$y_{N/2}$}

\psline[linecolor=black,linewidth=0.7pt]{->}(3,9)(6.2,9)
\rput[c](3.3,9.3){$u_{N/2+1}$}
\psline[linecolor=black,linewidth=0.7pt]{->}(6.8,9)(8.5,9)
\rput[c](7.3,9.6){$s_{N/2+1}$}
\pscircle(6.5,9){0.3}
\rput[c](6.5,9){$+$}

\psline[linecolor=black,linewidth=0.7pt]{->}(11.5,9)(14,9)
\rput[c](12.75,9.4){$v_{N/2+1}$}
\psline[linecolor=black,linewidth=0.7pt]{->}(17,9)(19,9)
\rput[l](17.7,9.4){$y_{N/2+1}$}

\psline[linecolor=black,linewidth=0.7pt]{->}(3,7)(8.5,7)
\rput[c](3.3,7.3){$u_{N/2+2}$}
\psline[linecolor=black,linewidth=0.7pt]{->}(6.5,7)(6.5,8.7)
\rput[ct](7.3,6.5){$s_{N/2+2}$}
\pscircle(6.5,7){0.1}

\psline[linecolor=black,linewidth=0.7pt]{->}(11.5,7)(14,7)
\rput[c](12.75,7.4){$v_{N/2+2}$}
\psline[linecolor=black,linewidth=0.7pt]{->}(17,7)(19,7)
\rput[l](17.7,7.4){$y_{N/2+2}$}

\psline[linecolor=black,linewidth=0.7pt]{->}(3,4)(6.2,4)
\rput[c](3.3,4.3){$u_{N-1}$}
\psline[linecolor=black,linewidth=0.7pt]{->}(6.8,4)(8.5,4)
\rput[c](7.5,4.5){$s_{N-1}$}
\pscircle(6.5,4){0.3}
\rput[c](6.5,4){$+$}

\psline[linecolor=black,linewidth=0.7pt]{->}(11.5,4)(14,4)
\rput[c](12.75,4.3){$v_{N-1}$}
\psline[linecolor=black,linewidth=0.7pt]{->}(17,4)(19,4)
\rput[l](17.7,4.4){$y_{N-1}$}

\psline[linecolor=black,linewidth=0.7pt]{->}(3,2)(8.5,2)
\rput[c](3.3,2.3){$u_{N}$}
\psline[linecolor=black,linewidth=0.7pt]{->}(6.5,2)(6.5,3.7)
\rput[c](7.5,2.3){$s_{N}$}
\pscircle(6.5,2){0.1}

\psline[linecolor=black,linewidth=0.7pt]{->}(11.5,2)(14,2)
\rput[c](12.75,2.3){$v_{N}$}
\psline[linecolor=black,linewidth=0.7pt]{->}(17,2)(19,2)
\rput[l](17.7,2.4){$y_{N}$}

\rput[c](3.5,17){$\vdots$}
\rput[c](3.5,5.8){$\vdots$}

\rput[c](6.5,17){$\vdots$}
\rput[c](6.5,5.8){$\vdots$}

\rput[c](12.75,17){$\vdots$}
\rput[c](12.75,5.8){$\vdots$}

\rput[c](18.25,17){$\vdots$}
\rput[c](18.25,5.8){$\vdots$}

\end{pspicture}
\caption{Recursive construction of $W_{N}$ from two copies of $W_{N/2}$.}
\label{mtFolding2}
\end{center}
\end{figure}

%% file: fig5.tex
\begin{figure}[thb]
\begin{center}
\psset{arrowscale=1}
\psset{unit=0.55cm}
\psset{xunit=1.5,yunit=1}
\begin{pspicture}(1,0)(12,16)

\psline[linecolor=black,linewidth=0.3pt]{-}(2,1)(5,9)
\psline[linecolor=black,linewidth=0.3pt]{-}(2,9)(5,1)
\psline[linecolor=black,linewidth=0.3pt]{-}(2,5)(5,13)
\psline[linecolor=black,linewidth=0.3pt]{-}(2,13)(5,5)
\psline[linecolor=black,linewidth=0.3pt]{-}(2,3)(5,11)
\psline[linecolor=black,linewidth=0.3pt]{-}(2,11)(5,3)
\psline[linecolor=black,linewidth=0.3pt]{-}(2,7)(5,15)
\psline[linecolor=black,linewidth=0.3pt]{-}(2,15)(5,7)

\psline[linecolor=black,linewidth=0.3pt]{-}(2,1)(5,1)
\psline[linecolor=black,linewidth=0.3pt]{-}(2,3)(5,3)
\psline[linecolor=black,linewidth=0.3pt]{-}(2,5)(5,5)
\psline[linecolor=black,linewidth=0.3pt]{-}(2,7)(5,7)
\psline[linecolor=black,linewidth=0.3pt]{-}(2,9)(5,9)
\psline[linecolor=black,linewidth=0.3pt]{-}(2,11)(5,11)
\psline[linecolor=black,linewidth=0.3pt]{-}(2,13)(5,13)
\psline[linecolor=black,linewidth=0.3pt]{-}(2,15)(5,15)
\psline[linecolor=black,linewidth=0.3pt]{-}(5,1)(8,5)
\psline[linecolor=black,linewidth=0.3pt]{-}(5,5)(8,1)
\psline[linecolor=black,linewidth=0.3pt]{-}(5,3)(8,7)
\psline[linecolor=black,linewidth=0.3pt]{-}(5,7)(8,3)
\psline[linecolor=black,linewidth=0.3pt]{-}(5,9)(8,13)
\psline[linecolor=black,linewidth=0.3pt]{-}(5,13)(8,9)
\psline[linecolor=black,linewidth=0.3pt]{-}(5,11)(8,15)
\psline[linecolor=black,linewidth=0.3pt]{-}(5,15)(8,11)

\psline[linecolor=black,linewidth=0.3pt]{-}(5,1)(8,1)
\psline[linecolor=black,linewidth=0.3pt]{-}(5,3)(8,3)
\psline[linecolor=black,linewidth=0.3pt]{-}(5,5)(8,5)
\psline[linecolor=black,linewidth=0.3pt]{-}(5,7)(8,7)
\psline[linecolor=black,linewidth=0.3pt]{-}(5,9)(8,9)
\psline[linecolor=black,linewidth=0.3pt]{-}(5,11)(8,11)
\psline[linecolor=black,linewidth=0.3pt]{-}(5,13)(8,13)
\psline[linecolor=black,linewidth=0.3pt]{-}(5,15)(8,15)
\psline[linecolor=black,linewidth=0.3pt]{-}(8,1)(11,3)
\psline[linecolor=black,linewidth=0.3pt]{-}(8,3)(11,1)
\psline[linecolor=black,linewidth=0.3pt]{-}(8,5)(11,7)
\psline[linecolor=black,linewidth=0.3pt]{-}(8,7)(11,5)
\psline[linecolor=black,linewidth=0.3pt]{-}(8,9)(11,11)
\psline[linecolor=black,linewidth=0.3pt]{-}(8,11)(11,9)
\psline[linecolor=black,linewidth=0.3pt]{-}(8,13)(11,15)
\psline[linecolor=black,linewidth=0.3pt]{-}(8,15)(11,13)

\psline[linecolor=black,linewidth=0.3pt]{-}(8,1)(11,1)
\psline[linecolor=black,linewidth=0.3pt]{-}(8,3)(11,3)
\psline[linecolor=black,linewidth=0.3pt]{-}(8,5)(11,5)
\psline[linecolor=black,linewidth=0.3pt]{-}(8,7)(11,7)
\psline[linecolor=black,linewidth=0.3pt]{-}(8,9)(11,9)
\psline[linecolor=black,linewidth=0.3pt]{-}(8,11)(11,11)
\psline[linecolor=black,linewidth=0.3pt]{-}(8,13)(11,13)
\psline[linecolor=black,linewidth=0.3pt]{-}(8,15)(11,15)

\pscircle*(2,1){0.1}\pscircle*(5,1){0.1}\pscircle*(8,1){0.1}\pscircle*(11,1){0.1}
\pscircle*(2,3){0.1}\pscircle*(5,3){0.1}\pscircle*(8,3){0.1}\pscircle*(11,3){0.1}
\pscircle*(2,5){0.1}\pscircle*(5,5){0.1}\pscircle*(8,5){0.1}\pscircle*(11,5){0.1}
\pscircle*(2,7){0.1}\pscircle*(5,7){0.1}\pscircle*(8,7){0.1}\pscircle*(11,7){0.1}
\pscircle*(2,9){0.1}\pscircle*(5,9){0.1}\pscircle*(8,9){0.1}\pscircle*(11,9){0.1}
\pscircle*(2,11){0.1}\pscircle*(5,11){0.1}\pscircle*(8,11){0.1}\pscircle*(11,11){0.1}
\pscircle*(2,13){0.1}\pscircle*(5,13){0.1}\pscircle*(8,13){0.1}\pscircle*(11,13){0.1}
\pscircle*(2,15){0.1}\pscircle*(5,15){0.1}\pscircle*(8,15){0.1}\pscircle*(11,15){0.1}

\rput[bc](11,1.5){\small $W$}
\rput[bc](11,3.5){\small $W$}
\rput[bc](11,5.5){\small $W$}
\rput[bc](11,7.5){\small $W$}
\rput[bc](11,9.5){\small $W$}
\rput[bc](11,11.5){\small $W$}
\rput[bc](11,13.5){\small $W$}
\rput[bc](11,15.5){\small $W$}

\rput[bc](8,1.6){\small $W_2^{(2)}$}
\rput[bc](8,3.6){\small $W_2^{(1)}$}
\rput[bc](8,5.6){\small $W_2^{(2)}$}
\rput[bc](8,7.6){\small $W_2^{(1)}$}
\rput[bc](8,9.6){\small $W_2^{(2)}$}
\rput[bc](8,11.6){\small $W_2^{(1)}$}
\rput[bc](8,13.6){\small $W_2^{(2)}$}
\rput[bc](8,15.6){\small $W_2^{(1)}$}

\rput[bc](5,1.6){\small $W_4^{(4)}$}
\rput[bc](5,3.6){\small $W_4^{(2)}$}
\rput[bc](5,5.6){\small $W_4^{(3)}$}
\rput[bc](5,7.6){\small $W_4^{(1)}$}
\rput[bc](5,9.6){\small $W_4^{(4)}$}
\rput[bc](5,11.6){\small $W_4^{(2)}$}
\rput[bc](5,13.6){\small $W_4^{(3)}$}
\rput[bc](5,15.6){\small $W_4^{(1)}$}

\rput[bc](2,1.6){\small $W_8^{(8)}$}
\rput[bc](2,3.6){\small $W_8^{(4)}$}
\rput[bc](2,5.6){\small $W_8^{(6)}$}
\rput[bc](2,7.6){\small $W_8^{(2)}$}
\rput[bc](2,9.6){\small $W_8^{(7)}$}
\rput[bc](2,11.6){\small $W_8^{(3)}$}
\rput[bc](2,13.6){\small $W_8^{(5)}$}
\rput[bc](2,15.6){\small $W_8^{(1)}$}
\end{pspicture}
\caption{The channel transformation process with $N=8$ channels.}
\label{fig:channelTransform}
\end{center}
\end{figure}

%% file: fig6.tex
\begin{figure}[thb]
\begin{center}
\psset{arrowscale=1}
\psset{unit=0.5cm}
\psset{xunit=1,yunit=1}
\begin{pspicture}(0,0)(16,17.5)
\psline[linecolor=black,linewidth=0.5pt]{-}(0,8)(2,8)
\psline[linecolor=black,linewidth=0.5pt]{->}(0,8.2)(0,10)
\psline[linecolor=black,linewidth=0.5pt]{->}(0,7.8)(0,6)
\rput[cb](0,10.4){$0$}
\rput[ct](0,5.6){$1$}
\psline[linecolor=black,linewidth=0.5pt]{-}(2,4)(2,12)
\psline[linecolor=black,linewidth=0.5pt]{-}(2,12)(5,12)
\psline[linecolor=black,linewidth=0.5pt]{-}(2,4)(5,4)
\psline[linecolor=black,linewidth=0.5pt]{-}(5,10)(5,14)
\psline[linecolor=black,linewidth=0.5pt]{-}(5,2)(5,6)
\psline[linecolor=black,linewidth=0.5pt]{-}(5,10)(8,10)
\psline[linecolor=black,linewidth=0.5pt]{-}(5,14)(8,14)
\psline[linecolor=black,linewidth=0.5pt]{-}(5,2)(8,2)
\psline[linecolor=black,linewidth=0.5pt]{-}(5,6)(8,6)
\psline[linecolor=black,linewidth=0.5pt]{-}(8,1)(8,3)
\psline[linecolor=black,linewidth=0.5pt]{-}(8,5)(8,7)
\psline[linecolor=black,linewidth=0.5pt]{-}(8,9)(8,11)
\psline[linecolor=black,linewidth=0.5pt]{-}(8,13)(8,15)
\psline[linecolor=black,linewidth=0.5pt]{-}(8,1)(11,1)
\psline[linecolor=black,linewidth=0.5pt]{-}(8,3)(11,3)
\psline[linecolor=black,linewidth=0.5pt]{-}(8,5)(11,5)
\psline[linecolor=black,linewidth=0.5pt]{-}(8,7)(11,7)
\psline[linecolor=black,linewidth=0.5pt]{-}(8,9)(11,9)
\psline[linecolor=black,linewidth=0.5pt]{-}(8,11)(11,11)
\psline[linecolor=black,linewidth=0.5pt]{-}(8,13)(11,13)
\psline[linecolor=black,linewidth=0.5pt]{-}(8,15)(11,15)
\rput[lc](2.5,8){\small $W$}
\rput[lc](6.7,12){\small $W_2^{(1)}=W_0$}
\rput[lc](6.7,4){\small $W_2^{(2)}=W_1$}
\rput[lc](9.9,14){\small $W_4^{(1)}=W_{00}$}
\rput[lc](9.9,10){\small $W_4^{(2)}=W_{01}$}
\rput[lc](9.9,6){\small $W_4^{(3)}=W_{10}$}
\rput[lc](9.9,2){\small $W_4^{(4)}=W_{11}$}
\rput[lc](12.9,15){\small $W_8^{(1)}=W_{000}$}
\rput[lc](12.9,13){\small $W_8^{(2)}=W_{001}$}
\rput[lc](12.9,11){\small $W_8^{(3)}=W_{010}$}
\rput[lc](12.9,9){\small $W_8^{(4)}=W_{011}$}
\rput[lc](12.9,7){\small $W_8^{(5)}=W_{100}$}
\rput[lc](12.9,5){\small $W_8^{(6)}=W_{101}$}
\rput[lc](12.9,3){\small $W_8^{(7)}=W_{110}$}
\rput[lc](12.9,1){\small $W_8^{(8)}=W_{111}$}
\rput[l](15,14){\small $\cdots$}
\rput[l](15,10){\small $\cdots$}
\rput[l](15,6){\small $\cdots$}
\rput[l](15,2){\small $\cdots$}
\end{pspicture}
\caption{The tree process for the recursive channel construction.}
\label{fig:channelTree}
\end{center}
\end{figure}

%% file: fig8.tex
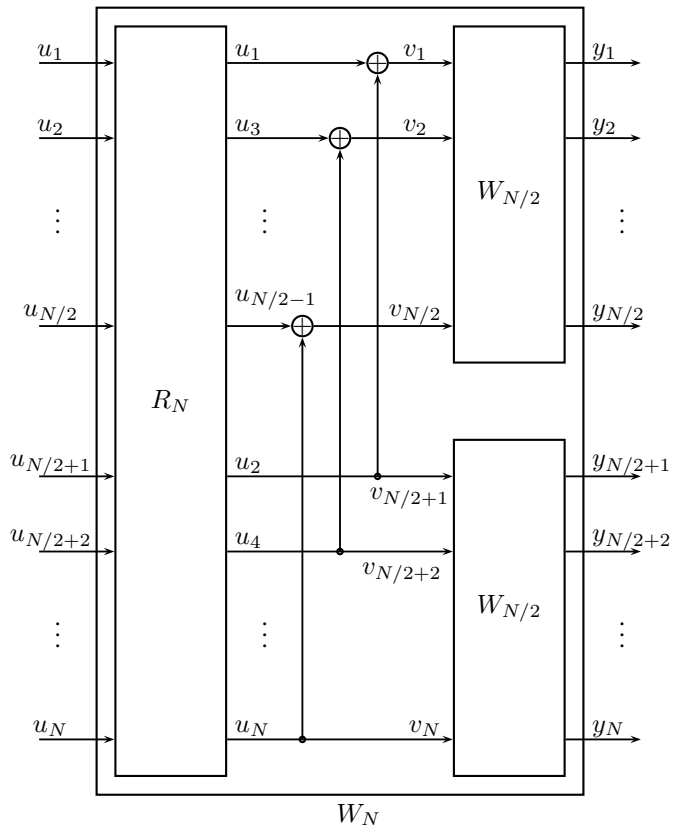
\begin{figure}[thb]
\begin{center}
\psset{arrowscale=1}
\psset{unit=0.5cm}
\psset{xunit=1,yunit=1}
\begin{pspicture}(2.3,0)(20.3,22)
\psframe(4.5,0.5)(17.5,21.5)
\rput[c](11.5,0){$W_{N}$}
\psframe(5,1)(8,21)
\psframe(14,1)(17,10)
\psframe(14,12)(17,21)
\rput[c](6.5,11){$R_{N}$}
\rput[c](15.5,5.5){$W_{N/2}$}
\rput[c](15.5,16.5){$W_{N/2}$}
\psline[linecolor=black,linewidth=0.7pt]{->}(3,20)(5,20)
\rput[c](3.3,20.3){$u_1$}
\psline[linecolor=black,linewidth=0.7pt]{->}(3,18)(5,18)
\rput[c](3.3,18.3){$u_2$}
\psline[linecolor=black,linewidth=0.7pt]{->}(3,13)(5,13)
\rput[c](3.3,13.3){$u_{N/2}$}
\rput[c](3.5,16){$\vdots$}
\rput[c](9,16){$\vdots$}
\rput[c](18.5,16){$\vdots$}

\psline[linecolor=black,linewidth=0.7pt]{->}(8,20)(11.7,20)
\rput[l](8.2,20.3){$u_1$}
\pscircle(12,20){0.3}
\rput[c](12,20){$+$}
\psline[linecolor=black,linewidth=0.7pt]{->}(12.3,20)(14,20)
\rput[c](13,20.3){$v_1$}
\psline[linecolor=black,linewidth=0.7pt]{->}(17,20)(19,20)
\rput[l](17.7,20.3){$y_1$}

\psline[linecolor=black,linewidth=0.7pt]{->}(8,18)(10.7,18)
\rput[l](8.2,18.3){$u_3$}
\pscircle(11,18){0.3}
\rput[c](11,18){$+$}
\psline[linecolor=black,linewidth=0.7pt]{->}(11.3,18)(14,18)
\rput[c](13,18.3){$v_2$}
\psline[linecolor=black,linewidth=0.7pt]{->}(17,18)(19,18)
\rput[l](17.7,18.3){$y_2$}

\psline[linecolor=black,linewidth=0.7pt]{->}(8,13)(9.7,13)
\rput[l](8.2,13.7){$u_{N/2-1}$}
\pscircle(10,13){0.3}
\rput[c](10,13){$+$}
\psline[linecolor=black,linewidth=0.7pt]{->}(10.3,13)(14,13)
\rput[c](13,13.3){$v_{N/2}$}
\psline[linecolor=black,linewidth=0.7pt]{->}(17,13)(19,13)
\rput[l](17.7,13.3){$y_{N/2}$}

\psline[linecolor=black,linewidth=0.7pt]{->}(3,9)(5,9)
\rput[c](3.3,9.3){$u_{N/2+1}$}
\psline[linecolor=black,linewidth=0.7pt]{->}(3,7)(5,7)
\rput[c](3.3,7.3){$u_{N/2+2}$}
\psline[linecolor=black,linewidth=0.7pt]{->}(3,2)(5,2)
\rput[c](3.3,2.3){$u_{N}$}
\rput[c](3.5,5){$\vdots$}
\rput[c](9,5){$\vdots$}
\rput[c](18.5,5){$\vdots$}

\psline[linecolor=black,linewidth=0.7pt]{->}(8,9)(14,9)
\rput[l](8.2,9.3){$u_{2}$}
\rput[rt](13.9,8.7){$v_{N/2+1}$}
\psline[linecolor=black,linewidth=0.7pt]{->}(8,7)(14,7)
\rput[l](8.2,7.3){$u_{4}$}
\rput[rt](13.7,6.7){$v_{N/2+2}$}
\psline[linecolor=black,linewidth=0.7pt]{->}(8,2)(14,2)
\rput[l](8.2,2.3){$u_{N}$}
\rput[r](13.7,2.3){$v_{N}$}

\psline[linecolor=black,linewidth=0.7pt]{->}(17,9)(19,9)
\rput[l](17.7,9.3){$y_{N/2+1}$}
\psline[linecolor=black,linewidth=0.7pt]{->}(17,7)(19,7)
\rput[l](17.7,7.3){$y_{N/2+2}$}
\psline[linecolor=black,linewidth=0.7pt]{->}(17,2)(19,2)
\rput[l](17.7,2.3){$y_{N}$}

\pscircle(10,2){0.1}
\pscircle(11,7){0.1}
\pscircle(12,9){0.1}
\psline[linecolor=black,linewidth=0.7pt]{->}(10,2)(10,12.7)
\psline[linecolor=black,linewidth=0.7pt]{->}(11,7)(11,17.7)
\psline[linecolor=black,linewidth=0.7pt]{->}(12,9)(12,19.7)
\end{pspicture}
\caption{An alternative realization of the recursive construction for $W_{N}$.}
\label{mtFolding}
\end{center}
\end{figure}

%% file: fig9.tex
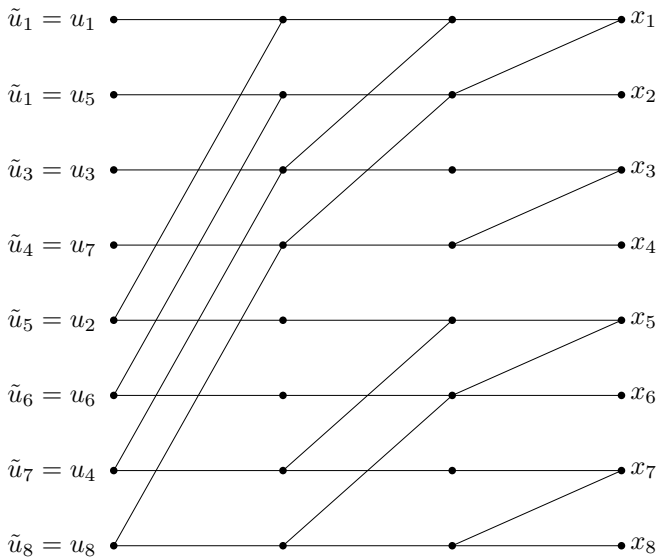
\begin{figure}[thb]
\begin{center}
\psset{arrowscale=1}
\psset{unit=0.5cm}
\psset{xunit=1.5,yunit=1}
\begin{pspicture}(0,0)(12,16)

\psline[linecolor=black,linewidth=0.3pt]{-}(2,1)(5,9)
\psline[linecolor=black,linewidth=0.3pt]{-}(2,5)(5,13)
\psline[linecolor=black,linewidth=0.3pt]{-}(2,3)(5,11)
\psline[linecolor=black,linewidth=0.3pt]{-}(2,7)(5,15)

\psline[linecolor=black,linewidth=0.3pt]{-}(2,1)(5,1)
\psline[linecolor=black,linewidth=0.3pt]{-}(2,3)(5,3)
\psline[linecolor=black,linewidth=0.3pt]{-}(2,5)(5,5)
\psline[linecolor=black,linewidth=0.3pt]{-}(2,7)(5,7)
\psline[linecolor=black,linewidth=0.3pt]{-}(2,9)(5,9)
\psline[linecolor=black,linewidth=0.3pt]{-}(2,11)(5,11)
\psline[linecolor=black,linewidth=0.3pt]{-}(2,13)(5,13)
\psline[linecolor=black,linewidth=0.3pt]{-}(2,15)(5,15)
\psline[linecolor=black,linewidth=0.3pt]{-}(5,1)(8,5)
\psline[linecolor=black,linewidth=0.3pt]{-}(5,3)(8,7)
\psline[linecolor=black,linewidth=0.3pt]{-}(5,9)(8,13)
\psline[linecolor=black,linewidth=0.3pt]{-}(5,11)(8,15)

\psline[linecolor=black,linewidth=0.3pt]{-}(5,1)(8,1)
\psline[linecolor=black,linewidth=0.3pt]{-}(5,3)(8,3)
\psline[linecolor=black,linewidth=0.3pt]{-}(5,5)(8,5)
\psline[linecolor=black,linewidth=0.3pt]{-}(5,7)(8,7)
\psline[linecolor=black,linewidth=0.3pt]{-}(5,9)(8,9)
\psline[linecolor=black,linewidth=0.3pt]{-}(5,11)(8,11)
\psline[linecolor=black,linewidth=0.3pt]{-}(5,13)(8,13)
\psline[linecolor=black,linewidth=0.3pt]{-}(5,15)(8,15)
\psline[linecolor=black,linewidth=0.3pt]{-}(8,1)(11,3)
\psline[linecolor=black,linewidth=0.3pt]{-}(8,5)(11,7)
\psline[linecolor=black,linewidth=0.3pt]{-}(8,9)(11,11)
\psline[linecolor=black,linewidth=0.3pt]{-}(8,13)(11,15)

\psline[linecolor=black,linewidth=0.3pt]{-}(8,1)(11,1)
\psline[linecolor=black,linewidth=0.3pt]{-}(8,3)(11,3)
\psline[linecolor=black,linewidth=0.3pt]{-}(8,5)(11,5)
\psline[linecolor=black,linewidth=0.3pt]{-}(8,7)(11,7)
\psline[linecolor=black,linewidth=0.3pt]{-}(8,9)(11,9)
\psline[linecolor=black,linewidth=0.3pt]{-}(8,11)(11,11)
\psline[linecolor=black,linewidth=0.3pt]{-}(8,13)(11,13)
\psline[linecolor=black,linewidth=0.3pt]{-}(8,15)(11,15)

\pscircle*(2,1){0.1}\pscircle*(5,1){0.1}\pscircle*(8,1){0.1}\pscircle*(11,1){0.1}
\pscircle*(2,3){0.1}\pscircle*(5,3){0.1}\pscircle*(8,3){0.1}
\pscircle[fillstyle=solid,fillcolor=black](11,3){0.1}
\pscircle*(2,5){0.1}\pscircle*(5,5){0.1}
\pscircle[fillstyle=solid,fillcolor=black](8,5){0.1}
\pscircle*(11,5){0.1}
\pscircle*(2,7){0.1}\pscircle*(5,7){0.1}
\pscircle[fillstyle=solid,fillcolor=black](8,7){0.1}
\pscircle[fillstyle=solid,fillcolor=black](11,7){0.1}
\pscircle*(2,9){0.1}
\pscircle[fillstyle=solid,fillcolor=black](5,9){0.1}
\pscircle*(8,9){0.1}\pscircle*(11,9){0.1}
\pscircle*(2,11){0.1}
\pscircle[fillstyle=solid,fillcolor=black](5,11){0.1}
\pscircle*(8,11){0.1}
\pscircle[fillstyle=solid,fillcolor=black](11,11){0.1}
\pscircle*(2,13){0.1}
\pscircle[fillstyle=solid,fillcolor=black](5,13){0.1}
\pscircle[fillstyle=solid,fillcolor=black](8,13){0.1}
\pscircle*(11,13){0.1}
\pscircle*(2,15){0.1}
\pscircle[fillstyle=solid,fillcolor=black](5,15){0.1}
\pscircle[fillstyle=solid,fillcolor=black](8,15){0.1}
\pscircle[fillstyle=solid,fillcolor=black](11,15){0.1}

\rput[cl](11.4,1){$x_8$}
\rput[cl](11.4,3){$x_7$}
\rput[cl](11.4,5){$x_6$}
\rput[cl](11.4,7){$x_5$}
\rput[cl](11.4,9){$x_4$}
\rput[cl](11.4,11){$x_3$}
\rput[cl](11.4,13){$x_2$}
\rput[cl](11.4,15){$x_1$}

\rput[r](1.7,1){$\tilde{u}_8=u_8$}
\rput[r](1.7,3){$\tilde{u}_7=u_4$}
\rput[r](1.7,5){$\tilde{u}_6=u_6$}
\rput[r](1.7,7){$\tilde{u}_5=u_2$}
\rput[r](1.7,9){$\tilde{u}_4=u_7$}
\rput[r](1.7,11){$\tilde{u}_3=u_3$}
\rput[r](1.7,13){$\tilde{u}_1=u_5$}
\rput[r](1.7,15){$\tilde{u}_1=u_1$}
\end{pspicture}
\caption{A circuit for implementing the transformation $F^{\otimes 3}$. 
Signals flow from left to right. Each edge carries a signal 0 or 1. Each node adds (mod-2) the signals on all incoming edges from the left and sends the result out on all edges to the right. (Edges carrying the signals $u_i$ and $x_i$ are not shown.)}
\label{fig:transform}
\end{center}
\end{figure}

%% file: fig10.tex
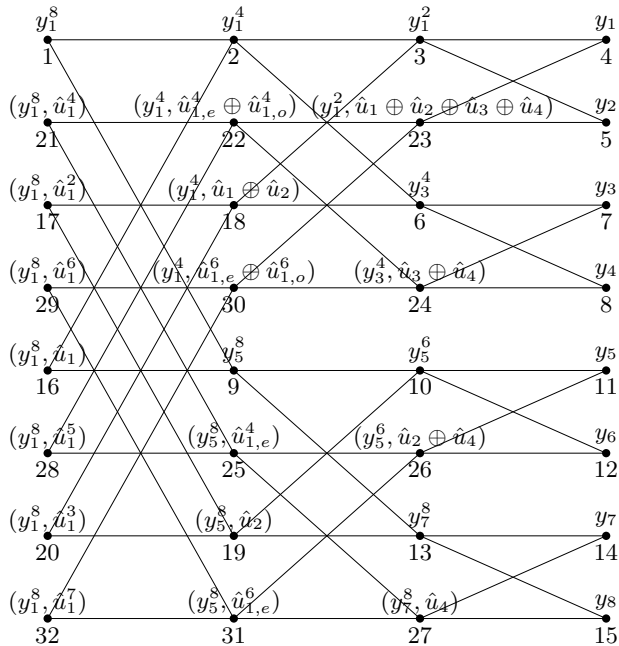
\begin{figure}[thb]
\begin{center}
\psset{arrowscale=1}
\psset{unit=0.55cm}
\psset{xunit=1.5,yunit=1}
\begin{pspicture}(2,0)(11,16)

\psline[linecolor=black,linewidth=0.3pt]{-}(2,1)(5,9)
\psline[linecolor=black,linewidth=0.3pt]{-}(2,9)(5,1)
\psline[linecolor=black,linewidth=0.3pt]{-}(2,5)(5,13)
\psline[linecolor=black,linewidth=0.3pt]{-}(2,13)(5,5)
\psline[linecolor=black,linewidth=0.3pt]{-}(2,3)(5,11)
\psline[linecolor=black,linewidth=0.3pt]{-}(2,11)(5,3)
\psline[linecolor=black,linewidth=0.3pt]{-}(2,7)(5,15)
\psline[linecolor=black,linewidth=0.3pt]{-}(2,15)(5,7)

\psline[linecolor=black,linewidth=0.3pt]{-}(2,1)(5,1)
\psline[linecolor=black,linewidth=0.3pt]{-}(2,3)(5,3)
\psline[linecolor=black,linewidth=0.3pt]{-}(2,5)(5,5)
\psline[linecolor=black,linewidth=0.3pt]{-}(2,7)(5,7)
\psline[linecolor=black,linewidth=0.3pt]{-}(2,9)(5,9)
\psline[linecolor=black,linewidth=0.3pt]{-}(2,11)(5,11)
\psline[linecolor=black,linewidth=0.3pt]{-}(2,13)(5,13)
\psline[linecolor=black,linewidth=0.3pt]{-}(2,15)(5,15)
\psline[linecolor=black,linewidth=0.3pt]{-}(5,1)(8,5)
\psline[linecolor=black,linewidth=0.3pt]{-}(5,5)(8,1)
\psline[linecolor=black,linewidth=0.3pt]{-}(5,3)(8,7)
\psline[linecolor=black,linewidth=0.3pt]{-}(5,7)(8,3)
\psline[linecolor=black,linewidth=0.3pt]{-}(5,9)(8,13)
\psline[linecolor=black,linewidth=0.3pt]{-}(5,13)(8,9)
\psline[linecolor=black,linewidth=0.3pt]{-}(5,11)(8,15)
\psline[linecolor=black,linewidth=0.3pt]{-}(5,15)(8,11)

\psline[linecolor=black,linewidth=0.3pt]{-}(5,1)(8,1)
\psline[linecolor=black,linewidth=0.3pt]{-}(5,3)(8,3)
\psline[linecolor=black,linewidth=0.3pt]{-}(5,5)(8,5)
\psline[linecolor=black,linewidth=0.3pt]{-}(5,7)(8,7)
\psline[linecolor=black,linewidth=0.3pt]{-}(5,9)(8,9)
\psline[linecolor=black,linewidth=0.3pt]{-}(5,11)(8,11)
\psline[linecolor=black,linewidth=0.3pt]{-}(5,13)(8,13)
\psline[linecolor=black,linewidth=0.3pt]{-}(5,15)(8,15)
\psline[linecolor=black,linewidth=0.3pt]{-}(8,1)(11,3)
\psline[linecolor=black,linewidth=0.3pt]{-}(8,3)(11,1)
\psline[linecolor=black,linewidth=0.3pt]{-}(8,5)(11,7)
\psline[linecolor=black,linewidth=0.3pt]{-}(8,7)(11,5)
\psline[linecolor=black,linewidth=0.3pt]{-}(8,9)(11,11)
\psline[linecolor=black,linewidth=0.3pt]{-}(8,11)(11,9)
\psline[linecolor=black,linewidth=0.3pt]{-}(8,13)(11,15)
\psline[linecolor=black,linewidth=0.3pt]{-}(8,15)(11,13)

\psline[linecolor=black,linewidth=0.3pt]{-}(8,1)(11,1)
\psline[linecolor=black,linewidth=0.3pt]{-}(8,3)(11,3)
\psline[linecolor=black,linewidth=0.3pt]{-}(8,5)(11,5)
\psline[linecolor=black,linewidth=0.3pt]{-}(8,7)(11,7)
\psline[linecolor=black,linewidth=0.3pt]{-}(8,9)(11,9)
\psline[linecolor=black,linewidth=0.3pt]{-}(8,11)(11,11)
\psline[linecolor=black,linewidth=0.3pt]{-}(8,13)(11,13)
\psline[linecolor=black,linewidth=0.3pt]{-}(8,15)(11,15)

\pscircle*(2,1){0.1}\pscircle*(5,1){0.1}\pscircle*(8,1){0.1}\pscircle*(11,1){0.1}
\pscircle*(2,3){0.1}\pscircle*(5,3){0.1}\pscircle*(8,3){0.1}\pscircle*(11,3){0.1}
\pscircle*(2,5){0.1}\pscircle*(5,5){0.1}\pscircle*(8,5){0.1}\pscircle*(11,5){0.1}
\pscircle*(2,7){0.1}\pscircle*(5,7){0.1}\pscircle*(8,7){0.1}\pscircle*(11,7){0.1}
\pscircle*(2,9){0.1}\pscircle*(5,9){0.1}\pscircle*(8,9){0.1}\pscircle*(11,9){0.1}
\pscircle*(2,11){0.1}\pscircle*(5,11){0.1}\pscircle*(8,11){0.1}\pscircle*(11,11){0.1}
\pscircle*(2,13){0.1}\pscircle*(5,13){0.1}\pscircle*(8,13){0.1}\pscircle*(11,13){0.1}
\pscircle*(2,15){0.1}\pscircle*(5,15){0.1}\pscircle*(8,15){0.1}\pscircle*(11,15){0.1}

\rput[bc](11,1.4){\small $y_8$}
\rput[bc](11,3.4){\small $y_7$}
\rput[bc](11,5.4){\small $y_6$}
\rput[bc](11,7.4){\small $y_5$}
\rput[bc](11,9.4){\small $y_4$}
\rput[bc](11,11.4){\small $y_3$}
\rput[bc](11,13.4){\small $y_2$}
\rput[bc](11,15.4){\small $y_1$}

\rput[bc](8,1.4){\small $(y_7^8,\hat{u}_4)$}
\rput[bc](8,3.5){\small $y_7^8$}
\rput[bc](8,5.4){\small $(y_5^6,\hat{u}_2\oplus \hat{u}_4)$}
\rput[bc](8,7.5){\small $y_5^6$}
\rput[bc](8,9.4){\small $(y_3^4,\hat{u}_3\oplus \hat{u}_4)$}
\rput[bc](8,11.5){\small $y_3^4$}
\rput[bc](8.2,13.4){\small $(y_1^2,\hat{u}_1\oplus \hat{u}_2\oplus \hat{u}_3\oplus \hat{u}_4)$}
\rput[bc](8,15.5){\small $y_1^2$}

\rput[bc](5,1.4){\small $(y_5^8,\hat{u}_{1,e}^6)$}
\rput[bc](5,3.4){\small $(y_5^8,\hat{u}_2)$}
\rput[bc](5,5.4){\small $(y_5^8,\hat{u}_{1,e}^4)$}
\rput[bc](5,7.5){\small $y_5^8$}
\rput[bc](5,9.4){\small $(y_1^4,\hat{u}_{1,e}^6\oplus \hat{u}_{1,o}^6)$}
\rput[bc](5,11.4){\small $(y_1^4,\hat{u}_1\oplus \hat{u}_2)$}
\rput[bc](4.7,13.4){\small $(y_1^4,\hat{u}_{1,e}^4\oplus \hat{u}_{1,o}^4)$}
\rput[bc](5,15.5){\small $y_1^4$}

\rput[bc](2,1.4){\small $(y_1^8,\hat{u}_1^7)$}
\rput[bc](2,3.4){\small $(y_1^8,\hat{u}_1^3)$}
\rput[bc](2,5.4){\small $(y_1^8,\hat{u}_1^5)$}
\rput[bc](2,7.4){\small $(y_1^8,\hat{u}_1)$}
\rput[bc](2,9.4){\small $(y_1^8,\hat{u}_1^6)$}
\rput[bc](2,11.4){\small $(y_1^8,\hat{u}_1^2)$}
\rput[bc](2,13.4){\small $(y_1^8,\hat{u}_1^4)$}
\rput[bc](2,15.5){\small $y_1^8$}

\rput[tc](11,0.6){\small $15$}
\rput[tc](11,2.6){\small $14$}
\rput[tc](11,4.6){\small $12$}
\rput[tc](11,6.6){\small $11$}
\rput[tc](11,8.6){\small $8$}
\rput[tc](11,10.6){\small $7$}
\rput[tc](11,12.6){\small $5$}
\rput[tc](11,14.6){\small $4$}
\rput[tc](8,0.6){\small $27$}
\rput[tc](8,2.6){\small $13$}
\rput[tc](8,4.6){\small $26$}
\rput[tc](8,6.6){\small $10$}
\rput[tc](8,8.6){\small $24$}
\rput[tc](8,10.6){\small $6$}
\rput[tc](8,12.6){\small $23$}
\rput[tc](8,14.6){\small $3$}
\rput[tc](5,0.6){\small $31$}
\rput[tc](5,2.6){\small $19$}
\rput[tc](5,4.6){\small $25$}
\rput[tc](5,6.6){\small $9$}
\rput[tc](5,8.6){\small $30$}
\rput[tc](5,10.6){\small $18$}
\rput[tc](5,12.6){\small $22$}
\rput[tc](5,14.6){\small $2$}
\rput[tc](2,0.6){\small $32$}
\rput[tc](2,2.6){\small $20$}
\rput[tc](2,4.6){\small $28$}
\rput[tc](2,6.6){\small $16$}
\rput[tc](2,8.6){\small $29$}
\rput[tc](2,10.6){\small $17$}
\rput[tc](2,12.6){\small $21$}
\rput[tc](2,14.6){\small $1$}

\end{pspicture}
\caption{An implementation of the successive cancellation decoder for polar coding at block-length $N=8$.} 
\label{fig:Decoder}
\end{center}
\end{figure}

%% file: fig11.tex
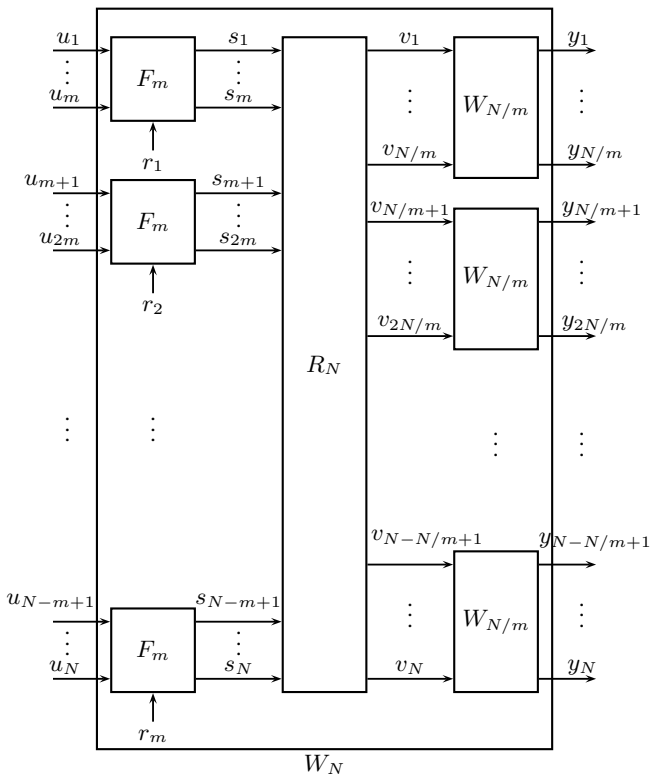
\begin{figure}[thb]
\begin{center}
\psset{arrowscale=1}
\psset{unit=0.38cm}
\psset{xunit=1,yunit=1}
\begin{small}
\begin{pspicture}(4,-1)(22,26)
\psframe(4.5,0)(20.5,26)
\psframe(5,2)(8,5)
\rput[c](6.5,3.5){$F_m$}
\psframe(5,22)(8,25)
\rput[c](6.5,23.5){$F_m$}
\psframe(5,17)(8,20)
\rput[c](6.5,18.5){$F_m$}

\psframe(11,2)(14,25)
\rput[c](12.5,13.5){$R_{N}$}

\psframe(17,2)(20,7)
\rput[c](18.5,4.5){$W_{N/m}$}
\psframe(17,20)(20,25)
\rput[c](18.5,22.5){$W_{N/m}$}
\psframe(17,14)(20,19)
\rput[c](18.5,16.5){$W_{N/m}$}

\psline[linecolor=black,linewidth=0.7pt]{->}(3,2.5)(5,2.5)
\rput[r](4,2.9){$u_{N}$}
\psline[linecolor=black,linewidth=0.7pt]{->}(3,4.5)(5,4.5)
\rput[r](4.5,5.2){$u_{N-m+1}$}
\rput[c](3.5,4){$\vdots$}
\psline[linecolor=black,linewidth=0.7pt]{->}(8,2.5)(11,2.5)
\rput[c](9.5,2.9){$s_{N}$}
\psline[linecolor=black,linewidth=0.7pt]{->}(8,4.5)(11,4.5)
\rput[c](9.5,5.2){$s_{N-m+1}$}
\rput[c](9.5,4){$\vdots$}
\psline[linecolor=black,linewidth=0.7pt]{->}(6.5,1)(6.5,2)
\rput[c](6.5,0.5){$r_{m}$}

\psline[linecolor=black,linewidth=0.7pt]{->}(3,17.5)(5,17.5)
\rput[r](4,17.9){$u_{2m}$}
\psline[linecolor=black,linewidth=0.7pt]{->}(3,19.5)(5,19.5)
\rput[r](4,19.9){$u_{m+1}$}
\rput[c](3.5,19){$\vdots$}
\psline[linecolor=black,linewidth=0.7pt]{->}(8,17.5)(11,17.5)
\rput[c](9.5,17.9){$s_{2m}$}
\psline[linecolor=black,linewidth=0.7pt]{->}(8,19.5)(11,19.5)
\rput[c](9.5,19.9){$s_{m+1}$}
\rput[c](9.5,19){$\vdots$}
\psline[linecolor=black,linewidth=0.7pt]{->}(6.5,16)(6.5,17)
\rput[c](6.5,15.5){$r_{2}$}

\psline[linecolor=black,linewidth=0.7pt]{->}(3,22.5)(5,22.5)
\rput[r](4,22.9){$u_{m}$}
\psline[linecolor=black,linewidth=0.7pt]{->}(3,24.5)(5,24.5)
\rput[r](4,24.9){$u_{1}$}
\rput[c](3.5,24){$\vdots$}
\psline[linecolor=black,linewidth=0.7pt]{->}(8,22.5)(11,22.5)
\rput[c](9.5,22.9){$s_{m}$}
\psline[linecolor=black,linewidth=0.7pt]{->}(8,24.5)(11,24.5)
\rput[c](9.5,24.9){$s_{1}$}
\rput[c](9.5,24){$\vdots$}
\psline[linecolor=black,linewidth=0.7pt]{->}(6.5,21)(6.5,22)
\rput[c](6.5,20.5){$r_{1}$}

\psline[linecolor=black,linewidth=0.7pt]{->}(14,2.5)(17,2.5)
\rput[c](15.5,2.9){$v_{N}$}
\psline[linecolor=black,linewidth=0.7pt]{->}(14,6.5)(17,6.5)
\rput[c](16.1,7.4){$v_{N-N/m+1}$}
\rput[c](15.5,5){$\vdots$}

\psline[linecolor=black,linewidth=0.7pt]{->}(14,14.5)(17,14.5)
\rput[c](15.5,14.9){$v_{2N/m}$}
\psline[linecolor=black,linewidth=0.7pt]{->}(14,18.5)(17,18.5)
\rput[c](15.5,18.9){$v_{N/m+1}$}
\rput[c](15.5,17){$\vdots$}

\psline[linecolor=black,linewidth=0.7pt]{->}(14,20.5)(17,20.5)
\rput[c](15.5,20.9){$v_{N/m}$}
\psline[linecolor=black,linewidth=0.7pt]{->}(14,24.5)(17,24.5)
\rput[c](15.5,24.9){$v_{1}$}
\rput[c](15.5,23){$\vdots$}

\psline[linecolor=black,linewidth=0.7pt]{->}(20,2.5)(22,2.5)
\rput[l](21,2.9){$y_{N}$}
\psline[linecolor=black,linewidth=0.7pt]{->}(20,6.5)(22,6.5)
\rput[l](20.0,7.3){$y_{N-N/m+1}$}
\rput[c](21.5,5){$\vdots$}

\psline[linecolor=black,linewidth=0.7pt]{->}(20,14.5)(22,14.5)
\rput[l](20.8,14.9){$y_{2N/m}$}
\psline[linecolor=black,linewidth=0.7pt]{->}(20,18.5)(22,18.5)
\rput[l](20.8,18.9){$y_{N/m+1}$}
\rput[c](21.5,17){$\vdots$}

\psline[linecolor=black,linewidth=0.7pt]{->}(20,20.5)(22,20.5)
\rput[l](21,20.9){$y_{N/m}$}
\psline[linecolor=black,linewidth=0.7pt]{->}(20,24.5)(22,24.5)
\rput[l](21,24.9){$y_{1}$}
\rput[c](21.5,23){$\vdots$}

\rput[c](6.5,11.5){$\vdots$}
\rput[c](3.5,11.5){$\vdots$}
\rput[c](18.5,11){$\vdots$}
\rput[c](21.5,11){$\vdots$}

\rput[c](12.5,-0.5){$W_{N}$}

\end{pspicture}
\end{small}
\caption{General form of channel combining.}
\label{fig:GenCombine}
\end{center}
\end{figure}

%% file: fig12.tex
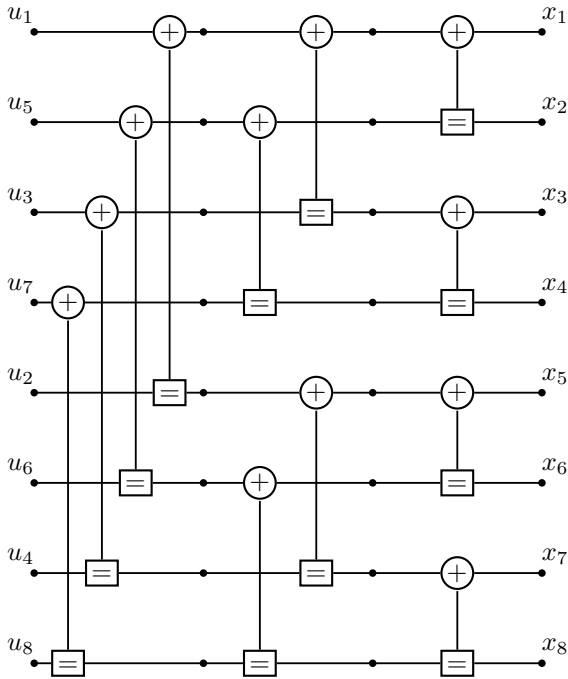
\begin{figure}[thb]
\begin{center}
\psset{arrowscale=1}
\psset{unit=0.5cm}
\psset{xunit=1.5,yunit=1.2}
\begin{pspicture}(0,0)(9,16)
\psline[linecolor=black,linewidth=0.7pt]{-}(7.5,1.3)(7.5,2.6)
\psline[linecolor=black,linewidth=0.7pt]{-}(7.5,5.3)(7.5,6.6)
\psline[linecolor=black,linewidth=0.7pt]{-}(7.5,9.3)(7.5,10.6)
\psline[linecolor=black,linewidth=0.7pt]{-}(7.5,13.3)(7.5,14.6)
\pscircle(7.5,3){0.45}\pscircle(7.5,7){0.45}\pscircle(7.5,11){0.45}\pscircle(7.5,15){0.45}
\psframe(7.2,0.7)(7.8,1.3)\psframe(7.2,4.7)(7.8,5.3)\psframe(7.2,8.7)(7.8,9.3)\psframe(7.2,12.7)(7.8,13.3)
\psline[linecolor=black,linewidth=0.7pt]{-}(6,1)(7.2,1)
\psline[linecolor=black,linewidth=0.7pt]{-}(7.8,1)(9,1)
\rput[c](7.5,0.95){$=$}
\psline[linecolor=black,linewidth=0.7pt]{-}(6,3)(7.2,3)
\psline[linecolor=black,linewidth=0.7pt]{-}(7.8,3)(9,3)
\rput[c](7.5,3){$+$}
\psline[linecolor=black,linewidth=0.7pt]{-}(6,5)(7.2,5)
\psline[linecolor=black,linewidth=0.7pt]{-}(7.8,5)(9,5)
\rput[c](7.5,4.95){$=$}
\psline[linecolor=black,linewidth=0.7pt]{-}(6,7)(7.2,7)
\psline[linecolor=black,linewidth=0.7pt]{-}(7.8,7)(9,7)
\rput[c](7.5,7){$+$}
\psline[linecolor=black,linewidth=0.7pt]{-}(6,9)(7.2,9)
\psline[linecolor=black,linewidth=0.7pt]{-}(7.8,9)(9,9)
\rput[c](7.5,8.95){$=$}
\psline[linecolor=black,linewidth=0.7pt]{-}(6,11)(7.2,11)
\psline[linecolor=black,linewidth=0.7pt]{-}(7.8,11)(9,11)
\rput[c](7.5,11){$+$}
\psline[linecolor=black,linewidth=0.7pt]{-}(6,13)(7.2,13)
\psline[linecolor=black,linewidth=0.7pt]{-}(7.8,13)(9,13)
\rput[c](7.5,12.95){$=$}
\psline[linecolor=black,linewidth=0.7pt]{-}(6,15)(7.2,15)
\psline[linecolor=black,linewidth=0.7pt]{-}(7.8,15)(9,15)
\rput[c](7.5,15){$+$}

\psline[linecolor=black,linewidth=0.7pt]{-}(4,1.3)(4,4.6)
\psline[linecolor=black,linewidth=0.7pt]{-}(5,3.3)(5,6.6)
\psline[linecolor=black,linewidth=0.7pt]{-}(4,9.3)(4,12.6)
\psline[linecolor=black,linewidth=0.7pt]{-}(5,11.3)(5,14.6)
\pscircle(4,5){0.45}\pscircle(5,7){0.45}\pscircle(4,13){0.45}\pscircle(5,15){0.45}
\rput[c](4,5){$+$}\rput[c](5,7){$+$}\rput[c](4,13){$+$}\rput[c](5,15){$+$}
\psframe(3.7,0.7)(4.3,1.3)\psframe(4.7,2.7)(5.3,3.3)\psframe(3.7,8.7)(4.3,9.3)\psframe(4.7,10.7)(5.3,11.3)
\rput[c](4,0.95){$=$}\rput[c](5,2.95){$=$}\rput[c](4,8.95){$=$}\rput[c](5,10.95){$=$}
\psline[linecolor=black,linewidth=0.7pt]{-}(3,1)(3.7,1)
\psline[linecolor=black,linewidth=0.7pt]{-}(4.3,1)(6,1)
\psline[linecolor=black,linewidth=0.7pt]{-}(3,3)(4.7,3)
\psline[linecolor=black,linewidth=0.7pt]{-}(5.3,3)(6,3)
\psline[linecolor=black,linewidth=0.7pt]{-}(3,5)(3.7,5)
\psline[linecolor=black,linewidth=0.7pt]{-}(4.3,5)(6,5)
\psline[linecolor=black,linewidth=0.7pt]{-}(3,7)(4.7,7)
\psline[linecolor=black,linewidth=0.7pt]{-}(5.3,7)(6,7)
\psline[linecolor=black,linewidth=0.7pt]{-}(3,9)(3.7,9)
\psline[linecolor=black,linewidth=0.7pt]{-}(4.3,9)(6,9)
\psline[linecolor=black,linewidth=0.7pt]{-}(3,11)(4.7,11)
\psline[linecolor=black,linewidth=0.7pt]{-}(5.3,11)(6,11)
\psline[linecolor=black,linewidth=0.7pt]{-}(3,13)(3.7,13)
\psline[linecolor=black,linewidth=0.7pt]{-}(4.3,13)(6,13)
\psline[linecolor=black,linewidth=0.7pt]{-}(3,15)(4.7,15)
\psline[linecolor=black,linewidth=0.7pt]{-}(5.3,15)(6,15)

\psline[linecolor=black,linewidth=0.7pt]{-}(0.6,1.3)(0.6,8.6)
\psline[linecolor=black,linewidth=0.7pt]{-}(1.2,3.3)(1.2,10.6)
\psline[linecolor=black,linewidth=0.7pt]{-}(1.8,5.3)(1.8,12.6)
\psline[linecolor=black,linewidth=0.7pt]{-}(2.4,7.3)(2.4,14.6)

\pscircle(0.6,9){0.45}\pscircle(1.2,11){0.45}\pscircle(1.8,13){0.45}\pscircle(2.4,15){0.45}
\rput[c](0.6,9){$+$}\rput[c](1.2,11){$+$}\rput[c](1.8,13){$+$}\rput[c](2.4,15){$+$}
\psframe(0.3,0.7)(0.9,1.3)\psframe(0.9,2.7)(1.5,3.3)\psframe(1.5,4.7)(2.1,5.3)\psframe(2.1,6.7)(2.7,7.3)
\rput[c](0.6,0.95){$=$}\rput[c](1.2,2.95){$=$}\rput[c](1.8,4.95){$=$}\rput[c](2.4,6.95){$=$}

\psline[linecolor=black,linewidth=0.7pt]{-}(0,1)(0.3,1)\psline[linecolor=black,linewidth=0.7pt]{-}(0.9,1)(3,1)
\psline[linecolor=black,linewidth=0.7pt]{-}(0,3)(0.9,3)\psline[linecolor=black,linewidth=0.7pt]{-}(1.5,3)(3,3)
\psline[linecolor=black,linewidth=0.7pt]{-}(0,5)(1.5,5)\psline[linecolor=black,linewidth=0.7pt]{-}(2.1,5)(3,5)
\psline[linecolor=black,linewidth=0.7pt]{-}(0,7)(2.1,7)\psline[linecolor=black,linewidth=0.7pt]{-}(2.7,7)(3,7)
\psline[linecolor=black,linewidth=0.7pt]{-}(0,9)(0.3,9)\psline[linecolor=black,linewidth=0.7pt]{-}(0.9,9)(3,9)
\psline[linecolor=black,linewidth=0.7pt]{-}(0,11)(0.9,11)\psline[linecolor=black,linewidth=0.7pt]{-}(1.5,11)(3,11)
\psline[linecolor=black,linewidth=0.7pt]{-}(0,13)(1.5,13)\psline[linecolor=black,linewidth=0.7pt]{-}(2.1,13)(3,13)
\psline[linecolor=black,linewidth=0.7pt]{-}(0,15)(2.1,15)\psline[linecolor=black,linewidth=0.7pt]{-}(2.7,15)(3,15)
\pscircle*(0,1){0.1}\pscircle*(3,1){0.1}\pscircle*(6,1){0.1}\pscircle*(9,1){0.1}
\pscircle*(0,3){0.1}\pscircle*(3,3){0.1}\pscircle*(6,3){0.1}\pscircle*(9,3){0.1}
\pscircle*(0,5){0.1}\pscircle*(3,5){0.1}\pscircle*(6,5){0.1}\pscircle*(9,5){0.1}
\pscircle*(0,7){0.1}\pscircle*(3,7){0.1}\pscircle*(6,7){0.1}\pscircle*(9,7){0.1}
\pscircle*(0,9){0.1}\pscircle*(3,9){0.1}\pscircle*(6,9){0.1}\pscircle*(9,9){0.1}
\pscircle*(0,11){0.1}\pscircle*(3,11){0.1}\pscircle*(6,11){0.1}\pscircle*(9,11){0.1}
\pscircle*(0,13){0.1}\pscircle*(3,13){0.1}\pscircle*(6,13){0.1}\pscircle*(9,13){0.1}
\pscircle*(0,15){0.1}\pscircle*(3,15){0.1}\pscircle*(6,15){0.1}\pscircle*(9,15){0.1}
\rput[br](0,1.2){$u_8$}
\rput[br](0,3.2){$u_4$}
\rput[br](0,5.2){$u_6$}
\rput[br](0,7.2){$u_2$}
\rput[br](0,9.2){$u_7$}
\rput[br](0,11.2){$u_3$}
\rput[br](0,13.2){$u_5$}
\rput[br](0,15.2){$u_1$}
\rput[bl](9,1.2){$x_8$}
\rput[bl](9,3.2){$x_7$}
\rput[bl](9,5.2){$x_6$}
\rput[bl](9,7.2){$x_5$}
\rput[bl](9,9.2){$x_4$}
\rput[bl](9,11.2){$x_3$}
\rput[bl](9,13.2){$x_2$}
\rput[bl](9,15.2){$x_1$}

\end{pspicture}
\normalsize
\caption{The factor graph representation for the transformation $F^{\otimes 3}$.}
\label{fig:factorgraph}
\end{center}
\end{figure}